\documentclass[12pt,reqno]{amsart}

\usepackage[utf8]{inputenc}
\usepackage{amsfonts,latexsym,amssymb,amsmath,amsthm,amsrefs,etex,slashed,cancel}
\usepackage{todonotes}
\usepackage{hyperref}
\usepackage{a4wide}
\usepackage{enumerate}
\usepackage{mathrsfs}
\usepackage{MnSymbol}
\usepackage{nicefrac}
\usepackage{tikz}

\usepackage{latexsym}
\usepackage{amssymb}
\usepackage[active]{srcltx}
\newcommand{\dd}{\mathrm{d}}
\newcommand{\ee}{\mathrm{e}}
\newcommand{\rmi}{\mathrm{i}}
\usepackage{color}

\newcommand{\ord}{\mathrm{ord}}

\newcommand{\edit}[1]{{\color{red}{$\clubsuit$#1$\clubsuit$}}}

\usepackage{epsfig}

\newcommand{\wt}{\widetilde}

\newcommand{\R}{{\mathbb R}}

\newcommand{\Z}{{\mathbb Z}}

\newcommand{\Tr}{{\operatorname{Tr}\,}}

\newcommand{\half}{{\frac{1}{2}}}

\newcommand{\supp}{{\operatorname{supp\,}}}
\renewcommand{\phi}{\varphi}

\newcommand{\acal}{\mathcal{A}}

\newcommand{\ccal}{\mathcal{C}}
\newcommand{\dcal}{\mathcal{D}}
\newcommand{\ecal}{\mathcal{E}}
\newcommand{\fcal}{\mathcal{F}}

\newcommand{\hcal}{\mathcal{H}}
\newcommand{\ical}{\mathcal{I}}

\newcommand{\kcal}{\mathcal{K}}
\newcommand{\lcal}{\mathcal{L}}
\newcommand{\pcal}{\mathcal{P}}

\newcommand{\ncal}{\mathcal{N}}

\newcommand{\rcal}{\mathcal{R}}
\newcommand{\scal}{\mathcal{S}}

 \newtheorem{theo}{Theorem}[section]
 \newtheorem{lem}[theo]{Lemma}
 \newtheorem{cor}[theo]{Corollary}
 \newtheorem{prop}[theo]{Proposition}

 \theoremstyle{definition}
 \newtheorem{defin}[theo]{Definition}
 
 \newtheorem{rem}[theo]{Remark}

\title{
Semi-classical mass  asymptotics  on stationary spacetimes }

\author{Alexander Strohmaier}
\author{Steve Zelditch}

\address{School of Mathematics \\
University of Leeds\\
Leeds, LS2 9JT, UK}
\email{a.strohmaier@leeds.ac.uk}

\address{Department of Mathematics, Northwestern  University,
Evanston, IL 60208-2370, USA} \email{
zelditch@math.northwestern.edu}

\thanks{Research partially supported by NSF grant   DMS-1810747}

\begin{document}
\maketitle

\begin{abstract}  We study the spectrum of a timelike Killing vector field $Z$ acting as a differential operator on the solution space
$\hcal_m: = \{u \mid (\Box_g + m^2) u = 0\}$ of the Klein-Gordon equation on a globally hyperbolic stationary spacetime $(M,g)$ with compact Cauchy
hypersurface $\Sigma$. We endow $\hcal_m$ with a natural inner product, so that $D_Z = \frac{1}{i} \nabla_Z$ is a self-adjoint operator
on $\hcal_m$ with discrete spectrum $\{\lambda_j(m)\}$. In  earlier work, we proved a Weyl law for the number of eigenvalues
$\lambda_j(m)$ in an interval for fixed mass $m$. In this sequel, we prove a Weyl law along `ladders' $\{(m, \lambda_j(m): m \in \R_+\}$
such that $\frac{\lambda_j(m)}{m} \simeq \nu$   as $m \to \infty$. More precisely, we given an asymptotic formula as $m \to \infty$ for the counting function 
 $N_{\nu, C}(m) : = \# \{j \mid  \frac{\lambda_j(m)}{m}  \in [\nu - \frac{C}{m}, \nu + \frac{C}{m} ]\}$ for $C > 0$. The asymptotics are determined from the dynamics of the Killing flow $e^{tZ} $ on the hypersurface $\ncal_{1, \nu} $ in the   space $\ncal_1$ of mass $1$ geodesics $\gamma$ where
 $\langle \dot{\gamma}, Z \rangle= \nu$. The method is to treat $m$ as a semi-classical parameter $h^{-1}$ and employ techniques of homogeneous quantization.
\end{abstract}
\bigskip

In a recent article \cite{SZ18}, the authors introduced a generalization of the Gutzwiller-Duister\-maat-Guillemin trace formula and Weyl law
on a globally hyperbolic, stationary spacetime $(M, g)$ of dimension $n$   with compact Cauchy hypersurface $\Sigma$.   The trace formula was for
the trace $\Tr e^{t Z} |_{\ker \Box_g} $ of the translation operator $e^{t Z}$ by the flow of the timelike Killing vector field $Z$ on the space $\ker \Box_g$ of null-solutions
of the Klein-Gordon operator $\Box_g$. To define the trace, a Hilbert space structure is introduced on $\ker \Box_g$, and the
trace is  expressed as an integral over $\Sigma$ of  a spacelike $(n-1)$-form constructed  from the solution operator $E$ of the Cauchy problem.   
One of the principal ingredients was the description
of $E$  as a Fourier integral operator in \cite[Theorem 6.5.3 and Theorem 6.6.1]{DH72}.  A second  principal ingredient was the  calculation by Duistermaat-Guillemin \cite{DG75} 
  of the trace $\Tr \exp \rmi t \sqrt{-\Delta}$ of the wave group of a compact
 Riemannian manifold $(\Sigma, h)$ using the symbol calculus of Fourier integral operators.  The authors of \cite{DH72}
 write  that one motivation for their article was to clarify the nature of `propagators' in relativistic quantum mechanics,  and their article has since
 become foundational  for quantum field theory on a curved spacetime. The relativistic trace formula of  \cite{SZ18} grew naturally from the   combination of 
 results of \cite{DH72, DG75}, and suggests that much of spectral asymptotics in non-relativistic quantum mechanics has 
a  generalizations to  globally hyperbolic, stationary spacetimes.

The present article is a continuation of \cite{SZ18}, still in the setting of a globally hyperbolic, stationary spacetime $(M, g)$
(time-oriented and oriented)  { of dimension $n \geq 2$} with 
compact Cauchy hypersurface $\Sigma$. In \cite{SZ18}, the spectral problem was that of $D_Z = \frac{1}{\rmi} \nabla_Z$ in $\ker \Box_g$, or more generally
the kernel of $\Box_g + V$ where $V$ is a $D_Z$-invariant potential $V$.  In this article, we fix $V = m^2$ (the mass squared), and study the spectrum of $D_Z$ in the kernel  \begin{equation} \label{hcalmdef} \hcal^{(m)}: = \ker (\Box_g + m^2), \end{equation}  of  the  Klein-Gordon operator $$ \Box_g = -\frac{1}{\sqrt{|g|}} \partial_i \left(\sqrt{|g|} g^{ik}\partial_k \right).$$ In Section \ref{HSSECT}, we equip $\hcal^{(m)}$ with an energy inner product (Definition \ref{ENERGYIP})  so that \eqref{hcalmdef} becomes a Hilbert space when $m \not= 0$.

 The main idea of this article is to treat
$m$ as a semi-classical parameter, so that in effect we study the semi-classical Klein-Gordon operator
$m^{-2} \Box_g + 1$ in the semi-classical limit $m \to \infty$.  Although it is rather obvious that $m^{-1}$ is formally analogous  to the Planck constant
$\hbar$ in a non-relativistic Schr\"odinger operator, we have not found prior studies of  semi-classical mass asymptotics for the Klein - Gordon operator. Equivalently, since 
 $[\Box_g, D_Z ] = 0$, we are    studying the joint spectrum  $\{(m^2, \lambda_j(m)) \}$ of $(\Box_g, D_Z)$, i.e.
\begin{equation} \label{BOXZ} \left\{ \begin{array}{l} (\Box_g + m^2)u = 0, \\ \\
    D_Z u =  \lambda u  \end{array} \right. \end{equation} in the spirit of equivariant  spectral asymptotics. Since $D_Z$ is first order, we reparametrize the joint spectrum as $\{(m, \lambda_j(m)), m \in \R\}$, and  study 
 the parametrized eigenvalues $m \to \lambda_j(m)$ as $m \to \infty, \lambda_j(m) \to \infty$ in such a way that $\frac{\lambda_j(m)}{m}$
tends to a limit value $\nu \in \R$. Following \cite{GS82}, we  refer to such pairs as a `ray' or `ladder' $L_{\nu}$  in the joint spectrum. As is pointed out in 
\cite{GU89},   ladder asymptotics can be used to obtain  semi-classical asymptotics, and that is our purpose in this article. 

The basic goal is
to determine the asymptotics as $m \to \infty$ of the  mass-semi-classical  Weyl eigenvalue counting function,
\begin{equation} \label{NnuC} N_{\nu, C}(m) : = \# \{j\mid  \frac{\lambda_j(m)}{m}  \in [\nu - \frac{C}{m}, \nu + \frac{C}{m} ]\},  \end{equation}
where $C > 0$ is a given constant.  Theorem \ref{WEYLCOR}  gives asymptotics as $m \to \infty$ of these Weyl spectral functions. 
To study this sharp counting function, we first smooth out the indicator function 
${\bf 1}_{[-C, C]}$ to a Schwartz function $\psi \in \scal(\R)$ with $\hat{\psi} \in C_0^{\infty}(\R)$ and form, 
\begin{equation} \label{Nnupsi} N_{\nu, \psi} (m) : =  \sum_{j \in \Z} \psi(\lambda_j (m) - \nu m). \end{equation}
The asymptotics of this smoothed Weyl function are determined in Theorem  \ref{LADDERCOR}, and the asymptotics of \eqref{NnuC} are 
determined in Theorem \ref{WEYLCOR}.

 We view the eigenvalues of $ D_Z$ on $\hcal_m$ as `energy levels' of a mass
$m$ free particle. In special relativity the energy of particle is the time component of its 4-momentum.
Here we view the classical energy $E$ of a mass  $m$ geodesic $\gamma$ as the constant value
$\langle Z, \dot{\gamma} \rangle$, which plays the role of the time component of its 4-momentum.
Mass semi-classical asymptotics  determines how  Klein-Gordon particles scale at large energies when
the  mass $m$ and energy $E$ are scaled in the same way. Having a fixed ratio $\frac{E}{m}$  intuitively means roughly that the velocity of the
particle is fixed while letting $m, E \to \infty$. This is explained in detail in Section \ref{PHYSINT}.

The study of the joint spectrum of $(\Box_g, D_Z)$ may seem reminiscent of the study
of the joint spectrum of the Laplacian $\Delta_h$ and a Killing vector field $Z$ on a compact
manifold $(\Sigma, h)$, e.g. the Laplacian on a surface of revolution. Among the many articles studying equivariant
spectral asymptotics and its applications to semi-classical asymptotics, the article \cite{GU89} is  most
relevant to our results and we refer there for references to the literature. The equations \eqref{BOXZ}
seem formally analogous to $(\Delta_h + \lambda^2) u = 0, D_Z = \lambda u$ in $L^2(\Sigma, dV_h)$,  but there
is a fundamental difference in the way that we treat the joint spectral problem:  as in  \cite{SZ18}, we do not study the joint spectrum on $L^2(M, dV)$, 
but rather endow \eqref{hcalmdef}  with its own Hilbertian inner product (Definition \ref{ENERGYIP}  in Section \ref{HSSECT}). Hence a `ladder'
in the joint spectrum involves the spectrum of $D_Z$ in the  one-parameter family \eqref{hcalmdef} of Hilbert spaces.  Since this type of joint spectral problem is non-standard even  in the case of
product (`ultra-static') spacetimes, we include a short Section \ref{PRODUCT} to explain what it
boils down to in the (rather trivial) product case.

 \subsection{\label{HOMOG} Homogenizing and discretizing  the $m$ parameter}
 
To deal with the one-parameter family \eqref{hcalmdef} of Hilbert spaces as $m \to \infty$, we construct the  globally hyperbolic stationary
spacetime $$(\wt{M}, \wt{g}) = (M, g)  \times S^1_T,$$
where $S^1_T = \R/ T \mathbb{Z}$, and where we take the semi-Riemannian product of the Lorentzian manifold
$(M, g)$ with $S^1_T$. The circle factor $S^1$ is spacelike, so that the metric $\wt{g} = g \oplus  ( \dd \theta^2)$  is Lorentzian of signature $(-1, +1, \dots, +1)$, and has the  compact Cauchy hypersurface   $\Sigma \times S^1_T$.    The group $\R \times S^1_T$ acts by isometries,  where $\R$ acts by 
the Killing flow $e^{t Z}$ and $S^1_T$ acts on the $S^1_T$ factor. 
The eigenvalues of $D_\theta $ are  $ m_k = \frac{2 \pi}{T} k$
with $k \in \Z$.  Since $T$ is arbitrary we essentially obtain all mass parameter
values, and since we are interested in ladders, little is lost by the discretization. \bigskip

\noindent{\bf Notational Conventions}
To avoid cluttering up the notation, we usually assume that $T = 2 \pi$ and write $S^1 $ for
$S^1 = S^1_{2 \pi}$. We often use the notation $C_n$ for any purely  dimensional constant that may vary in each occurrence.

When it could be confusing whether a point $\zeta$ lies in $T^* M$ or $T^* \wt M$, we use the notation $\zeta = (x, \xi)$
for the former and $\zeta = (\tilde x, \tilde \xi)$ for the latter. When there is no possibility of confusion, we drop the tilde's. In homogeneous quantization,
one often deletes the zero section of a cotangent bundle and works on $\dot{T}^*M=T^* M \backslash 0$. 
%for notational simplicity,
%we often omit the $\backslash 0$.

Our signature convention is $(-, +, \cdots, +)$ but $\Box_g$ is minus the Lorentzian Laplacian. The Minkowski
metric is thus $-dt^2 + \sum_j dx_j^2$ and its Klein-Gordon operator is $\frac{\partial^2}{\partial t^2} - \Delta + m^2. $ Note that $ - \Delta$ is a positive operator,
and that the mass term $m^2$ is positive, so that the semi-classical symbol $- \tau^2 + |\xi|^2 + m^2$ has zeros on the two-sheeted  mass hyperboloid (bundle)
$ \tau^2 = |\xi|^2 + m^2$.

Finally, we use the term {\it Liouville measure} for general  surface measures $\mu_L$ induced by the symplectic volume measure $\Omega$  on co-isotropic
submanifolds defined by some Poisson-commuting functions. In the case of a hypersurface $\{H = E\}$ in a symplectic manifold, the Liouville
measure is given by $\mu_L = \frac{\Omega}{\dd H} |_{H = E}. $ It satisfies, $\mu_L(\{H = E\} = \frac{\dd}{\dd E} \mathrm{Vol}_{\Omega}(\{H \leq E\})$. For a co-dimension two surface $\{f_1 = E_1, f_2 = E_2\}$, the 
Liouville
measure is given by $\mu_{\{f_1 = E_1, f_2 = E_2\}} = \frac{\Omega}{\dd f_1 \wedge \dd f_2} |_{f_1 = E_1, f_2 = E_2}. $ To avoid excessive subscripts, we often denote the Liouville
(or, Leray) measure on a manifold $Y$ by $\mu_Y$.

\subsection{Quantum ladders }

We define the null-space of the wave operator $\wt \Box = \wt \Box_{\wt g}$ of $(\wt M, \wt g)$, $\wt{\Box} = \Box_g - \frac{\partial^2}{\partial \theta^2}$,
 \begin{equation} \label{HKG} \hcal_{KG}^\infty : = \ker \wt{\Box} \cap C^\infty = \ker (\Box_g   - \frac{\partial^2}{\partial \theta^2}) \cap C^\infty, \end{equation}
 i.e. the space of smooth solutions.
 Following \cite{SZ18},  we endow in Section \ref{HSSECT} the space of smooth solutions $\ker \wt{\Box} \cap C^\infty$ with a positive semi-definite quadratic form with finite dimensional null-space and define $\hcal_{KG}$ as the completion in the associated topology. This therefore defines a notion 
of a trace $\Tr_{\hcal_{KG}}$.  We are  interested in the joint spectrum of $\R \times S^1_T$ on the null-space \eqref{HKG}, where $\R$ acts by 
the Killing flow $e^{t Z}$ and $S^1_T$ acts on the $S^1_T$ factor. 
The results of \cite{SZ18} apply  directly to the spectrum of $D_Z$ or $D_\theta$ in \eqref{HKG}, and in particular imply
that  the spectrum of the $\R \times S^1_T$ action on $\ker \wt{\Box}$  is discrete and consists of
pairs $\{m_k, \lambda_j(m_k)\}$ with $m_k = \frac{2 \pi}{T} k$ (see Theorem \ref{ponteigs}).

Under the $S^1$ action, \eqref{HKG} splits up into Fourier modes of $S^1$ (i.e. the eigenspaces of $D_\theta$ in \eqref{HKG}),
\begin{equation} \label{HKGm} \hcal_{KG} = \bigoplus_{m \in \Z} \hcal^{(m)}_{KG}, \;\;\;\hcal^{(m)}_{KG} : = \hcal_{KG} \cap \{D_\theta = m\},  \end{equation}
%Corresponding to the classical Hamiltonian $p_{\nu}$ is the quantum Hamiltonian
A  `ladder subspace' of \eqref{HKGm} is, roughly speaking, defined as the kernel, \begin{equation}\label{QLADDER}  \hcal_{\lcal_{\nu}} ``: ='' \{u \in \hcal_{KG} \mid P_{\nu} u = 0\}.  \end{equation} of the operator,
\begin{equation} \label{Pnu} P_{\nu} := D_Z - \nu D_\theta. \end{equation}
 Since the eigenvalues generally do not intersect lines through the origin this heuristic definition is not literally correct and only provides intuition. The precise definition  is  that the joint eigenvalues $\{(m, \lambda_j(m)\}$   lie in a strip (or thickening)
of the ray $\{(m, \lambda)\mid \lambda = \nu m\}$.  To make this precise, we  introduce real valued test functions $\psi$ with $\hat{\psi} \in C_0^{\infty}$ and define the the `fuzzy' ladder subspace to be the range of the operator
\begin{equation} \label{Pinupsi} \Pi_{\nu, \psi}: = \psi(D_Z - \nu D_{\theta}): \hcal_{KG} \to \hcal_{KG}. \end{equation}
The spectral problem is then to find the distribution of eigenvalues of $D_Z$ in a ladder. 
The notion of a fuzzy ladder is due to Guillemin-Uribe \cite{GU89} in the compact elliptic setting. We will see that our results have   many overlaps, and also
significant differences, with \cite{GU89}.

\begin{rem} In \cite{SZ18}, the authors studied the distribution of eigenvalues of $D_Z$ in a fixed Klein-Gordon space \eqref{hcalmdef}. In terms of ladders, 
the mass $m$ is fixed, and if we think of $m$ as the vertical axis, the corresponding ladder is horizontal. Horizontal ladders do not fit well with ladder notation. If we  replace \eqref{Pnu} by $\mu D_Z - D_\theta$, then a horizontal ladder cutoff would correspond to $\mu =0$ and $\psi(D_\theta)$ can be thought of as  
localizing around $m=0$ in the  spectrum of $D_\theta$. Since we assume that $\hat{\psi} \in C_0^{\infty}(\R)$,  it requires Tauberian theorems to make
such localizations. \end{rem}

\subsection{Generating functions and their singularities}
To find the asymptotics as $m \to \infty$ of \eqref{Nnupsi} and then \eqref{NnuC}, we form the generating functions,
\begin{equation} \label{Upsilon} \left\{ \begin{array}{l} \Upsilon^{(1)}_{\nu, \psi}(s) : = \sum_{m, j \in \Z} \psi(\lambda_j(m) - m \nu) 
e^{\rmi m s}, \\ \\
\Upsilon^{(2)}_{\nu, \psi}(s) : = \sum_{m, j \in \Z} \psi(\lambda_j(m) - m \nu) 
e^{\rmi \lambda_j(m)  s}. \end{array} \right., \end{equation}
The Weyl functions \eqref{Nnupsi} are essentially the Fourier coefficients of these generating functions in the
decomposition \eqref{HKGm}. More precisely,  $\Upsilon^{(2)}_{\nu, \psi}$ is the Fourier transform 
of the periodic ladder spectral function, 
\begin{equation} \dd N^{(1)}_{\nu, \psi} (x)
= \sum_{j \in \Z, m \in \Z} \psi(\lambda_j (m) - \nu m) \delta(x - m),\end{equation} 
and
$\Upsilon^{(2)}_{\nu, \psi}$ is the Fourier transform 
of the ladder spectral function,  \begin{equation} \dd N_{\nu, \psi} (x)
= \sum_{j, m \in \Z} \psi(\lambda_j (m) - \nu m) \delta(x - \lambda_j(m)).\end{equation} 
To determine the  behavior of \eqref{Nnupsi}  as $m \to \infty$, it suffices to determine  the singularities at $s = 0$ of these generating functions. 

The Weyl asymptotics of \eqref{Nnupsi}  follows from a singularity analysis of \eqref{Upsilon} 
by a Fourier Tauberian theorem. 
The results are reminiscent of  the ladder asymptotics of Guillemin-Uribe \cite{GU89} in the compact elliptic setting, and
as they explain, the singularities of $\Upsilon^1_{\nu, \psi}$ resp. $\Upsilon^2_{\nu, \psi}$ at $s = 0$ are essentially
the same.  In fact, it is simpler to use $\Upsilon_{\nu, \psi}^{(1)}$ and only sum over $m \geq 0$. 
The sum is then a Hardy distribution, and as in \cite[Section 7]{GU89} one can easily determine the large $m$ asymptotics
from the singularities by matching periodic distributions on $S^1$.

%\edit{
%The ladder Gutzwiller trace formula is, by definition, a singularities trace formula for the trace 
%\begin{equation} \label{HSnu} S_{\nu, \psi} (t): =  \Tr_{\hcal_{KG}}  \psi(D_Z - \nu D_{\theta}) e^{t Z}  |_{\hcal_{KG}}.  %%\end{equation}
%Under the $S^1$ action, the trace splits up as a Fourier series \eqref{TRACEm},  
%\begin{equation} \label{HSnum} S_{\nu, \psi} (t; m): =  \Tr_{\hcal_{KG}}  \psi(D_Z - \nu D_{\theta}) e^{t Z}  |_{\hcal^{(m)}%_{KG}}.  \end{equation}
%The large mass asymptotics refers to the asymptotics of the trace as $m \to \infty$.}

%We could also study   a joint Weyl law for pairs $(m, \lambda_j(m))$ in the joint spectrum  and a joint Duistermaat-Guillemin trace formula for the joint spectrum. %But
%our main purpose is to study Weyl laws and Duistermaat -Guillemin trace
%formulae along ladders. 

\subsection{\label{CLADDERINTRO} Classical ladders}

It is well-known in spectral asymptotics that the singularities of \eqref{Upsilon} are controlled by an underlying classical  Hamiltonian dynamical system. We need
to define it before stating our main results.  More details will be given in Section \ref{FLOWSECT}.

The characteristic variety of $\wt{\Box}$ is the null-cone bundle,
\begin{equation} \label{wtChardef} \mathrm{Char}(\wt{\Box}) = T^*_0\wt M = \{(\tilde x, \tilde \xi) \in \dot{T}^* \wt{M}\mid \wt{g}(\tilde\xi, \tilde\xi)  = 0\}, \end{equation}
where $\wt{g} (\tilde\xi, \tilde\xi)$ is the Lorentzian inner product of $\wt{M}$. 
%We define the metric Hamiltonians,
%\begin{equation} \label{HDEF} H(x, \xi)= g(\xi, \xi), \;\;\; \wt{H}(x, \xi) = \wt{g}(\xi, \xi). \end{equation}
%The homogeneous and semi-classical characteristic varieties are compared in Section \ref{SYMPSECT}.
The quotient of $\mathrm{Char}(\wt{\Box}) $ by the (null)  geodesic flow $\wt{G}^t: T^*\wt{M} \to T^*\wt{M}$ 
is the symplectic cone, 
\begin{equation} \label{NDEF} \wt{\ncal}  = \mathrm{Char}(\wt{\Box}) / \sim = T^*_0 \wt M,   \end{equation}
which is called the space 
of %future directed
scaled null-geodesics of $(\wt{M}, \wt{g})$. The space $\wt{\ncal}$ is a disjoint union $\wt{\ncal} = \wt{\ncal}_+ \cup \wt{\ncal}_-$ of future and past directed scaled null-geodesics.
In \cite[Section 2.1]{SZ18} it is shown to be a homogeneous (i.e. conic) symplectic manifold, and it is 
explained that $\wt \ncal$ is the symplectic cone of dimension $2n$ corresponding to the Hilbert space $\hcal_{KG}$ in the sense of geometric quantization.

The $\R \times S^1$ action lifts to $T^* \wt{M}$ as a Hamiltonian action with  classical  moment map 
\begin{equation} \label{mudef} \pcal: T^* \wt{M} \to \R^2, \;\; \pcal(\tilde x,\tilde \xi) = (\langle \tilde \xi, Z \rangle, \langle \tilde \xi, \frac{\partial}{\partial \theta} \rangle).  \end{equation}

For notational simplicity,  we set
$$p_Z( \tilde x, \tilde \xi) := \langle \tilde \xi, Z \rangle,  \;\;\;\; p_{\theta}( \tilde x, \tilde \xi) :=  \langle \tilde \xi, \frac{\partial}{\partial \theta} \rangle. $$
We also denote the Hamilton vector field of a function $f$ (on any cotangent bundle or symplectic manifold) by $H_f$ and we denote its
Hamiltonian flow by $\exp t H_f$. 

Given $\nu \in \R_+$, the principal symbol of \eqref{Pnu} is  the Hamiltonian \begin{equation} \label{pnudef} p_{\nu}(\tilde x, \tilde \xi) =  p_Z(\tilde x, \tilde \xi) - \nu  p_{\theta}( \tilde x, \tilde \xi). \end{equation}

\begin{defin} \label{DCHAR} The classical ladder $L_{\nu}$ is the  double characteristic variety $D \mathrm{Char}_{\nu}  = \mathrm{Char}(\wt{\Box}, P_{\nu})$,  i.e. the set of $(\tilde x,\tilde  \xi) \in \dot{T}^*\wt{M}$ such that
$$\left\{ \begin{array}{l}  \half \wt{g} (\tilde \xi, \tilde \xi ) = 0, \\ \\
p_{\nu}(\tilde x, \tilde \xi) = 0. \end{array} \right. $$
\end{defin}
We assume throughout  that  $(M,g)$ is a standard stationary spacetime,  so that we may choose a Cauchy hypersurface  $\Sigma$ of $M$ and  represent $M = \R \times \Sigma$ with $Z = \frac{\partial}{\partial t}$ (see Section \ref{GHSTSP}). We let   coordinates $(t, x, \theta) $ be corresponding coordinates  on $M \times S^1 \simeq \R \times \Sigma \times S^1$, and let  $(\tau, \xi, \sigma)$ be the dual symplectic coordinates. 
As  explained in more detail in Sections \ref{HYPERBOLOIDSECT}-\eqref{FLOWSECT},  $p_Z(t, x, \theta, \tau, \xi, \sigma) = \tau, \;
p_\theta = \sigma$. 
The equations of the double characteristic variety  are then,
\begin{equation} \label{DCHARintro} \left\{ \begin{array}{ll}  (i) & g((\tau, \xi), (\tau, \xi)) + \sigma^2 = 0, \\ & \\(ii) & 
  \tau -  \nu \sigma = 0. \end{array} \right. \end{equation}
  Thus, $\sigma$ plays the role of the mass parameter.  
  
  \begin{defin}\label{ADMISSIBLEDEF}  We say that $\nu$ is admissible if $0$ is in the range of $p_{\nu}$ on $\wt \ncal$ and is a regular value.  \end{defin}
 Criteria for  admissibility in terms of Killing horizons  are given in Section \ref{ADMSECT}. 

 In Theorems \ref{LADDERSING1} and \ref{LADDERCOR} we assume that $\nu$ is admissible in the sense of Definition \ref{ADMISSIBLEDEF}. To understand this condition, we use the equations to eliminate $\sigma$ and 
find that $\mathrm{DChar}_{\nu}$  corresponds to points of $T^*\wt{M}$ satisfying   \begin{equation}\label{CONSTRAINTS}  g((\tau, \xi), (\tau, \xi)) + \nu^{-2} \tau^2 =0. \;\;  \end{equation}
As will be proved in Proposition \ref{DCharLEM},  the projection of  $\mathrm{DChar}_{\nu}$ to $\wt M$ equals set \eqref{ACALDEF} of $\wt M$ where 
$Z - \nu \frac{\partial}{\partial \theta}$ is spacelike, and thus is bounded by the Killing horizon. It follows that $\mathrm{DChar}_{\nu} $ is empty 
if $\nu$ is such that $Z - \nu \frac{\partial}{\partial \theta}$ is nowhere spacelike. 
%in Section \ref{GHSTSP} (see \eqref{CONSTRAINTSg}), this equations has no solutions for certain values of $\nu$. 
For instance,
if $M = \R \times \Sigma$ is a product spacetime with metric $-dt^2 + h$, then the equation states that $  |\xi|_h^2 = (1 -\nu^{-2}) \tau^2$,
which has no solutions unless  $\nu \geq 1$.

 To obtain a Hamiltonian flow on a symplectic manifold, we reduce the classical ladder.  Since $p_Z, p_\theta$ are constant on orbits of the geodesic flow, 
the moment map  is well defined as a map on  \eqref{NDEF}, and defines a reduced moment map,
\begin{equation} \label{muncaldef} \pcal_{\wt \ncal}: \wt{\ncal} \to \R^2.  \end{equation}

\begin{defin}
A {\em reduced classical ladder} $\mathcal{L}_{\nu}$  is the co-isotropic submanifold, \begin{equation} \label{CLADDER} \mathcal{L}_{\nu}: = \{(\tilde x, \tilde \xi) \in \wt{\ncal}\mid p_{\nu}(\tilde  x, \tilde \xi) = 0\} =   \pcal_{\wt \ncal}^{-1} \left\{(\nu \sigma,  \sigma))\mid \sigma \in \R_+ \right\}. \end{equation}  
\end{defin}  $\mathcal{L}_{\nu}$  is a conic   hypersurface
of $\wt{\ncal}$  whose null-foliation is given by orbits of the Hamiltonian flow of $p_{\nu}$ \eqref{pnudef} on the symplectic manifold $\wt \ncal$.
%In Section \ref{FLOWSECT} we will determine the projection to $\wt M$ of $\mathrm{DChar} $ and its relation to the {\it Killing horizon} of the %Killing
%vector field $Z - \nu \frac{\partial}{\partial \theta}$. 

\subsection{\label{CONTROLSECT} The governing dynamical system}

As mentioned above, the singularities of \eqref{Upsilon} correspond to periods $s$ of a  Hamiltonian flow on 
a symplectic manifold. We say that this Hamiltonian system `controls' the singularities of \eqref{Upsilon}. The flow for massive 
ladders  is obtained by reducing 
the joint Hamiltonian flow of $p_Z, p_\theta$ on $\wt\ncal$ with respect to $\sigma =1$. For the massless ladder it is obtained by reducing with
respect to $\sigma =0$. Setting $\sigma =1$ is arbitrary, reflecting the fact that taking the product $\wt M = M \times S^1$ is an artifice whose
purpose is simply to homogenize the semi-classical theory. Reducing with respect to $\sigma =1$ is one way to de-homogenize the theory. 
 Since we will set $p_{\nu} = 0$ \eqref{pnudef}, setting $\sigma =1$ is
equivalent to setting $p_{Z} = \nu$ in $\dot{T}^* M$.

First let us consider the reduced symplectic manifold  obtained from \eqref{DCHARintro} by
setting $p_\theta = \sigma =1$ in $\wt \ncal$ and quotienting it by the Hamiltonian $S^1$ action. 
This of course un-does the product with $S^1$ that converts $M $ to $ \wt M$; the quotient depends on the level $\sigma =1$. Evidently, the reduction of the equations
\eqref{DCHARintro}
in $T^*M$ gives,  \begin{equation} \label{DCHARred} \left\{ \begin{array}{ll}  (i) & g((\tau, \xi), (\tau, \xi)) =-1, \\ & \\(ii) &  p_Z = 
  \tau = \nu. \end{array} \right. \end{equation}

The first equation defines the unit  `hyperboloid bundle' of $T^*M$ (see Section \ref{HYPERBOLOIDSECT}). The second equation gives a `time-slice
of each cotangent space hyperboloid, selecting out a `sphere' in each cotangent space.  We note that $p_Z$ is constant on $G^s$-orbits in 
the hyperboloid bundle, so it descends to the quotient,
\begin{equation} \label{NCAL1} \ncal_1 : = \{(x, \xi) \in T^*M\mid g(\xi, \xi) = -1\}/\sim, \end{equation}
where $\sim$ is the equivalent relation of lying on the same orbit of the geodesic flow. 
\begin{defin} \label{NCALNUDEF} Let $\nu$ be admissible in the sense of Definition \ref{ADMISSIBLEDEF}. The governing dynamical system for \eqref{Upsilon}
 is the Killing flow of $e^{t Z}$ on the level set \begin{equation} \label{ncal1nu} \ncal_{1}(\nu): = \{p_{Z}= \nu\} \subset \ncal_1. \end{equation}
\end{defin}

\begin{rem} In the massless case of \cite{SZ18}, the governing dynamical system is the Hamiltonian flow $e^{t Z}$ of $p_Z$ on the conic
symplectic manifold $\ncal$ of null geodesics. Since $p_Z$ is homogeneous, its Hamiltonian flow has the same periods on all non-zero
level sets. In this case the classical dynamical system can be shown to be equivalent to a magnetic flow (see for example \cite{SZ20} for a discussion of this). 
\end{rem}

We introduce the following notation for volume forms. 

 \begin{defin}  As in the notational introduction,  we denote by $\Omega$ the symplectic volume form of a symplectic manifold. Given a Hamiltonian $H$
 we denote by $\mu_L = \frac{\Omega}{\dd H} $ the Liouville volume form on a level set of $H$; we also denote it by $\mu_{\{H=1\}}$ when we
 need to specify the relevant manifold.  Thus,
  $\Omega_{\ncal_1}$ denotes the symplectic volume form of $\ncal_1$, and   $ \mu_{\ncal_1(\nu)} $ denote Liouville measure on
  the set $\ncal_1(\nu)$.  \end{defin}

The Liouville measure of $\ncal_1(\nu)$  can be computed explicitly for a standard spacetime.
When $\nu$ is a regular value of $p_Z$, it  is given by
\begin{equation} \label{volumeweyl}
 \mu_{\ncal_1(\nu)}(\ncal_1(\nu))= \mathrm{Vol}(S_{n-1}) \int_{\Sigma} \frac{\nu N}{(N^2-\beta^2)^{\frac{n}{2}}} \left( \nu^2 - (N^2-|\beta|_h^2) \right)_+^{\frac{n-3}{2}} \dd \mathrm{Vol}_h,
\end{equation}
The derivation of this formula can be found in the proof of Prop. \ref{LDENSITY}.

\subsection{\label{PHYSINT} Physical interpretation of the governing dynamical system}

Propagation of light and of particles at extremely high energy follows the rules of geometrical optics on spacetimes in the sense that singularities travel on lightlike geodesics.
In this sense the natural dynamical system associated with the high energy limit of \cite{SZ18} was the dynamics on the space of lightlike geodesics.
Here the ladder theory corresponds to increasing the energy at the same time as the mass and keeping the ratio between mass and energy at a fixed rate $\nu$.
This amounts to fixing the ratio between energy and momentum. One therefore ends up with a dynamical system on a space of timelike geodesics, thus 
corresponding to the dynamics of particles traveling at less than the speed of light in the gravitational field. The parameter $\nu$ plays as similar role as the famous Lorentz factor
$\left(1 -v^2 \right)^{-\frac{1}{2}}$ does in Minkowski spacetime (here units are chosen so that the speed of light $c$ equals $1$). To see this in a general stationary spacetime assume that we have a metric of the form
$g =  -N^2 \dd t^2 + h_{ij}( \dd x^i + \beta^i \dd t)(\dd x^j + \beta^j \dd t)$
as in \eqref{STANDARD} (see Section \ref{GHSTSP} for details) with a Lapse function $N$ and a shift vector field $\beta$.
Now assume that $\gamma(t)$ is a timelike geodesic in $M$ corresponding to the dynamics in $\ncal_1(\nu)$, i.e.
with the property that $g(\frac{\dd}{\dd t} \gamma(t),\frac{\dd}{\dd t} \gamma(t))=-1$ and such that $g(Z,\frac{\dd}{\dd t} \gamma(t))=\nu$.
If we fix the mass $m$ to be one then the covector $\xi(t)$ associated to $\frac{\dd}{\dd t} \gamma(t)$ corresponds to the relativistic momentum of the particle: in local coordinates this momentum has the energy as its time component and the momentum as its spatial component.
If we choose a local orthonormal frame that identifies the tangent space at a point $x \in M$
with Minkowski space so that the Killing vector $Z$ is given at $x$ by the vector $\sqrt{N^2-|\beta|_h^2}\partial_{t_0}$, then such a geodesic will have relativistic momentum with $t_0$-component equal to
$\nu (N^2 - |\beta|_h^2)^{-\frac{1}{2}}$ and with spatial component of length $\sqrt{\frac{\nu^2}{N^2 - |\beta|_h^2}-1}$. Therefore in this frame the spacelike component of the speed is
$$
 v(t) = \frac{\sqrt{\nu^2 (N^2 - |\beta|_h^2)^{-\frac{1}{2}}-1}}{\nu (N^2 - |\beta|_h^2)^{-\frac{1}{2}}}=\sqrt{1-\frac{(N^2-|\beta|_h^2)}{\nu^2}},
$$
and thus,
$$
 \nu =  \frac{\sqrt{N^2-|\beta|_h^2}}{\sqrt{1-v^2}}.
$$
Note that the speed is not constant since the particle moves in a gravitational field, but $\nu$ is a conserved quantity and has frame-independent meaning once $Z$ is fixed. In the product case we have $N=1$ and $\beta=0$ and this reduces to the Lorentz factor.

\section{Statement of results} \label{resection}

Our goal is to determine the singularity at $s = 0$ of $\Upsilon^{(1)}_{\nu, \psi}(s)$, and to apply the result to prove a Weyl law
for the joint spectrum along the ray.  By elaborating the techniques, we could develop the theory as in \cite{SZ18}
to prove a Gutzwiller formula giving the  singularities of the trace at non-zero $s$  in terms of periods of periodic orbits
of the governing Hamiltonian flow of Section \ref{CONTROLSECT}.  However, we do not give the details of the calculations
of the singularities at $s \not=0$ for the sake of expository brevity. Henceforth we redefine \eqref{Upsilon} as a Hardy
distribution,
\begin{equation} \label{UPSILONHARDY}  \Upsilon^{(1)}_{+,\nu, \psi}(s) : = \sum_{m = 0}^{\infty} \sum_{ j \in \Z} \psi(\lambda_j(m) - m \nu) e^{\rmi m s} = \Tr_{\hcal_{KG}}  \Pi_{\nu, \psi} \Pi_+ e^{\rmi s D_\theta},
\end{equation}
where $\Pi_+$ is the spectral projection of $D_\theta$ onto the non-negative part of the spectrum.
Then   up to a constant $c= \sum_{ j \in \Z} \psi(\lambda_j(0))$ the real part of $2 \Upsilon^{(1)}_{+}$ equals $\Upsilon^{(1)}$. The singularities of  $\Upsilon^{(1)}_{+}$ thus determine the singularities of $\Upsilon^{(1)}$. We
show that $\Upsilon^{(1)}_+$ is a polyhomogeneous Hardy distribution with a co-normal singularity at $s = 0$, i.e. it has a singularity expansion 
as a sum of model periodic distributions defined by the Fourier series (cf. \cite[Lemma 7.1]{GU89})
$$b_k(s; \omega)  = \sum_{m=0}^{\infty} m^k \omega^{-m} e^{\rmi m s},$$
where $\omega \in S^1$. In general the $\omega$ arise as holonomies corresponding to singularities at $s \not= 0$. For the singularity at $s = 0$, $\omega =1$  and the singularity is the same as for the
 model homogeneous distributions  $\mu_k(s)$ on $\R$,   defined by the oscillatory integrals (cf. \cite{GS77, GU89})
$$
 \mu_k(s) = \int_0^{\infty} \mathrm{e}^{ \rmi s \tau} \tau^{k} \mathrm{d} \tau.
$$

%\begin{theo} \label{LADDERSING1} Let $(M, g)$ 
%be a spatially compact stationary globally hyperbolic spacetime with $\dim M = n$, and let $\nu \in \R_+$. Then, $$\Upsilon^{(1)}_{\nu, \psi}(s) = e_{0; \nu, \psi} (s) %+ \rho_{0; \nu,  \psi}(s)$$ where
%$\rho_{0; \nu, \psi}(s)$ is a distribution that is smooth near $0$, and $e_{0; \nu, \psi}(t)$ is a Lagrangian distribution with
%singulariy at $s = 0$ of the form
%$$e_{0; \nu, \psi} (s) \sim 2 (2 \pi)^{-n+1} (n-1) \mathrm{Vol}(L_{\nu} \cap \{H \leq 1\}) \mu_{n-1}(s) + c_1  \mu_{n-2}(s) 
%+ \ldots. $$\end{theo}

\begin{theo}\label{LADDERSING1} Let $(M, g)$ 
be a spatially compact stationary globally hyperbolic spacetime with $\dim M = n$,   with timelike Killing vector field $Z$,  let $\nu$ be admissible in the sense of Definition \ref{ADMISSIBLEDEF}.
Let $\psi \in \scal(\R)$ with
$\hat{\psi} \in C_0^{\infty}(\R)$. Then the singular support of
$\Upsilon^{(1)}_{+,\nu, \psi}(s)$ is contained in the set { $\{s \in \R\mid \exists s' \in \pcal_{\nu}\cap \supp(\hat\psi),  s - \nu s'  \in 2 \pi \Z \} $,
where $\pcal_{\nu}$ is the set of periods of  periodic orbits of $e^{s Z}$ on $\ncal_1(\nu)$.
In case $\supp \psi \subset (-\pi/\nu,\pi/\nu)$ the distribution} 
  $\Upsilon^{(1)}_{+,\nu, \psi}(s)  =  \Tr_{\hcal_{KG}} \Pi_{\nu, \psi} \Pi_+ e^{\rmi s D_\theta} $ \eqref{UPSILONHARDY}  is a Hardy distribution on $S^1$ with a conormal singularity at $s = 0$.
That is,
$$\Upsilon^{(1)}_{+,\nu, \psi}(s) = e_{0; \nu, \psi} (s) + \rho_{0; \nu,  \psi}(s)$$ where
$\rho_{0; \nu, \psi}(s)$ is a distribution that is smooth near $0$, and $e_{0; \nu, \psi}(s)$ is a Lagrangian distribution with
singularity at $s = 0$ of the form
$$e_{0; \nu, \psi} (s) \sim (2 \pi)^{-n+1}  \;  \mu_{\ncal_1(\nu)} (\ncal_{1}(\nu))\;  \hat{\psi}(0) \mu_{n-2 }(s) + c_1(\nu,\psi)  \mu_{n-3}(s) 
+ \ldots,$$
where $\mu_{L}$ is the   Liouville measure on $\ncal_1(\nu)$ (see Definition \ref{NCALNUDEF}).
\end{theo}

There is very similar statement for $\Upsilon^{(2)}_{\nu, \psi}(s)$. 

%\begin{theo}\label{LADDERSING2}  $\Upsilon^{(2)}_{\nu, \psi}(t)  =  \Tr_{\hcal_{KG}} \Pi_{\nu, \psi} e^{t Z} $ is a Fourier integral distribution on $\R$ with a %conormal singularity at $t = 0$ and at $t$ equal
%to periodic orbits of $e^{t Z}$ on the (partially) reduced symplectic manifold $\wt{\ncal}_{\nu} $ \eqref{QUOT}.  The principal
%symbol at $t = 0$  is the given by the Liouville  symplectic volume $\mu_{L, \nu}(S \ncal_{\nu} |_{\mathrm{Supp}(\psi)} )$  of $S \ncal_{\nu}$ in the support of $\hat{\psi}%$. \end{theo}

Dually, we obtain an asymptotic expansion for smoothed ladder Weyl spectral functions. 
\begin{theo} \label{LADDERCOR} Let $(M, g)$ satisfy the assumptions of Theorem \ref{LADDERSING1}. Let $\{\lambda_j(m)\}_{j=1}^{\infty} $ be the eigenvalues of $D_Z$ in $\hcal_m$ \eqref{hcalmdef}, and let $\nu $  be admissible.  Let $\psi \in \scal(\R)$ with
$\hat{\psi} \in C_0^{\infty}(\R)$ of sufficiently small support around $s =0$ so that no other periods of  periodic orbits of $e^{sZ}$ on $\wt \ncal$
lie in its support.   Then, there exists a complete asymptotic expansion as $m \to \infty$,
$$ N_{\nu, \psi}(m) \simeq m^{n -2} \sum_{k =0}^{\infty} a_k(\nu, \psi) m^{-k}, $$
where $a_0(\nu, \psi) = (2 \pi)^{-n+1} \hat{\psi}(0) \mu_{\ncal_1(\nu)} (\ncal_{1}(\nu)). $ 
\end{theo}

Theorem  \ref{LADDERCOR} is deduced from Theorem \ref{LADDERSING1} using a  Hardy-Fourier Tauberian argument, \cite[Section 7]{GU89}. 
For the corresponding compact elliptic case, see \cite[Lemma 7.1-Corollary 7.2]{GU89}.

The sharp Weyl law for the the counting function \eqref{NnuC}  of the  ladder $\hcal_{\nu}$ states: 
\begin{theo}\label{WEYLCOR}  With the same notation and assumptions as in Theorem \ref{LADDERCOR}, 
   assume in addition  that the closed orbits of the  Killing flow on \eqref{ncal1nu}  
are non-degenerate. Then the sharp Weyl function \eqref{NnuC} admits the asymptotic expansion,
$$N_{\nu, C}(m) = 2C (2\pi)^{-n+1} \mu_{\ncal_1(\nu)}(\ncal_{1}(\nu)) \; m^{n-2}  + o(m^{n-2}),
$$
as $m \to \infty.$ 
\end{theo}

 In fact,  as in \cite[Lemma 3.3]{DG75}, the hypothesis of non-degeneracy may be weakened to the statement that the fixed point sets  of the Killing flow $e^{t Z}$  on $ \ncal_1(\nu)$ for $t \not= 0$ are of Liouville measure zero. The need for such an hypothesis is illustrated in the product case  in Section \ref{PRODUCT}, where the spectral problem coincides
 with the one studied in \cite{DG75}. When the Hamilton flow is periodic, the ladder eigenvalues $\{\lambda_j(m)\}_{m=1}^{\infty} $ cluster along an arithmetic
 progression (depending on $m$), and one does not have two-term Weyl asymptotics.   This phenomenon is well-known in spectral asymptotics and 
 is discussed for compact elliptic ladder asymptotics on \cite[page 420]{GU89}.  The sharp Weyl result is deduced from Theorem \ref{LADDERCOR} by a Tauberian theorem as in   \cite[Theorem 3.2]{BrU91}.  To prove it, we will need a slight extension of Theorem \ref{LADDERCOR}, where we remove
 the assumption on $\supp \;\hat{\psi}$.
 
 \begin{lem} \label{LADDERCORMORE} With the same notation  as in Theorem \ref{LADDERCOR},  assume that the set of periodic orbits of $e^{t Z}$ of
 period $t \not=0$  on 
 \eqref{QUOT} have measure zero. Let   $\psi \in \scal(\R)$ with
$\hat{\psi} \in C_0^{\infty}(\R)$.  Then,
$$ N_{\nu, \psi}(m) = m^{n -2} a_0(\nu, \psi) + o(m^{n-2}).$$
\end{lem}

 \begin{rem} The Weyl asymptotics are consistent with the elementary product case \eqref{PRODCASE}. They are referred
 to as `differentiated Weyl laws' in spectral asymptotics, because they count eigenvalues in short intervals. For this reason, 
  the order of the singularity in Lemma \ref{LADDERSING1} and the order of growth of eigenvalues in Theorem \ref{WEYLCOR}
 differ by $1$ from the orders in the massless case of \cite{SZ18}.
 As mentioned above, the trace asymptotics of \cite{SZ18} correspond to horizontal ladders, while the techniques of this
 paper require that $0 < \nu < \infty$.

 \end{rem}

 %For the sake of brevity we omit the details of the proof here. For the same reason, the ladder spectral asymptotics
% of \cite{GU89} pertain only to the smoothed Weyl sums $\Upsilon_{\nu, \psi}$.

\subsection{Relation to compact elliptic ladder asymptotics} 

As mentioned above, Theorem \ref{LADDERSING1} is not simply the ladder asymptotics of \cite{GU89}  transported to the Lorentzian
setting, because $\hcal_{KG}$ is not an $L^2$ subspace of $L^2(\wt{M}, dV_{\wt g})$, and there does not exist an orthogonal projection from the latter to the former. In Section \ref{PRODUCT} we specialize the results to the  case of product spacetimes in order to clarify  the nature of the results in this special case 
and to tie together the  traces \eqref{Upsilon}  on the Hilbert
spaces \eqref{hcalmdef} with the compact elliptic ladder asymptotics as in \cite{GU89}, as well as  with traditional semi-classical 
asymptotics.

Let us explain the analogies between objects of this article and objects in \cite{GU89}. The Hilbert space \eqref{HKG} corresponds to the homogeneous symplectic manifold \eqref{NDEF}. Since \eqref{HKG}
is the entire Hilbert space in our spectral problem, it plays the role of $L^2(X)$ in \cite{GU89}. Hence $\wt \ncal$ plays the role of $T^*X$ in that article. 
As mentioned above,  \eqref{mudef} generates a Hamiltonian $\R \times S^1$ action on $\wt \ncal$. Formally, the ladder Hilbert space \eqref{QLADDER}
corresponds to  the quotient, 
\begin{equation} \label{QUOT} \wt{\ncal}_{\nu}:  \{\zeta \in \wt{\ncal}: p_{\nu}(\zeta)  = 0\}/ \sim,  \end{equation}
where $\sim$ is the equivalence relation of belonging to the same orbit of the Hamiltonian flow $e^{t  Z} e^{- t \nu \frac{\partial}{\partial \theta}}$ of
$p_{\nu}$. We reduce this Hamiltonian system by setting $\sigma = 1$ on  $\wt \ncal_{\nu}$ and dividing by the $S^1$ action,
reducing to the Hamiltonian flow $e^{t Z}$ acting on the set $\{p_{Z} = \nu\} \subset \ncal_1$ \eqref{NCAL1}.  In Section \ref{FLOWSECT}, we show that there exists
a symplectic quotient for this action due to global hyperbolicity of the spacetime. 

\begin{rem} Since there exists a fibration of
$ \{\zeta \in \wt{\ncal}: p_{\nu}(\zeta)  = 0\} \to \wt \ncal_{\nu}$,   there exists an algebra (denoted $\rcal_{\Sigma}$ in \cite{GS79})  of Fourier integral operators associated
to this null foliation in the sense of \cite{GS79}. However, the result of  \cite[Proposition 4.1]{GS79}   on existence of projections  $\rcal_{\Sigma}$ is false in our setting (where the fibers are non-compact).  \end{rem}

We defined the governing dynamical system in Section \ref{CONTROLSECT}.
Let us compare this flow with the `governing Hamiltonian flow' in \cite{GU89}, especially in its applications to semi-classical Schr\"odinger
operators. In place of \eqref{mudef} of this article, the $\R \times S^1$ action in \cite{GU89} is generated by a positive elliptic operator $P$
of order $1$ and the same $D_\theta$ as in the present article. 
In the semi-classical application, one has Riemannian manifold $N$
and defines $M = N \times S^1$. Let  $P = \sqrt{-\Delta_N + V D_\theta^2}$ where $V$ is the potential.
 Denote by  $a = \sigma$, the generator of the $S^1$ action. Let $\{f_t\}$ 
be the Hamiltonian flow of the principal symbol $p$ of $P$. Then the governing flow of \cite{GU89} is the reduction of $f_t$ on the co-isotropic submanifold $p = E, a =1$. More precisely, let $Z = a^{-1}(1)$
and $B = Z/S^1$, and let $\pi: Z \to B$ be the natural projection, 
 and let $\wt p \in C^{\infty}(B)$ be the reduced symbol. Let $\{\phi_t\}$ be the Hamilton flow of $\wt p$ on $B$. The
closed trajectories of $\phi_t$ on $\wt{p}= E$ determine the singularities of $\Upsilon$.  Since  $Z = T^* N \times \{\sigma =1\} $,  $Z/S^1 = T^*N$, and the  governing flow is  that of $H = |\xi|^2 + V$ on $\{H= E^2\}$. 

 As explained above,  in this article, $\wt \ncal$ replaces $T^* (N \times S^1)$. We  set $\sigma = 1$ in $\wt \ncal$ and divide to get $\ncal_1$. 
We  then set $p_Z = \nu$ on $\ncal_1$ and study the Hamiltonian flow $e^{t Z}$ on this level set, analogously to the Hamiltonian flow on $\{H = E\}$ in
the compact elliptic setting. \bigskip

\subsection{Semi-classical heuristics}

It would be natural to work semi-classically throughout rather than taking the product with $S^1$ and
working homogeneously, which is rather artificial.   To do this,  one would need a semi-classical parametrix for the 
Green's functions of $\Box + m^2$ (see Section \ref{GREENSECT}.) Such a parametrix exists on Minkowski
space, but to our knowledge it has not been generalized to  general stationary spacetimes.
Perhaps for this reason, the physical interpretation of high mass asymptotics has not been developed. In some
sense, the semi-classical parametrix is defined implicitly in this article as a Fourier coefficient of the homogeneous
parametrix on $\wt{M}$. We hope to develop this theory further in a later article. The strategy of studying semi-classical asymptotics on $M$
by using homogeneous asymptotics on $M \times S^1$ was first used  in \cite{GU89}.

Although we  homogenized the ladder asymptotics problem in Section \ref{HOMOG}, it is useful to keep in 
mind the corresponding objects in the  semi-classical, non-homogeneous, picture with large
parameter $m$.  The Klein-Gordon operator $m^{-2} \Box_g + 1$ on $M$ has semi-classical symbol $-g(\xi, \xi) +1$, and the
semi-classical characteristic variety is the unit mass hyperboloid (bundle) $g(\xi, \xi) = 1$ (see Section \ref{HYPERBOLOIDSECT}.) 
We denote its  quotient by the (massive) geodesic flow by $\ncal_1$. 
 
The semi-classical symbol of $m^{-1} D_Z$ is
$p_Z(x, \xi) = \langle \xi, Z \rangle$ as above.  Weyl
asymptotics concerns the distribution of the eigenvalues $\lambda_j(m)$ with $m^{-1} \lambda_j(m) \simeq \nu$, which corresponds to the level 
set $p_Z = \nu$ on the quotient $\ncal_1 $. 
Note that the homogeneous approach gives the same governing Hamiltonian flow.

            \subsection{Organization}  In Section \ref{GHSTSP}  we review the geometry of globally hyperbolic stationary spacetimes from 
            a more `global geometric' view than in \cite{SZ18}.   This will prove useful when we analyze the symplectic geometry of the null geodesics
            of $\wt{M} = M \times S^1$ and the symplectic reduction in Sections \ref{HYPERBOLOIDSECT}-\ref{FLOWSECT}. In Section \ref{GREENSECT}, we review the existence of
            advanced/retarded Green's functions for the Klein-Gordon equation on $\wt{M} = M \times S^1$ and the Hilbert space inner product (Definition 
            \ref{ENERGYIP}),
            on $\hcal_{KG}$, which does not differ in any essential
            way from the discussion in \cite{SZ18}. In Section \ref{QLADDERSECT} we discuss the quantum ladder  subspaces $\hcal_{\nu, \psi}$
            in relation to symplectic reduction. Using this material, we prove Theorem \ref{LADDERSING1}  in Section \ref{THM1SECT}.
            We then derive Theorem \ref{LADDERCOR} in Section \ref{LADDERCORSECT}. 
          In Section \ref{WEYLSECT} we prove Theorem \ref{WEYLCOR}.   In Section \ref{PRODUCT} we specialize to 
          product spacetimes to compare the result with standard Weyl asymptotics.        

\section{\label{GHSTSP} Geometry of globally hyperbolic stationary spacetimes}

In this section, we briefly review the geometry of globally hyperbolic stationary spacetimes.

Lorentzian manifolds with a complete timelike Killing vector field are called {\it stationary}. If $(M,g)$ is a
 stationary globally hyperbolic spacetime then (cf. \cite{SZ18}) \begin{equation}\label{St1}  (M,
g) \simeq (\R \times \Sigma, -(N^2-|\eta|^2_h) \dd t^2 + \dd t \otimes \eta + \eta \otimes \dd t +h),
\end{equation} 
where $(\Sigma,h)$ is a Riemannian manifold, $N: \Sigma \to \R_+$ is a positive smooth function, and $\eta$ a a covector field on $\Sigma$. In this case $\partial_t$ is a Killing vector field. Such stationary spacetimes are sometimes referred to as {\it standard stationary spacetimes}.  If $(M,g)$ is a spatially compact stationary globally hyperbolic spacetime then $(M,g)$ isometric to a product $\R \times \Sigma$
with metric \eqref{St1}.
It is sometimes convenient to write this metric in the form
\begin{equation} \label{STANDARD}
  g =  -N^2 \dd t^2 + h_{ij}( \dd x^i + \beta^i \dd t)(\dd x^j + \beta^j \dd t),
\end{equation}
where $\beta$ (the shift vector field) is the vector field obtained from $\eta$ by identifying vectors with covectors using $h$ the above metric, i.e. $\beta^i = h^{ij} \eta_j$ in local coordinates.  
The coefficients are independent of $t$ so that the vector field
 $Z = \frac{\partial}{\partial t}$ is a timelike Killing vector field.

We define $\kcal = M/\R$ to be the space of Killing orbits, i.e. the quotient of $M$ by the $\R$ action
defined by $e^{t Z}$. We thus have
an $\R$-principal bundle
$$\pi: M \to \kcal$$
whose fibers are Killing orbits $[x]: = \{e^{tZ} \cdot x \mid t \in \R\}$. A model for
$\kcal$ is obtained by choosing a cross section to the Killing flow, and by
definition any Cauchy hypersurface $\Sigma$ defines such a cross section. 
A globally hyperbolic spacetime will be called {\it spatially compact} if there exists a compact Cauchy surface. In this case all Cauchy surfaces will be compact. Restriction of the the projection $M \to \kcal$
    to $\Sigma$ gives a diffeomorphism $\Sigma \simeq \kcal$.
%The volume form $dV_g$ of $(M,g)$ induces a volume form $dV_{\kcal}$ on $\kcal$ by insertion
%of $Z$ into $dV_g$  and $\mathrm{Vol}_{\kcal}(\kcal) < \infty$.
%\edit{This volume form is not used anywhere. It is neither the volume appearing in Weyl's law, neither is it the Riemannian volume of the quotient (since %this one scales with Z). I suggest to remove this sentence.}
%Geometrically, the invariant definitions
%       are based on the fact that the assumption that $M$ is stationary and globally hyperbolic implies that there exists a %well-defined quotient $\kcal = M/\R$  of
   %    the $\R$ action 
   %    $(t, x) \to e^{t Z} x,$ and therefore a principal $\R$-bundle 
   %    $$\pi: M \to \kcal, \;\; \kcal : = M / \R. \;\;  $$
    %   A Cauchy hypersurface $\Sigma$ is diffeomorphic to $\kcal$ and
     %  may be viewed as a `slice' of the $\R$ action. 
     Thus, we obtain an identification of a product decomposition $\R \times \Sigma$ with $\R \times \kcal$. In terms of any choice of Cauchy hypersurface;
  %  we may separate variables and identify $f \in C(M)$ as $f (t, x)$ with
 %   $(t, x) \in \R \times \Sigma.$ 
 we may
    also view the product decomposition as $\R \times \kcal$, namely
    the $\R$ fiber bundle $\pi: M \to \kcal$ is a product bundle.
    
   % We define a 1-form $\beta$ by specifying that $\beta(Z)  =1$ and $\ker \beta = Z^{\perp}$ (the orthogonal complement %of $Z$;  $\beta$ is a connection on the $\R$ bundle
 %   $M \to \kcal$.      
  As in  the presentation of Cortier-Minerbe \cite{CM16}, the orthogonal distributions to the fibers determine a connection $1$-form
$\Theta$ for which $\Theta(Z) = 1, \lcal_Z \Theta =0$. The metric $g$ on the horizontal spaces
$\ker \Theta$ induces a metric $g_{\kcal}$ on $\kcal$.  Define $u$ by
$u^2 = - g(Z,Z)$. It is constant along the fibers (Killing orbits), hence defines a function
on $\kcal$. The spacetime metric is then:
\begin{equation} \label{gtheta} g = - u^2 \Theta \otimes \Theta + \pi^* g_{\kcal}, \;\;  u^2 = - g(Z,Z).\end{equation}
Indeed, by definition, $g$ and $- u^2 \Theta \otimes \Theta + \pi^* g_{\kcal}$
agree on $Z$. Moreover they agree on the orthogonal complement to $Z$.
Hence they agree on all vectors. 
%We thus have two formulae:
%\begin{equation} \label{gtheta} g = - u^2 \theta \otimes \theta + \pi^* g_{\kcal}, \;\;  u^2 = - g(Z,Z).\end{equation}
  %  \begin{equation} \label{STANDARD2}
 % g =  -N^2 \dd t^2 + h_{ij}( \dd x^i + \beta^i \dd t)(dx^j + \beta^j \dd t),
%\end{equation} In general, $d \theta \not=0$; the condition $d \theta = 0$ is that of a {\it static} spacetime, where there exists an $e^{tZ}$-invariant foliation of $M$ %by Cauchy hypersurfaces.

For $\wt M = M \times S^1$ as above, $\wt{g} = g + \dd \theta^2$, so the standard metric on $\wt M$ has the form,
\begin{equation} \label{STANDARDS1} \wt g = -N^2 \dd t^2 + h_{ij}( \dd x^i + \beta^i \dd t)(dx^j + \beta^j \dd t) + \dd \theta^2. \end{equation}
\subsection{\label{SYMPSECT} Null geodesics} 

Since  $(\wt{M},\wt{g}) $  is the product of a Lorentzian space $M$ and a Riemannian space $S^1$,  its geodesic flow is
the product flow,   $\wt{G}^s (x, \xi, \theta, \sigma) = (G^s(x, \xi), G^s_{S^1}(\theta, \sigma))$ where $G^s_{S^1}$ is the Hamiltonian flow on $T^*S^1$
generated by $\sigma^2$.  That is, the  Hamiltonian generating the product flow 
is $q(x, t, \xi', \tau, \theta, \sigma): =   \ g_x((\tau, \xi'), (\tau, \xi')) - \sigma^2$ and the Hamiltonian flow on $T^*\wt{M}$ is the product
flow $G^s \times G^s_{S^1}$ on $T^* M \times  T^*S^1$.

 If $(\wt M,\wt g)$ is globally hyperbolic and $\wt \Sigma$ a Cauchy surface then each element in $\wt \ncal$ intersects $\wt \Sigma$ exactly once. The  tangent (co-)vector of the geodesic is lightlike, and its pullback to  $\wt \Sigma$ defines a covector in $\dot{T}^*\wt \Sigma$.
 This defines { a diffeomorphisms \begin{equation} \label{iota}  \iota_\pm: \wt \ncal_\pm \to \dot{T}^*\wt \Sigma , \end{equation}
 } since,  for each element $\eta \in \dot{T}^* \wt \Sigma$,  there is precisely one lightlike { future/past} directed covector $\xi \in T^*\wt M \setminus 0$ 
 whose pull-back is $\eta$. We quote the following result from 
 \cite[Proposition 2.1]{SZ18}:
 \begin{prop} \label{symplmap}
  If $(\wt M,\wt g)$ is a globally hyperbolic spacetime then  the symplectic structure on $\wt{\ncal}$ does not depend on the Cauchy surface, and  for any Cauchy surface $\tilde \Sigma$ the map \eqref{iota}
  is a homogeneous symplectic diffeomorphism.
 \end{prop}

\subsection{\label{HYPERBOLOIDSECT} Moment maps}
   
   We return to  Section \ref{CLADDERINTRO} and continue the discussion there.
   There are three commuting operators,  $\frac{1}{2}\wt{\Box}, D_Z, D_{s}$ on $\wt{M}$, whose principal symbols are Poisson-commuting functions
   on $T^*\wt{M}$,  given by \begin{equation} \label{SYMBOLS}\left\{ \begin{array}{l}  \frac{1}{2}\sigma_{\wt{\Box}}(\tilde x, \tilde \xi)
 =\frac{1}{2} \tilde \xi \cdot \tilde \xi : = \frac{1}{2} \sum_{i,j} g^{ij}(x) \tilde \xi_i \tilde \xi_j, \\ \\ \; p_Z(\tilde  x,\tilde  \xi)  =  \sigma_{D_Z}(\tilde x, \tilde \xi) =   \langle \tilde \xi, Z \rangle, \\ \\ \; p_{\theta}(\tilde x, \tilde \xi) = \sigma_{D_{\theta}} =\langle \tilde \xi, \frac{\partial}{\partial \theta} \rangle. \end{array} \right. \end{equation}  They jointly define the moment map,
 $$\wt \pcal: T^* \wt{M} \to \R^3, \; \wt \pcal( \tilde x, \tilde \xi) = (\frac{1}{2} \sigma_{\wt{\Box}}( \tilde x, \tilde \xi), \sigma_{D_Z}( \tilde x, \tilde \xi), p_{\theta}( \tilde x, \tilde \xi)). $$
 If we omit the first component we get the homogeneous moment map \eqref{mudef}.

    Both $Z$ and $\frac{\partial}{\partial \theta}$ are Killing vector fields, but $Z$ is timelike and $\frac{\partial}{\partial \theta}$ is spacelike. 
 This moment map is not homogeneous and it is better to 
 reduce the dimension and define  the degree one homogeneous  moment map \eqref{muncaldef}
% \begin{equation} \label{MM} \mu: \wt{\ncal} \to \R^2, \;\; \mu(x, \xi) = (\langle \xi, Z \rangle, \langle \xi, \frac{\partial}%% {\partial \theta} \rangle).   \end{equation}
  of the Hamiltonian action of   $\R \times S^1$ on $\wt{\ncal}$. 
  %%The characteristic variety of $\wt{\Box}$ is the homogeneous (conic) hypersurface \eqref{wtChardef}.
%$$\mathrm{Char}(\wt{\Box}): = \{(x, \xi) \in T^* \wt{M}, \; \sigma_{\wt{\Box}}(x, \xi) = (\xi, \xi) =0, \}$$
%where $(\xi, \xi)$ is the Lorentzian inner product of $\wt{M}$. 
The  Hamiltonian flow of $\frac{1}{2}(\tilde \xi, \tilde \xi)$ on $\mathrm{Char}(\wt{\Box})$ is the null geodesic
flow $\wt{G}^s$. We define the symplectic conic manifold of null geodesics of $\wt{M}$ as the orbit space \eqref{NDEF}.
%\begin{equation} \label{NDEF} \wt{\ncal}: = \mathrm{Char}(\wt{\Box})/ \sim. \end{equation}
We refer to \cite[Section 2.1]{SZ18} for proof that it is symplectic. It inherits the conic structure from $\mathrm{Char}(\wt{\Box})$.
On $\wt \ncal$ we have the reduced  moment map \eqref{muncaldef}. 
%\begin{equation} \label{wtncaldef}  \pcal_{\wt \ncal} : \ncal \to \R^2, \; \pcal(x, \xi) = (\sigma_{D_Z}, p_{\theta}).  \end{equation}

\subsection{\label{FLOWSECT} Ladder varieties, their projections and the Killing horizon} The  characteristic variety of $\wt{\Box}$ is
$$( \bigcup_{\sigma \in \R} H_{\sigma} \times S^1 \times  \{ \sigma\}) \setminus 0 \subset \dot{T}^* \wt{M},$$
where $H_{\sigma} =  \{(x, t, \xi', \tau') \in T^* M : g_x((\tau, \xi'), (\tau, \xi')) = -\sigma^2 \}$ is the mass hyperboloid
bundle in $T^* M$. The additional equation $\tau = \nu \sigma$   of $L_{\nu}$ 
pulls out one sphere 
$$S_{\sigma, \nu}: = \{(x, \xi) \in H_{\sigma} \mid \tau = \nu \sigma\} $$
in the $\sigma$- hyperboloid. 
Thus, the double characteristic variety, the zero set of the map $(\frac{1}{2} \sigma_{\wt \Box }, p_\nu)$,   is
$$\mathrm{DChar}_{\nu} = (\bigcup_{\sigma \in \R} S_{\sigma, \nu} \times S^1 \times \{\sigma\})\setminus 0. $$
At the regular points of $(\frac{1}{2} \sigma_{\wt \Box }, p_\nu)$ we have the Liouville density $\mu_{D\mathrm{Char}_\nu}$
 induced by this map.
The equations define a codimension two homogeneous co-isotropic submanifold of $\dot{T}^*\wt{M}$.  The purpose of this section is, first, to specify
this submanifold more precisely for a standard spacetime. Secondly, we determine its projection to $\wt{M}$ and compare it to the Killing horizon
of $Z - \nu \frac{\partial}{\partial \theta}$. The projection may be viewed as the ``classically allowed region'' for our problem and we determine its
boundary.  We assume the  metrics are in  standard form \eqref{STANDARDS1} on $\wt{M}$  (see Definition \ref{DCHAR})
in coordinates $(t, x, \theta)$ relative to a splitting $M = \R \times \Sigma \times S^1 =  \R \times \tilde \Sigma$ as in the previous section. We let $\tau, \xi, \sigma$ be the dual symplectic
coordinates, i.e. they pick out the coefficients of $\dd t, \dd x_j, \dd \sigma$ in a covector $\dd f(t, x,\theta)$.
We denote by $\mathrm{DChar}_{\nu} (\theta, \sigma) \subset T^*M$ the slice of $\mathrm{DChar}_{\nu}$ where $(\theta, \sigma)$ is held fixed.

A symplectic quotient of $\mathrm{DChar}_{\nu} $ is obtained by setting $\sigma =1$ and dividing by the $S^1$ action, and then dividing by the 
geodesic flow to obtain $ \ncal_1(\nu)=\{p_Z = m \nu\} \subset \ncal_1$. This hypersurface also has a symplectic quotient since each mass $1$ geodesic 
intersects $\Sigma$ exactly once. 

The projection of $\mathrm{DChar}_{\nu} (s, \sigma)$ is closely related to the set of points where the Lorentzian norm of the Killing vector field $Z - \nu \frac{\partial}{\partial \theta}$ vanishes. We will refer to this set as the Killing horizon $\wt \kcal$.
\begin{defin} The Killing horizon $\wt \kcal$ of $Z -  \nu \frac{\partial}{\partial \theta}$ is the subset in $\R \times \wt{M}$ where
$g(Z, Z) = - \nu^2$.
\end{defin}
%To avoid excessive  technicalities, we  assume that that the Killing horizon is a smooth hypersurface, i.e. that $d g(Z, Z) \not= 0$ when $g(Z, Z) = - \nu^2$.  \edit{ In case the complement of $Z - \nu \frac{\partial}{\partial \theta}$ is an integrable distribution this horizon is a null hypersurface, and the Killing horizon in our sense then coincides with the usual notions in general relativity (which usually requires this surface to be null).}

\begin{lem} \label{KHLEM} The Killing horizon of $Z -  \nu \frac{\partial}{\partial \theta}$ is the product  $\R \times \kcal_{\Sigma} (\nu) \times S^1$, where
$$\kcal_{\Sigma} (\nu):=  \{ |\beta|_h^2 = N^2 - \nu^2\}.$$

\end{lem}
 \begin{proof} Consider $(x, \theta) \in \wt{\Sigma}$ and
$V \in T_{(x, \theta)} \wt{\Sigma}$. Then $V = (V_M, V')$ where $V_M \in T_x M, V' \in T_\theta S^1$ and
$\wt{g}(V,V) = g(V_M, V_M) + |V'|^2$. \end{proof}

The set  \begin{equation} \label{ACALDEF}  \acal_{\nu}: = \{(x, \theta) \in \wt \Sigma \mid N^2-|\beta|_h^2 > \nu^2\} \end{equation} is the set where  $Z - \nu \frac{\partial}{\partial \theta}$ is spacelike,
and we give it the term `allowed region' by analogy with allowed regions $\{V \leq E\} $ for non-relativistic Schr\"odinger operators 
$-h^2 \Delta + V$. The region where  $Z - \nu \frac{\partial}{\partial \theta}$ is timelike
 plays the role of the `forbidden' region.

%\begin{lem} \label{DCHARPROJLEM} \end{lem}

We now determine $\mathrm{DChar}_{\nu} $ in phase space. 
For any manifold $X$, let $\pi: T^* X \to X$ be the natural projection.

\begin{prop}\label{DCharLEM} In a standard spacetime, for fixed $(\theta, \sigma)$,   $\mathrm{DChar}_{\nu}(\theta, \sigma)$ is  the set of solutions $(t, x; \tau, \xi, \sigma)$ of the equation,
$$  (\tau - (\beta, \xi))^2   - N^2 h^{-1}(\xi,\xi) - N^2 \sigma^2  = 0, \;\;\;
\tau = \nu \sigma,
$$
or equivalently,
 $$
  (\tau ( \nu ^2-N^2)  -(\beta, \xi)  \nu ^2)^2  =\nu ^2 N^2 \left((\beta,\xi) ^2+
   h^{-1}(\xi,\xi) \left(\nu^2-N^2\right)\right),
 \;\;\; \tau = \nu \sigma,  $$
  and   $\mathrm{DChar}_{\nu}  =\bigcup_{(\theta, \sigma) \in T^* S^1} \mathrm{DChar}_{\nu}(\theta, \sigma) \times \{(\theta, \sigma)\}.$  The projection of
   $\mathrm{DChar}_{\nu}(\theta, \sigma)$ to $M$ is the set \eqref{ACALDEF} of points  where $Z - \nu \frac{\partial}{\partial \theta}$ is spacelike. The boundary 
of the projection is the Killing horizon. 
 \end{prop}

\begin{proof}

 Let $\wt H(t, x, \theta; \tau, \xi, \sigma)$  denote the norm-square function
$\wt \xi \in T^* \wt M \to \wt g^* (\wt \xi, \wt \xi)$
of the dual co-metric metric $\wt{g}^*$ on $T^* \wt{M}$. The dual co-metric has the form,$$
 \wt g^{-1} =  N^{-2}\left( \begin{matrix} -1 & \beta & 0  \\ &&  \\  \beta^T &  N^2 h^{-1} -\beta \otimes \beta & 0 \\ && \\ 0 & 0 & N^2
  \end{matrix} \right),
$$
and a little bit of algebra gives,  $$\begin{array}{lll} \wt H & = &  - N^{-2}  \tau^2 + 2 N^{-2} \tau (\beta, \xi)  + ( h^{-1} - N^{-2} \beta \otimes \beta)( \xi, \xi) + \sigma^2
\\&&\\ &  =& - N^{-2}  [ (\tau - (\beta, \xi))^2 - (\beta, \xi)^2   - ( N^2 h^{-1} -\beta \otimes \beta) (\xi, \xi ) - N^2 \sigma^2]. \end{array}$$
hence  the double characteristic variety is defined by 
$$
\left\{ \begin{array}{l}    (\tau - (\beta, \xi))^2   - N^2 h^{-1}(\xi,\xi) - N^2 \sigma^2  = 0,   \\ \\
\tau = \nu \sigma. \end{array} \right. 
$$

Eliminating $\sigma$ by the second equation,  as in \eqref{CONSTRAINTS}, we get a constraint  equation for $(\tau, \xi) \in T^* M$:

 \begin{equation} \label{CONSTRAINTSg} \begin{array}{l}  
   (\tau (\nu ^2-N^2)  -(\beta, \xi)  \nu ^2)^2  =\nu ^2 N^2 \left((\beta,\xi) ^2+
   h^{-1}(\xi,\xi) \left(\nu^2-N^2\right)\right),
\end{array} \end{equation}
Obviously, a necessary condition that  $(t, x, \theta; \tau, \xi, \sigma)$ solve the equation is that $(x, \xi)$ lies in the following set,
 \begin{equation} \label{POS}  \dcal_{\Sigma}(\nu) : = \{(x, \xi) \in T^*\Sigma \mid  \frac{(\beta,\xi)^2}{ h^{-1}(\xi,\xi)} \geq 
   \left(N^2- \nu^2\right)\}. \end{equation} 
   Since $h(\beta, \beta) \geq \frac{(\beta,\xi)^2}{ h^{-1}(\xi,\xi)} $ for all $\xi \not=0$, the projection of $\dcal_{\Sigma}$ 
   to $\Sigma$ lies in \eqref{ACALDEF}.
   %the set $\{||\beta||^2 \geq N^2 - \nu^2\}$, the `inside' of the Killing horizon as in Lemma \ref{KHLEM}.   
    When $Z - \nu \frac{\partial}{\partial \theta}$ is timelike, the orthogonal complement of $Z - \nu \frac{\partial}{\partial \theta}$ has empty intersection
with the null cone bundle. It follows that   the projection of $\mathrm{DChar}_{\nu}$ is contained in \eqref{ACALDEF} . Conversely,
if $(t, x, \theta) \in \acal_{\nu}$ then $\xi = \beta \in \dcal_{\Sigma}$. Hence, the projection equals \eqref{ACALDEF}.

 Since the equations are independent of $t$,  $\pi \mathrm{DChar}_{\nu}$ is the set $\R \times  \dcal_{\Sigma}(\nu) \times S^1$. 
For $(x, \xi) \in \dcal_{\Sigma}(\nu), $ there exist two solutions $(t, x, \tau, \xi) \in \mathrm{DChar}_{\nu}(s, \sigma) $  (with multiplicity) of the defining equation in Lemma \ref{DCharLEM}, 
and $\mathrm{DChar}_{\nu} $ is the union over $\dcal_{\Sigma}(\nu) \times S^1 \times \R$ of the pair of solutions times $\{(\theta, \nu^{-1} \tau)\}$. 

\end{proof}

%\begin{cor} $\mathrm{DChar}_{\nu}$ is empty unless $\nu$ is admissible in the sense of Definition \ref{ADMISSIBLEDEF}. \end{cor}
%\edit{This seems incorrect as stated, probably because it is outdated. Also a regular value does not have to be a value, since the empty set is also a manifold. So I suppose the statement is $\mathrm{DChar}_{\nu}$ is empty unless $D_\nu$ is not timelike for some point. Do we even need this Corollary?
%}

The reduction $\mathrm{DChar}_\nu$ with respect to the Hamiltonian flow of $\sigma_{\wt \Box}$ and further reduction of the level set $\sigma=m$
gives the level set $\tau = \nu m$ on $\ncal_m$ and therefore $\ncal_m(\nu)$. 

\begin{prop} \label{LDENSITY}
 The Liouville measure $\mu$ of the level set $\tau - \nu \sigma=0$ in $\ncal_m$ is given by
 $$
  \mu=m^{n-2} \mu_{\ncal_1(\nu)}(\ncal_1(\nu))=m^{n-2} \mathrm{Vol}(S_{n-2})  \int_\Sigma \frac{N \nu}{(N^2-\beta^2)^{\frac{n}{2}}} \left( \nu^2 - (N^2-|\beta|_h^2) \right)_+^{\frac{n-3}{2}} \dd \mathrm{Vol}_h.
 $$
\end{prop}
\begin{proof}
 We first remark that reduction with respect to the Hamiltonian flow of $\half \sigma_{\wt \Box}$ can be achieved by choosing $T^* \wt{\Sigma}$ as a cross section.
 The level $\sigma=m$ and further reduction shows that $\ncal_m$ can be identified with
 $H_m |_{\Sigma} \times \{1\} \times \{m\} \subset T^* \wt{M} |_\Sigma$.
 
 We compute the volume over a point in $\R \times \Sigma \times S^1$ with respect to Riemann normal coordinates $y$ on $\Sigma$, so that
 $h$ at this point is the identity matrix. We can choose these coordinates in such a way that $dy_1(\beta) = | \beta |_h$ and $dy_k(\beta)=0$ for $k>1$.
 We denote the coordinates of a covector in $T^* \Sigma$ in this coordinate system by $(\xi_1,\xi_\perp)$, where $\xi_\perp = (\xi_2,\ldots,\xi_{n-1})$.
 Inserting $\tau=\nu m + q$ in the equation for the unit hyperboloid bundle we obtain
 \begin{gather*} 
  \left( \nu m + q - |\beta|_h \xi_1 \right)^2 = N^2 \xi^2 + N^2 \sigma^2 = N^2 \xi_1^2 + N^2 \xi_\perp^2 + N^2 m^2.
 \end{gather*}
 A little bit of algebra gives
 \begin{gather*} 
  \left( \xi_1 + \frac{|\beta|_h (\nu m+q)}{N^2-|\beta|_h^2} \right)^2 (N^2-|\beta|_h^2) + N^2 \xi_\perp^2 = \frac{N^2 m^2}{N^2-|\beta|_h^2} \left( (\nu+q/m)^2 - (N^2-|\beta|_h^2) \right),
 \end{gather*}
 which defines a distorted ellipsoid of volume
 $$
  \mathrm{Vol}(B_{n-1}) \frac{N m^{n-1}}{(N^2-|\beta|_h^2)^{\frac{n}{2}}} \left( (\nu+q/m)^2 - (N^2-|\beta|_h^2) \right)^{\frac{n-1}{2}},
 $$
 taking the derivative in $q$  gives
 $$
  \mathrm{Vol}(S_{n-1}) \frac{N m^{n-2} (\nu+q/m)}{(N^2-|\beta|_h^2)^{\frac{n}{2}}} \left( (\nu+q/m)^2 - (N^2-|\beta|_h^2) \right)^{\frac{n-3}{2}},
 $$
 which yields the claimed formula for $q=0$. Setting $m=1$ gives the Liouville measure for $\ncal_1(\nu)$.
\end{proof}

We close this section with the following observation:
\begin{lem} If $\wt \kcal$ is  regular and $(t, x, \theta) \in \wt \kcal$ then $(t, x, \theta, Z - \nu \frac{\partial}{\partial \theta}) \in \mathrm{DChar}_{\nu}. $ Hence, $N^* (\wt \kcal)$ is
a Lagrangian submanifold of $T^*\wt M$ contained in $\mathrm{DChar}_{\nu}$. \end{lem}

\begin{proof} If  $(t, x, \theta) \in \wt \kcal$, then by definition, identifying covectors and vectors using the metric, $(t, x, \theta, Z - \nu \frac{\partial}{\partial \theta}) \in \mathrm{Char}(\wt \Box).$  Also,  if $(t, x, \theta) \in \wt \kcal$ then 
$p_{\nu} (Z - \nu \frac{\partial}{\partial \theta}) 
= \langle  (Z - \nu \frac{\partial}{\partial \theta}),  (Z - \nu \frac{\partial}{\partial \theta}) \rangle = 0. $
 Since $\wt \kcal$ has codimension one in $\wt M$,  $\{(t, x, \theta, r (Z - \nu \frac{\partial}{\partial \theta})) :  (t, x, \theta) \in \wt \kcal, r \in \R\} = N^* (\wt \kcal)$ is
a Lagrangian submanifold of $T^*\wt M$ contained in $\mathrm{DChar}_{\nu}$.  \end{proof}

\subsection{\label{ADMSECT} Admissibility}

We now give equivalent criteria for admissibility in the sense of Definition \ref{ADMISSIBLEDEF}.

\begin{prop} $\nu$ is admissible unless $\nu = \sqrt{N^2 - |\beta|^2}$ at a point $x \in \Sigma$ where $d (N^2 - |\beta|^2) =0$.  
%if and only if $Z - \nu \frac{\partial}{\partial}$ is never lightlike. Thus,
 % $Z - \nu \frac{\partial}{\partial \theta}$ is everywhere  spacelike. 
  In particular, if  $\nu > \max_{x \in \Sigma} \sqrt{N^2 - |\beta|^2}$, then $\nu$ is admissible and $Z - \nu \frac{\partial}{\partial \theta}$ is 
  everywhere spacelike.  \end{prop}
  
\begin{proof}   Setting $\sigma = 1$ in $\mathrm{DChar}$ produces 
the equation for $\ncal_1 \subset T^* M$: 
$$   (\tau - (\beta, \xi))^2   =  N^2 \left( h^{-1}(\xi,\xi) + 1 \right)  \iff    \tau =  (\beta, \xi)   \pm  N \sqrt{ \left( h^{-1}(\xi,\xi) + 1 \right)}. $$ 
This equation is independent of $t$ and we may regard it as an equation on $T^* \Sigma$. 
Hence we would like levels $\nu$  and find $(x, \xi)$ such $\tau = \nu$ and that $$d \tau =  \dd (\beta, \xi)   \pm  \dd [ N \sqrt{ \left( h^{-1}(\xi,\xi) + 1 \right)}] = 0 . $$
Take the derivative in $\xi$ first, a necessary condition is, 
$$ \beta =  \pm \frac{N d_{\xi} h^{-1}(\xi, \xi) }{\sqrt{ \left( h^{-1}(\xi,\xi) + 1 \right)}}=  \pm \frac{N  h^{-1}(\xi) }{\sqrt{ \left( h^{-1}(\xi, \xi) +  1 \right)}}
\iff  \frac{\beta}{N} = \pm \frac{  h^{-1}(\xi) }{\sqrt{ \left( h^{-1}(\xi, \xi) +  1 \right)}}  . $$
For $\beta(x) \not= 0$, there exist two distinct solutions in $\xi$. We claim that the solution in the future sheet is, 
$$\xi = \frac{ h(\beta, \cdot)}{\sqrt{N^2 - |\beta|^2}} = \frac{\beta}{N}. $$ One has the corresponding values $\sqrt{h^{-1}(\xi, \xi) + 1} = \frac{N}{\sqrt{N^2 - |\beta|^2}}$, and the  value of $\tau$ in the future pointing sheet is, 
$$\begin{array}{lll}  \tau_{+}  &= &  \sqrt{N^2 - |\beta|^2}.  \end{array}  $$
It is the  minimum value of $\tau$ on the forward sheet, and this fiber-wise critical point exists and is unique (on the forward sheet) for
all $(t, x) \in M$. 
It follows that, if $\nu > \max_{x \in \Sigma} \sqrt{N^2 - |\beta|^2}$, then $\nu$ is a regular value.

We have not yet considered derivatives in $x$, so that these values of $\tau_+$ are only necessary conditions for a critical point. Fully critical points
are those   points $x \in \Sigma$ which additionally satisfy $d (N^2 - |\beta|^2) =0$.

\end{proof}
 
 By Lemma \ref{KHLEM},
the Killing horizon of $Z -  \nu \frac{\partial}{\partial \theta}$ is the product  $\R \times \kcal_{\Sigma}(\nu)\times S^1$, where
$\kcal_{\Sigma}(\nu) :=  \{ (x,\theta ) \in \wt M: \nu^2  = N^2 -   |\beta|_h^2\} $. Thus it consists of points $(t,x,\theta)$ where $(t, x) \in M$ and $\nu$ is the 	`height' $\tau$ of the bottom of the
unit mass hyperboloid. Since the bottom exists over all $(t, x)  \in M$, admissibility is not simply the condition that $\nu$ is never the bottom height, unless
(as in the case of product spacetimes), the bottom height is constant. Indeed, by Sard's theorem, the set of critical values has measure zero, while
the bottom height usually fills out an interval.

  \section{\label{HSSECT} Hilbert space topology on  $\hcal_{KG}$ and quantum ladders}
   
%  The purpose of this section is to review the Hilbert space topology on $\hcal_{KG}$ \eqref{HKG} and to define
%  the ladder space. This allows us to define the `fuzzy ladder projector'  
%
%   
%  \subsection{Hilbert space structure on $\ker \wt{\Box}$} 

The space $\ker \wt{\Box} \cap C^\infty$  of smooth solutions of $\wt{\Box} u=0$ is naturally a symplectic space, with symplectic form defined by
\begin{equation} \label{SYMPLECTIC}
 \sigma(u,v) = \int_{\wt{\Sigma}} \left( (\nu_{\tilde x} u)(\tilde x) v(\tilde x) - v(\tilde x) (\nu_{\tilde x} u)(\tilde x) \right) \mathrm{dVol}_{\wt{\Sigma}},
\end{equation}
where integration is over any Cauchy surface $\wt{\Sigma}$, and $\nu$ denotes the future directed unit normal vector field to $\wt{\Sigma}$.
By Green's identity,  the definition does not depend on the choice of Cauchy surface.

   We now equip $\ker \wt{\Box}\cap C^\infty$, as in  \eqref{HKG} with the energy inner product. Let $\wt{g} = g + \dd \theta^2$ be the
    Lorentzian metric on $\wt{M}$.
      For $u \in C^\infty(\wt{M})$, we  define the stress-energy tensor $T(u)$  by
\begin{gather} T(u): = 
\dd \overline{u} \otimes \dd u - \frac{1}{2} |\dd u|^2 \wt{g}.
\end{gather} 
\begin{defin} \label{ENERGYIP} The energy (quadratic) form on the space $\ker \wt{\Box} \cap C^\infty(\wt{M})$ is defined by polarization of
$$Q(u, u) = \int_{\wt{\Sigma}} \langle T(u)(Z), \nu \rangle \dd \mathrm{Vol}_{\wt {\Sigma}}$$
where $\nu$ is the unit normal to $\wt{\Sigma}$, a spacelike
hypersurface and $Z$ is the timelike Killing vector field.
\end{defin}
The following is a standard observation:
 \begin{lem} If $\wt{\Box} u = 0$ then the  covector field $T(u)(Z) $ is divergence free. \end{lem}

We extend the definition of the quadratic form in a sesquilinear manner to complex valued functions.
This quadratic form is independent of the chosen Cauchy surface.
As verified in \cite[Section 3.2]{SZ18},
\begin{lem}
 The energy quadratic form is invariant under the Killing flow, i.e.
 $$
  Q(e^{\rmi s D_Z} u,e^{\rmi s D_Z} u) = Q(u,u)
 $$
 for all $u \in \ker \wt{\Box}$ and related to the symplectic form by
 $$
  Q(u,v)=\frac{\rmi}{2}  \sigma(\bar u, D_Z v) = \frac{1}{2} \sigma(\bar u, \mathcal{L}_Z v).
 $$
\end{lem}

The quadratic form is positive definite in $\dd u$ and therefore only positive semi-definite on $\ker \wt{\Box}$.
The kernel of $Q$ is one dimensional and spanned by the constant function. In $\ker \wt{\Box}$ we have the two dimensional
invariant subspace spanned by the functions $1$ and $t$. On this subspace the operator $D_Z$ is not diagonalisable but has a nontrivial Jordan-block. We will denote this space spanned by $\{1,t\}$ by $V_0$ as it is the generalized eigenspace of the operator $D_Z$ with eigenvalue $0$.
If $V$ is the symplectic complement of $V_0$ then $V$ is an invariant for $D_Z$ and we have
$$
 \ker \wt{\Box}\cap C^\infty = V_0 \oplus V.
$$
This direct sum is orthogonal with respect to $Q$ and $Q$ restricted to $V$ is positive definite. We can therefore complete $V$ to a Hilbert space
$\mathcal{H}^c_{KG}$. Choosing any inner product on $V_0$ this defines the topology of a Hilbert space (a Hilbertisable locally convex topology)
on $\ker \wt{\Box}\cap C^\infty$. The completion with respect to the uniform structure induced by this locally convex topology then is the space
$$
 \mathcal{H}_{KG} =  V_0 \oplus  \mathcal{H}^c_{KG}.
$$
The symplectic form extends to $\mathcal{H}_{KG}$ and it is easy to see using the explicit representation in standard form
that $\mathcal{H}_{KG}$ coincides with the space of solutions with Cauchy data in $H^1(\wt{\Sigma}) \oplus L^2(\wt{\Sigma})$ and that the Cauchy data map is a continuous bijection (see  \cite[Section 3]{SZ18} for details).
Since the space $V_0$ is finite dimensional $\mathcal{H}_{KG}$ is essentially a Hilbert space.

Using Theorem \ref{UphiFORMa} below, the following is proved in \cite[Theorem 5.2]{SZ18}. We simplify the statement because the potential $V$
in that article is zero here. 
 
 \begin{theo} \label{ponteigs} Suppose that $(\wt{M},\wt{g})$ is a spatially compact globally hyperbolic stationary spacetime. Then
 $D_Z$ is a self-adjoint operator on $\mathcal{H}^c_{KG}$ with discrete spectrum that consists of  infinitely many real eigenvalues of finite multiplicity that accumulate at $-\infty$ and $+\infty$. 
 \end{theo}
 Since complex conjugation anti-commutes with $D_Z$ the spectrum is symmetric about the origin. This means apart from the non-trivial Jordan block in $V_0$ the operator $D_Z$ can be completely diagonalised.

       \section{\label{GREENSECT}Fundamental solutions and parametrices on $M \times S^1$}
       
       In this section, we review the basic results on retarded/advanced Green's functions on globally
       hyperbolic spacetimes, and apply them to $\wt{M} = M \times S^1$.
       
       We consider the Klein-Gordon operator $\wt{\Box}  = \Box_g - \frac{\partial^2}{\partial \theta^2}$. 
      A  {\it fundamental solution}  of $\wt{\Box} $  is a distribution kernel $F: C_0^{\infty}(\wt{M}) \to C^{\infty}(\wt{M})$
      such that  $\wt{\Box} F = F \wt{\Box} = \mathrm{id}_{C^\infty_0(\wt{M})}$.
A {\it parametrix} is a map $F: C^\infty_0(\wt{M}) \to C^\infty(\wt{M})$ such that $\wt{\Box} F = F \wt{\Box}  = \mathrm{id}_{C^\infty_0(\wt{M})} \textrm{ mod } C^\infty$. 
A fundamental solution $E_\mathrm{ret/adv}( x, y) $ is called {\it retarded/advanced} if $\supp(E_\mathrm{ ret/adv} f) \subset J_\pm(\supp f)$,
where $J_\pm(K)$ refers to the causal future/past of a subset $K \subset M$.

Let $T_0^*\wt{M}$ be the set of null covectors (the bundle of light-cones in cotangent space). Then $T_0^* \wt{M} \setminus 0$ is a closed conic subset in $T^*\wt{M}$
and the set
\begin{equation} \label{CDEF} \wt{C} = \{(\tilde x,\tilde \xi,\tilde x',\tilde \xi')  \in (T_0^* \wt{M} \setminus 0)^2 \;\; \mid (\tilde x, \tilde \xi) \in \mathrm{Char}(\wt{\Box}),  (\tilde x,\tilde \xi) = \wt{G}^s (\tilde x',\tilde \xi') \textrm{ for some } s \in \R\} \end{equation}
defines a homogeneous canonical relation from $T^*\wt{M} \setminus 0$ to $T^*\wt{M} \setminus 0$. Here $\wt{G}^s$ denotes the null  geodesic flow of $\wt{M}$.

A basic result of Duistermaat-H\"ormander \cite[Theorem 6.5.3]{DH72}   (see also  Theorem 3.4.7 in \cite{BGP} and \cite[Theorem 7.1]{SZ18}) is the following,
\begin{theo}\label{DHTH}
 If $(\wt{M},g)$ is a globally hyperbolic spacetime, then there exist unique retarded fundamental solutions $\wt{E}_\mathrm{ret/adv}$ for $\wt{\Box}$. The difference $\wt{E}=\wt{E}_\mathrm{ret}-\wt{E}_\mathrm{adv}$ 
 is a Fourier integral operator in $I^{-\frac{3}{2}}(\wt{M} \times \wt{M}, \wt{C}')$ with principal symbol 
 $\mp \frac{\rmi}{2} \sqrt{2 \pi} |\dd s|^{\half} \otimes | \mu_{\mathrm{char}(\half \wt{\Box})}|^\half$  on $C'_\pm$.
\end{theo}

\subsection{\label{ERSECT} Cauchy extension and restriction operators}
We now define Cauchy extension and restriction operators for general globally hyperbolic spacetimes $ \wt{M}$ with compact Cauchy
hypersurfaces $\wt{\Sigma}$. 

 For  $f,g \in C^\infty_0(\wt{\Sigma})$ we define the distributions $\delta_{\wt{\Sigma},g}$ and
$\delta'_{\wt{\Sigma},f}$ 
by
\begin{equation} \label{distcauchy}
   \delta_{\wt{\Sigma},g}(\phi) = \int_{\wt{\Sigma}}   g(\tilde x) \phi(\tilde x) \;\mathrm{dVol}_{\wt{\Sigma}}(\tilde x), \quad  \delta'_{\wt{\Sigma},f}(\phi) =-\int_{\wt{\Sigma}} f(\tilde x) (\nu_{\wt{\Sigma}} \phi)(\tilde x) \;\mathrm{dVol}_{\wt{\Sigma}}(\tilde x).
\end{equation}
where $\nu_{\wt{\Sigma}}$ is the future directed normal vector field to $\wt{\Sigma}$. 

\begin{defin} \label{EDEF} We define the Cauchy extension operators $\ecal_0, \ecal_1$ and $\ecal$ as the operators
\begin{gather*}
  \ecal_0: H^1(\wt{\Sigma}) \to \ker \wt{\Box}, \quad f \mapsto \wt{E} ( \delta'_{\wt{\Sigma},f} ),\\
  \ecal_1: L^2(\wt{\Sigma}) \to \ker \wt{\Box}, \quad g \mapsto  \wt{E}(\delta_{\wt{\Sigma},g}),\\ 
  \ecal: H^1(\wt{\Sigma}) \oplus L^2(\wt{\Sigma}) \to \ker \wt{\Box}, \quad (f,g) \mapsto  \wt{E} ( \delta'_{\wt{\Sigma},f} + \delta_{\wt{\Sigma},g}) = \ecal_0 f + \ecal_1 g.
\end{gather*}
\end{defin}
By the mapping properties of Fourier integral operators in $I^{-\frac{3}{2}}(\wt{M} \times \wt{M}, \wt{C}')$ the function $u=\ecal(f,g)$ is well defined for $f,g \in C^\infty(\wt{\Sigma})$, and by 
Green's identity applied to the causal future of $\wt{\Sigma}$,  $u$ is the unique solution of the Cauchy problem
$$
 \wt{\Box} u = 0, \quad (f,g)=(u\vert_{\wt{\Sigma}}, \nu_{\wt{\Sigma}} u\vert_{\wt{\Sigma}}).
$$
Indeed, by the support properties of retarded and advanced fundamental solutions we have $u = u_+ - u_-$, where $u_\pm = \wt{E}_\mathrm{ret/adv}(\delta'_{\wt{\Sigma},f} + \delta_{\wt{\Sigma},g})$ is supported in $J^\pm(\wt{\Sigma})=:\wt{M}_\pm$.
If $\phi \in C^\infty_0(\wt{M})$ is an arbitrary test function, then, 
\begin{gather*}
 (\delta'_{\wt{\Sigma},f} + \delta_{\wt{\Sigma},g}, \phi)= (u_+, \wt{\Box} \phi)=\int_{\wt{M}_+} (u \tilde \Box \phi) \mathrm{dVol}_g \\=
 \int_{\wt{M}_+} (u \wt{\Box} \phi - \phi \wt{\Box} u) \mathrm{dVol}_g = \int_{\wt{\Sigma}} - u(\tilde x)(\nu_{\wt{\Sigma}} \phi)(\tilde x) + \phi(\tilde x) (\nu_{\wt{\Sigma}} u )(\tilde x) \mathrm{dVol}_{\wt{\Sigma}}(\tilde x).
\end{gather*}
Since this is true for all test functions $\phi$ this shows $u$ has Cauchy data $(f,g)$.
%As an immediate consequence of Theorem \ref{DHTH}  we obtain,

The following corollary is well known and can also be proved more directly  \cite[Theorem 5.1.2]{D96}.

\begin{cor}\label{DHTHCOR}
 If $(\wt{M},g)$ is a globally hyperbolic spacetime, and $\wt{\Sigma}$ is a Cauchy hypersurface, then the Cauchy extension operators $\ecal_0$ and  $\ecal_1$ of Definition \ref{EDEF} are Fourier integral operators in $I^{-\frac{1}{4}}(\wt{M} \times \wt{\Sigma}, \wt{C_{\Sigma}}')$
 and $I^{-\frac{5}{4}}(\wt{M} \times \wt{\Sigma}, \wt{C_{\Sigma}}')$, respectively, 
 where $$\wt{C_{\Sigma}}  = \{(\tilde x, \tilde \xi, \tilde q, \tilde p) \in T^*\wt{M} \times T^* \wt{\Sigma}: g(\tilde \xi, \tilde \xi) = 0; \exists \tilde \eta \in T_{\wt{\Sigma}}^* \wt{M}:\tilde  \eta |_{T_q \wt{\Sigma}} = \tilde p, \exists s: \wt{G}^s (\tilde q, \tilde \eta) = (\tilde x, \tilde \xi) \}. $$

\end{cor}

\begin{proof} We may view $\ecal = \wt{E} \circ \delta$ where $\delta: H^1(\wt{\Sigma}) \oplus L^2(\wt{\Sigma}) \to \dcal'(\wt{M})$
is defined by $\delta(f,g) =  ( \delta'_{\wt{\Sigma},f} + \delta_{\wt{\Sigma},g})$. We then obtain the wave front relation of $\ecal$ as
the composition of the wave front relations of $\wt{E}$ given in Theorem \ref{DHTH} and the wave front relation of $\delta$, or equivalently the wave front relations of $\delta'_{\wt{\Sigma},f},  \delta_{\wt{\Sigma},g}).$ Both of these are extension operators, extending a distribution 
on $\wt{\Sigma}$ to a distribution on $\wt{M}$; they are easily seen to be adjoints of the restriction operators $\gamma_{\wt{\Sigma}},
\gamma_{\wt{\Sigma}}^{\nu}$, where $\gamma_{\wt{\Sigma}}(u) = u |_{\wt{\Sigma}}, \gamma^{\nu}_{\wt{\Sigma}}(u) = (\partial_{\nu} u)|_{\wt{\Sigma}}. $ It is well-known that the wave front relation of either  restriction operator is the graph of the restriction of a covector
$\nu$ on $\wt{M}$ to $T \wt{\Sigma}$, and the wave front relation of the extension operator is the adjoint relation. The composition
is transversal and hence is a Fourier integral operator.

\end{proof}

We further define the Cauchy restriction operator as follows:

\begin{defin} \label{RESDEF} Let $\Sigma$ be any Cauchy hypersurface in $M$ and let $\wt{\Sigma} = \Sigma \times S^1$. Define the associated  Cauchy 
restriction map by,
$$
 \rcal_{\wt{\Box}} : \ker \wt \Box \to H^1(\wt \Sigma) \oplus L^2(\wt \Sigma), \;\; \rcal_{\wt{\Box}}(u) = (u |_{\wt{\Sigma}},  \partial_{\nu_{\wt{\Sigma}}}u )  |_{\wt{\Sigma}}.
$$

\end{defin}
Comparison with \eqref{distcauchy} shows that if $u = \ecal(f,g)$ then,
$(u |_{\wt{\Sigma}},  \partial_{\nu_{\wt{\Sigma}}}u )  |_{\wt{\Sigma}} = (f,g)$, i.e. $R_{\wt{\Box}}$ is invertible as a bounded operator
on \eqref{HKG}  and 
\begin{equation} \label{INVERSE} \rcal_{\wt{\Box}}^{-1} = \ecal. \end{equation}
Although $\rcal_{\wt{\Box}}, \ecal$ obviously depend on the choice of Cauchy hypersurface, we suppress the dependence in the notation.

We also need the `global' restriction operator on all of $C^1(\wt{M})$, defined by the same formula as in Definition \ref{RESDEF}.

\begin{lem} \label{RESLEM} The global restriction operator $\rcal$ is in $I^{1/4}(\wt{\Sigma} \times \wt{M}, \Lambda), $
where $$\Lambda = \{(\tilde q, \tilde \eta; \tilde q, \tilde \xi) \mid \tilde q \in \wt{\Sigma}, \tilde \xi |_{T \wt{\Sigma}} = \tilde \eta\}.$$
\end{lem}

We refer to \cite[Page 113]{D96} for the proof. Another
 proof is given in \cite[Section 5.2]{ToZ13}, along with a discussion of a complication in the statement due to normal and tangential directions
to $\wt{\Sigma}$. Namely, the canonical relation contains points $(\tilde q, 0; \tilde q, \tilde \xi)$ when $\tilde \xi \in N^*\wt{\Sigma}$ and such points are not
allowed in the H\"ormander definition of homogeneous Fourier integral operators in  $I^k(X, \Lambda)$. To deal with this problem, one needs to 
introduce a cutoff away from normal direction, as in \cite[Section 5.1]{ToZ13}, and then the cutoff restriction operator does satisfy the statement
of Lemma \ref{RESLEM}. As long as such conormal directions do not arise in the compositions defining the trace, they may be ignored. For
the sake of expository brevity, we ignore the conormal directions in Lemma \ref{RESLEM} and only verify during the trace calculations that
they do not arise in the compositions.

\section{The wave trace of $D_Z$ in $\hcal_{KG}$ and ladder Hilbert spaces } 

In this section we review \cite[Theorem 4.1]{SZ18} and \cite[Theorem 5.2]{SZ18}, since we need these results in the present article. 
Define $\wt{U}(s) $ to be translation by $e^{s Z}$ on $\ker \wt{\Box}. $ This clearly extends to the spaces $\mathcal{H}_{KG}$ and  
$\mathcal{H}^c_{KG}.$

\begin{theo}\label{UphiFORMa}
 Suppose that $\phi \in C^\infty_0(\R)$. Then the operator 
 $$
  \wt{U}_\phi = \int_\R \phi(s) \wt U(s) \dd s:   \mathcal{H}_{KG} \to  \mathcal{H}_{KG}
 $$
 is trace-class, and its trace equals
 $$
  \Tr(\wt U_\phi)= \int_{\wt \Sigma} \int_\R \phi(s) \left(  \nu_{\tilde x} \ee^{\mathrm{i}(D_Z)_{\tilde x} t}\wt E(\tilde x,\tilde y) -  \nu_{\tilde y} \ee^{\mathrm{i}(D_Z)_{\tilde x} s}\wt E(\tilde x,\tilde y)   \right) \dd s |_{\tilde y=\tilde x} \;\mathrm{dV}_{\wt \Sigma(\tilde x)}.
 $$
\end{theo}

As explained in \cite[Section 4]{SZ18},  a  coordinate invariant way to state Theorem \ref{UphiFORMa} is that
\begin{equation} \label{TRUtDEFa} \begin{array}{ll}
 \Tr(\wt U(s)) =  \int_{\wt \Sigma} * \left( \dd_{\tilde x} \wt E_s(\tilde x,\tilde y) - \dd_{\tilde y} \wt E_s(\tilde x,\tilde y) \right)|_{\tilde y=\tilde x},
% \\  E_t(x,y) = \ee^{-\mathrm{i}(D_Z)_x t} E(x,y).
\end{array} \end{equation}
where $*$ is the Hodge star operator on $M$ and where
\begin{equation} \label{EtDEF} \wt E_s(\tilde x,\tilde y) = \ee^{\mathrm{i}(D_Z)_{\tilde x} s} \wt E(\tilde x,\tilde y). \end{equation}

Since $E$ is skew-symmetric and commutes with the flow,
$$
 \wt E_s(\tilde x,\tilde y) = -\wt E_{-s}(\tilde y,\tilde x).
$$
Hence, we also have
\begin{cor} The distributional trace defined above equals
\begin{equation} \label{TRUtDEF3b}
 \Tr(\wt U(s)) =  \int_{\wt \Sigma} * \left( \dd_{\tilde x} (\wt E_s(\tilde x,\tilde y) + \wt E_{-s}(\tilde x,\tilde y)) \right) |_{\tilde y=\tilde x}.
 \end{equation}
\end{cor}

  \subsection{\label{THM1SECT} The trace of $U(t,s)$ }

We now consider the joint propagator  \begin{equation} \label{U(t,s)} U(t, s) := e^{\rmi t D_Z} e^{\rmi s D_{\theta}} \end{equation} We also define,
\begin{equation} \label{EtsDEF} \wt{E}_{t, s} := \ee^{\mathrm{i}(D_Z) t}  \ee^{\mathrm{i}(D_{\theta}) s} \wt E. \end{equation}Since $\wt E$ is skew-symmetric and commutes with the flow,

$$
 \wt E_{t,s} = - \wt E_{-t, -s}.
$$

 %  Let ${\tilde x} = (x, \theta) \in \wt{M} = M \times S^1$.
   
   Similar to \cite[Theorem 4.1]{SZ18}, we have,
\begin{theo}\label{UphiFORMb}
 Suppose that $\phi \in C^\infty_0(\R \times S^1)$. Then the operator 
 $$
  U_\phi = \int_{\R} \int_{S^1}  \phi(t,s) U(t,s) \dd t \dd s: \mathcal{H}_{KG}\to \mathcal{H}_{KG}
 $$
 is trace-class, and its trace equals
 $$
  \Tr_{\hcal_{KG}} (U_\phi)= \int\limits_{\Sigma \times S^1}  \int\limits_\R \phi(t,s) \left(  \nu_{\tilde x} \ee^{\mathrm{i}(D_Z)_{\tilde x} t}  \ee^{\mathrm{i}(D_{\theta}) s} \wt E(\tilde x,\tilde y) -
    \nu_y \ee^{\mathrm{i}(D_Z)_{\tilde x} t}  \ee^{\mathrm{i}(D_{\theta}) s} \wt E(\tilde x,\tilde y)   \right) \dd t |_{\tilde y=\tilde x} \mathrm{dVol}_{\tilde  \Sigma}
 $$
 The integral is independent of the choice of Cauchy surface $\Sigma$.
\end{theo}

There is a more coordinate invariant way to write this expression, namely,
\begin{equation} \label{TRUtDEF2} \begin{array}{ll}
 \Tr_{\hcal_{KG}}(U(t,s)) =  \int_{\Sigma \times S^1} * \left( \dd_{\tilde x} \wt E_{t,s}(\tilde x,\tilde y) - \dd_{\tilde y} \wt E_{t,s} (\tilde x,\tilde y) \right)|_{\tilde y=\tilde x},
% \\  E_t(x,y) = \ee^{-\mathrm{i}(D_Z)_x t} E(x,y).
\end{array} \end{equation}
where $*$ is the Hodge star operator on $\wt{M}$. 

%\subsection{A related version}
 %Define the two-parameter  Cauchy evolution map, $$V(t,s ) = \rcal_{t,s}  \circ \rcal^{-1}, \;\;(\rcal_t:= \rcal \circ U(t,s)).$$
  %  If $\phi \in C_0(\R \times \R)$ define,
 %   $$V_{\phi} = \int_{\R} \int_{\R} \phi(t,s) V(t,s) \dd t \dd s. $$
%We consider  $V_\phi: H^1(\Sigma) \oplus L^2(\Sigma) \to H^1(\Sigma) \oplus L^2(\Sigma)$ as a block matrix
%$$
% V_\phi = \left( \begin{matrix} V_{11} & V_{12} \\ V_{21} & V_{22} \end{matrix} \right).
%$$

%Hence, $\Tr(V_\phi) = \Tr_{H^1}(V_{11}) + \Tr_{L^2}(V_{22})$. Now it is an easy observation that the trace on a smoothing operator is the %same on every Sobolev space, in particular we have
%$ \Tr_{H^1}(V_{11}) =  \Tr_{L^2}(V_{11})$.
%We can use the fundamental solution $E$ to express the integral kernels of the maps $V_{11}$ and $V_{22}$ as follows.
%\begin{gather*}
% V_{11}(x,y) =  -\int_\R   \phi(t) \ee^{\mathrm{i}(D_Z)_x t} \nu_y E(x,y) \dd t,\\
 %V_{22}(x,y) =  \int_\R  \phi(t) \nu_x \ee^{\mathrm{i}(D_Z)_x t}E(x,y) \dd t.
% \end{gather*}
% \edit{Here it appears that $x,y \in \wt{\Sigma}$?}

    We consider the  commuting joint eigenvalue problem on $\wt{M}$,
    $$\left\{ \begin{array}{l} \wt{\Box} u =  (\Box_g + D_\theta^2) u = 0, \\ \\
    D_Z u =  \lambda u\\ \\
    D_{\theta} u = m u. \end{array} \right. $$

   The following result  follows from Theorem \ref{ponteigs}    \begin{prop} \label{DISCRETE} The joint spectrum of   the eigenvalue  problem   $$ \begin{array}{l}     D_Z u =  \lambda u, \;\;
    D_{\theta} u = m u. \end{array} $$ in  $\hcal_{KG}$ on $\wt{M} = M \times S^1$    is discrete in $\R \times \R$,  and the eigenfunctions $(u_{m, \lambda_j(m)})_{\lambda_j \not= 0}$ are $C^{\infty}$ and define an orthonormal basis of $\hcal_{KG}^c$.  \end{prop}

    \section{\label{QLADDERSECT} Quantum Ladders and ladder traces}  
    %{
%We now return to the setting $\wt{M} = M \times S^1$ of this article,
%and  denote by $\wt{\Sigma} = \Sigma \times S^1$ the Cauchy hypersurface in $\wt{M}$ corresponding to $\Sigma$ on $M$.
%Suppose  $\hat{\psi} \in C_0^{\infty}(\R)$. The `fuzzy ladder' projector \eqref{Pinupsi}  of slope $\nu$ defined by $\psi$ is the approximate projector, 
  % \begin{equation} 
  %   \label{Pinupsiform}\Pi_{\nu, \psi}  = \psi(D_Z - \nu D_{s}) : \hcal_{KG} \to \hcal_{KG}.
 %  \end{equation}}
Next we  define the  quantum 
    ladders corresponding to the classical ladders in Section \ref{FLOWSECT}. Let $D_\theta = \frac{1}{i} \frac{\partial}{\partial \theta}$ and $D_Z = \frac{1}{i} \nabla_Z$. 
   Suppose  $\hat{\psi} \in C_0^{\infty}(\R)$.      We define the Fourier integral  operator,
  \begin{equation} \label{PinupsiformL2}  \psi(D_Z - \nu D_\theta): =   \frac{1}{2 \pi} \int_{\R} \hat{\psi}(s) e^{-\rmi s \nu   D_{\theta}} \circ e^{\rmi s D_Z}   \dd s: L^2(\wt{M}) \to L^2(\wt{M}).   \end{equation}
            Note that the operator  $\psi(D_Z - \nu D_{\theta})$ is a properly supported Fourier integral operator.

\subsection{Quantum ladder traces}

To prove Theorem \ref{LADDERSING1}, we express the generating functions \eqref{Upsilon} as ladder traces. 

    \begin{defin} \label{LADPROJDEF} Let $\hat{\psi} \in C_0^{\infty}(\R)$. The `fuzzy ladder' projector \eqref{Pinupsi}  of slope $\nu$  defined by $\psi$ is the approximate projector, 
    
    \begin{equation} 
   \label{Pinupsiform}\Pi_{\nu, \psi}  = \psi(D_Z - \nu D_{\theta}) : \hcal_{KG} \to \hcal_{KG}.\end{equation}

          \end{defin}

We can then represent the generating functions \eqref{Upsilon} by traces:

\begin{lem}
\begin{equation} \label{UpsilonTr} \left\{ \begin{array}{l} \Upsilon^{(1)}_{\nu, \psi}(s) : = \Tr_{\hcal_{KG}} \Pi_{\nu, \psi} e^{\rmi  s D_\theta},\\ \\
\Upsilon^{(2)}_{\nu, \psi}(s) : =  \Tr_{\hcal_{KG}} \Pi_{\nu, \psi} e^{\rmi s D_Z}  \end{array} \right.. \end{equation}

\end{lem}

\begin{proof} This follows directly from the definitions \eqref{Upsilon} and the spectral theory of Proposition \ref{DISCRETE}.

\end{proof}

We concentrate on the first trace, since it is sufficient for the proof of our main results. We now express the traces in  \eqref{UpsilonTr}
explicitly in terms of kernels.  As mentioned in the introduction, $\hcal_{KG}$ is not a subspace of $L^2(\wt{M})$ so we cannot express the
restriction to $\hcal_{KG}$ by an orthogonal projection.  
But,  the map $\wt{E}$ maps onto the space of smooth solutions we have
\begin{equation}\label{PiE} 
 \Pi_{\nu, \psi} \wt{E} = \psi(D_Z - \nu D_{\theta}) \wt{E}.
\end{equation}
We now express traces in terms of the kernel of this operator.

\begin{lem} \label{SMOOTH}

 For any $g \in C^\infty_0(\R)$ the distribution \begin{equation} \label{GDEF} G=\int_\R g(s) e^{\rmi  s D_\theta}   \psi(D_Z- \nu D_{\theta})  \wt{E} \dd s \end{equation} is smooth. Hence, the operator $G$ is trace-class.
 \end{lem}
 \begin{proof} This is proved by showing that the wave front set $\mathrm{WF}(G)$ of $G$ is empty. 
 More detailed wavefront set computations are performed in the next section and we therefore only sketch the argument.
The wavefront set of the distributional kernel of $\wt{E}$ only contains non-zero lightlike covectors of the form $(\zeta,\zeta')$ such that $\zeta$ and $-\zeta'$ are on the same orbit of the geodesic flow on $T^* \wt M$ (see for example \cite[Theorem 3.2]{SZ18}).  One then computes directly (see for example the proof of Lemma \ref{lemma81}) that
the wavefront set of $G$ is contained in the set of non-zero lightlike covectors $(\zeta,\zeta')$ with $p_\nu(\zeta)=0$, $p_\theta(\zeta)=0$ such that there exists $t,s \in \R$ with $G^{t} \exp(s \partial_{\theta}) \zeta = -\zeta'$. The equations $p_\nu(\zeta)=0, p_\theta(\zeta)=0$ imply $\zeta =0$ for lightlike covectors.  \end{proof}

\begin{theo}   \label{TRACEFORMULA} As a distribution on $\R$, 
\begin{equation}\label{FORMULAE}  \begin{array}{lll} \Upsilon^{(1)}_{\nu, \psi}(s)  & = &
   \int_{\R}  \int_{\Sigma \times S^1}   \hat{\psi}(s')   (\dd_{\tilde x}\wt{E}_{s', s - \nu s'} (\tilde x, \tilde y)|_{\tilde y=\tilde x} - \dd_{\tilde y} \wt{E}_{s', s - \nu s'} (\tilde x,\tilde y) |_{\tilde x = \tilde y} \;\mathrm{dVol}_{\wt{\Sigma}} \dd s'
   % \\&& \\ & = & \int_{\R} \hat{\psi}(t) \Tr(U(s, t - \nu s) \dd s 
    . \end{array} \end{equation}
    Equivalently,
    \begin{equation} \label{G2} \Upsilon^{(1)}_{\nu, \psi}(s)  = \int_{\wt \Sigma}   e^{\rmi  s D_\theta}   \{\nu_{\tilde x} \psi(D_Z- \nu D_{\theta})  \wt{E}(\tilde x,\tilde y) - \nu_{\tilde y}
    \psi(D_Z- \nu D_{\theta})  \wt{E}(\tilde x,\tilde y)\}|_{\tilde x = \tilde y} \;\mathrm{dVol}_{\wt{\Sigma}} \end{equation}
    
  More precisely, for any $g \in C^\infty_0(\R)$,

 $$
   \int g(s) \Upsilon^{(1)}_{\nu, \psi}(s) \dd s = \int_{\wt \Sigma}  \left(  \nu_{\tilde x} G(\tilde x,\tilde y) -  \nu_{\tilde y} G(\tilde x,\tilde y) \right)\vert_{\tilde y=\tilde x} \;\mathrm{dVol}_{\wt{\Sigma}} .  
   $$
    The integrals are independent of the Cauchy surface chosen.
\end{theo}
Here, $\mathrm{dVol}_{\wt{\Sigma}} $ is the Lorentzian surface measure.  In case $g$ is even and $\psi$ is even this simplifies to 
  $$
   \int g(s) \Upsilon^{(1)}_{\nu, \psi}(s) \dd s = 2\int_{\tilde \Sigma}  \left(  \nu_{\tilde x} G(\tilde x,\tilde y) \right)\vert_{\tilde y=\tilde x} \;\mathrm{dVol}_{\wt{\Sigma}}(\tilde x).  
  $$

\begin{proof}
 The trace formula follows from Theorem \ref{UphiFORMb}, in particular the simpler form \eqref{TRUtDEF2}. To obtain \eqref{FORMULAE} from 
 \eqref{TRUtDEF2}, it suffices to observe that the trace in \eqref{FORMULAE} replaces $(s,s')$ by $\alpha(s',s) = (s', s - \nu s')$ and then integrates
 in $s'$ against $\hat{\psi}(s')$.
\end{proof}

We now give a second proof that does not use \cite[Theorem 4.1]{SZ18} or Theorem \ref{UphiFORMb} explicitly.

\begin{proof} We introduce the $\Sigma$-dependent  ladder projector: 

    \begin{equation} \label{Pinupsiform2} \Pi^{\wt{\Sigma}}_{\nu, \psi} :=    \int_{\R} \hat{\psi}(s) \rcal e^{-\rmi s  D_{\theta}} \circ e^{\rmi s \nu D_Z } \ecal  \dd s
    : H^1(\wt{\Sigma}) \oplus L^2(\wt{\Sigma}) \to  H^1(\wt{\Sigma}) \oplus L^2(\wt{\Sigma}) . \end{equation}
  as an operator on the Cauchy data space $H^1(\wt{\Sigma}) \oplus L^2(\wt{\Sigma})$.   
     
        Note that under the
 extension-restriction identification in Definition \ref{EDEF}-Definition \ref{RESDEF}, the $S^1$ action acts by rotations on 
 $H^1(\wt{\Sigma}) \oplus L^2(\wt{\Sigma})$.
 Although \eqref{Pinupsiform2}  depends on the choice of $\wt{\Sigma},$ the trace formula it induces does not.

By \eqref{INVERSE}, and by  \eqref{Pinupsiform}-\eqref{Pinupsiform2}, $\Upsilon^{(1)}_{\nu, \psi}(s)  = \Tr_{\hcal_{KG}} \Pi_{\nu, \psi} e^{\rmi s D_\theta}$
is given by,

\begin{equation}\label{FORMULAE2}  \begin{array}{lll} \Upsilon^{(1)}_{\nu, \psi}(s)  & = &  
 \Tr_{\wt{\Sigma}} \Pi^{\wt{\Sigma}}_{\nu, \psi} e^{\rmi s D_\theta} \\ &&\\ & = & 
\int\limits_{\R}  \int\limits_{\tilde \Sigma} \hat{\psi}(s') *(\dd_{\tilde x} \wt{E}(\tilde x, e^{s \frac{\partial}{\partial \theta}} \circ e^{ s' ( Z - \nu\frac{\partial}{\partial
\theta })}\tilde y )  - \dd_{\tilde y} \wt{E}(\tilde x, e^{s \frac{\partial}{\partial \theta}} \circ e^{ s' ( Z - \nu\frac{\partial}{\partial
\theta })} \tilde y )|_{\tilde x = \tilde y} \dd s' \; \mathrm{dVol}_{\wt{\Sigma}}(\tilde x)
\\ &&\\
%& = &
% \int_{\R}  \int_{\Sigma \times S^1}   \hat{\psi}(s)   \wt{E}({\tilde x}, e^{s Z} e^{i(t - \nu s) } {\tilde x})  \dd \theta \dd S_{\Sigma} \dd s \\ &&\\
& = &
   \int\limits_{\R}  \int\limits_{\tilde \Sigma}   \hat{\psi}(s')   *(\dd_{\tilde x} \wt{E}_{s', s - \nu s'} (\tilde x,\tilde y) - \dd_{\tilde y}  \wt{E}_{s', s - \nu s'} (\tilde x,\tilde y)) |_{\tilde x = \tilde y}  \dd s' \;\mathrm{dVol}_{\wt{\Sigma}}(\tilde x)
   % \\&& \\ & = & \int_{\R} \hat{\psi}(t) \Tr(U(s, t - \nu s) \dd s \mathrm{dVol}_{\wt{\Sigma}}
    . \end{array} \end{equation}
    In the first line, $\Tr_{\wt{\Sigma}}$ denotes the trace on the finite energy  Cauchy data space. The first line holds
    because  $\rcal_{\wt{\Box}}$ is an  invertible 
    operator (see Definition \ref{RESDEF}) conjugating $\Psi_{\nu, \psi}$ on $\hcal_{KG}$ and  \eqref{Pinupsiform2}. 
    \end{proof}
    
    \section{\label{MICROSECT} Microlocal calculation of the singularity of $\Upsilon^{(1)}_{\nu, \psi}(s)$}

Our goal now is to determine the singular support of   $\Upsilon^{(1)}_{\nu, \psi}(s)$, and to calculate the coefficient of the leading singularity at $s = 0$.  We use the symbol calculus of Fourier
integral operators as in \cite{DG75,GU89} in   the calculation.
By Theorem \ref{TRACEFORMULA}, the distribution trace may be expressed in several ways as a composition of simpler Fourier
integral operators. We consider two ways to group the compositions, each of some independent interest in related problems.

The first is to use \eqref{G2} and to consider each term of the  trace  as the composition 
$$\left(\Pi_{\wt{\Sigma}}\right)_* \Delta^*_{\wt{\Sigma}}  R_{\wt{\Sigma} \times \wt{\Sigma}}  \nu_{\tilde{x}}
e^{\rmi s D_\theta} \psi(D_Z - \nu D_\theta)  \wt{E}  $$ of the following Fourier integral operators:
\begin{itemize}
\item $ \nu_{\tilde{x}}$, the unit normal vector field extended to a neighborhood of the Cauchy surface.
\item $\wt{E}$.
\item $\psi(D_Z - \nu D_\theta) \wt{E} \in \dcal'(\wt M \times \wt M). $ 
\item $e^{\rmi s D_\theta}$, the translation operator on $S^1$ by $e^{\rmi s}$.
\item $ R_{\wt{\Sigma} \times \wt{\Sigma}}: \wt{M} \times \wt{M} \to \wt{\Sigma} \times \wt{\Sigma}$, the restriction operator, i.e.
pullback under the embedding $ \wt{\Sigma} \times \wt{\Sigma} \to \wt{M} \times \wt{M}$.
\item $\Delta_{\wt{\Sigma}} : \wt{\Sigma} \to \wt{\Sigma} \times \wt{\Sigma}, \Delta_{\wt{\Sigma}}({\tilde x}) = ({\tilde x}, {\tilde x}),$ pullback under  the diagonal
embedding. 

\item $\Pi_{\wt{\Sigma} }: \R  \times \wt{\Sigma} \to \R \times \R, \Pi_{\wt{\Sigma}}(t, {\tilde x})= t, $ i.e. the pushforward is integration over $\wt \Sigma$.

\end{itemize}

A second grouping is to compose in the order: \begin{equation} 
\label{TRUpsilon2}\Upsilon^{(1)}_{\nu, \psi}(t) = \rho_* \hat{\psi}\left(\Pi_{\wt{\Sigma}}\right)_* \alpha^* \Delta_{\wt{\Sigma}}^* R_{\wt{\Sigma} \times \wt{\Sigma}}\nu_{\tilde{x}}  \wt{E}_{t,s},
\end{equation}  with a sequence of pullbacks and pushforwards under the maps defined by:
\begin{itemize} 

\item $\wt{E}_{t,s}({\tilde x}, {\tilde y})
=  U(t,s) \wt{E}({\tilde x}, {\tilde y})$

\item $\alpha(s,t) = (s, t - \nu s).$

\item $\rho(t,s) = t$.

\item $\Pi_{\wt{\Sigma} }: \R \times \R \times \wt{\Sigma} \to \R \times \R, \Pi_{\wt{\Sigma}}(t, s, {\tilde x}) = (t,s). $ 

\end{itemize}

The first method emphasizes the fundamental operator $\psi(D_Z - \nu D_\theta) \wt{E}$, and so we follow this approach. The second 
is useful if one would to study the singularities of the two-variable trace $\Tr_{\hcal_{KG}} U(t,s) \wt{E}$ first. 

\subsection{\label{FIOSECT} Background on Fourier integral operators and their symbols}
The advantage of expressing $\Upsilon_{\nu, \psi}^{(1)}(s)$ in terms of pullback and pushforward is
the symbol calculus of Lagrangian distributions is more elementary to describe for such compositions. 
We refer to \cite{HoIV,GS77,D96} for background but quickly review the basic definitions. Assume that $f: X \to Y$ is a smooth map between manifolds.
 Let  $\Lambda \subset T^* Y$ be a Lagrangian submanifold. Then its pullback is defined by, 
\begin{equation} \label{PB} f^* \Lambda = \{(x, \xi) \in T^*X \mid \exists (y, \eta) \in \Lambda, f(x) = y, f^* \eta = \xi\}. \end{equation}
On the other hand, let $\Lambda \subset T^*X$. Then its pushforward is defined by,
\begin{equation} \label{PF} f_* \Lambda = \{(y, \eta) \in T^* Y \mid y = f(x), (x, f^* \eta) \in \Lambda\}. \end{equation}

The principal symbol of a Fourier integral operator associated to a canonical relation $C$ is a half-density times a section of the  Maslov line bundle  on $C$.
We refer to \cite[Section 25.2]{HoIV} and to \cite[Definition 4.1.1]{D96} for the definition; see also \cite{DG75, GU89} for further expositions and for
several calculations of principal symbols closely related to those of this article.

The \emph{order} of a homogeneous  Fourier integral operator\index{Fourier integral operator!order} $A \colon L^2(X) \to L^2(Y)$ in the
non-degenerate case  is given in terms of a local oscillatory
integral formula $$K_A(x,y) = \frac{1}{(2 \pi)^{n/4 + N/2}}\int_{\R^N} e^{\rmi \phi(x, y, \theta)} a(x, y, \theta) \dd \theta $$by \begin{equation} \label{ORDDEF} \ord A = m +
 \frac{N}{2} - \frac{n}{4}, \;\; \mathrm{where} \; n = \dim X + \dim Y, \; m = \ord \;a \end{equation} where the order  of the amplitude $a(x, y, \theta)$
 is the degree of the top order term of the polyhomogeneous expansion of $a$ in $\theta$, and $N$ is the number of phase
variables $\theta$ in the local Fourier integral representation (see
\cite[Proposition~25.1.5]{HoIV}); in the general clean case with
excess $e$, the order goes up by $\frac{e}{2}$ (see \cite[Proposition~25.1.5']{HoIV}
). The order is designed to be independent of the specific representation of $K_A$ as an oscillatory integral. The space of
Fourier integral operators of order $\mu$ associated to a canonical relation $C$ is denoted by $$K_A \in I^{\mu}(Y \times X, C'). $$
If $A_1 \in I^{\mu_1}(X \times Y, C_1'), A_2 \in   I^{\mu_2}(Y \times Z, C_2')$, and if $C_1 \circ C_2$ is a `clean' composition,
 then by \cite[Theorem 25.2.3]{HoIV},
$$A_1 \circ A_2 \in I^{\mu_1 + \mu_2 + e/2} (X \times Z, C'), \;\; C = C_1 \circ C_2,$$
where $e$ is the `excess' of the composition, i.e. if $\gamma \in C$, then $e =\dim C_{\gamma}$, the dimension of the fiber of $C_1 \times C_2
\cap T^* X \times \Delta_{T^*Y} \times T^*Z$ over $\gamma$.
 
Pullback and pushforward of half-densities on Lagrangian submanifolds are more difficult to describe. 
They depend on the map $f$ being a {\it morphism} in the language of \cite[page 349]{GS77}. Namely,
if $f: X \to Y$ is a smooth map, we say it is a morphism on half-densities if it is augmented by a section
$r(x) \in \mathrm{Hom}(|\Lambda|^{\half}( TY_{f(x)}, |\Lambda|^{\half} T_x X)$, that is, a linear transformation 
mapping densities on $TY_{f(x)}$ to densities on $T_x X$. As pointed out in \cite[page 349]{GS77}, 
such a map is equivalent to augmenting $f$ with a special kind  of half-density on the co-normal bundle $N^*(\mathrm{graph}(f))$ to the graph of $f$, which is constant along the fibers of the co-normal bundle. In our application, the maps are all restriction maps or pushforwards under canonical maps, and they are morphisms
in quite obvious ways.  Note that under pullback by an immersion,  or  under a restriction, the number $n$ of independent variables is decreased  by the codimension $k$ and therefore
the order goes up by $\frac{k}{4}.$ Pullbacks under submersions increase the number $n$.  Pushforward is adjoint to pullback and therefore also decreases the order by the same amount.

% of $$\alpha^* R_{\wt{\Sigma}} \circ \Delta^*  \circ d_x E_{t,s}(x,y) = \wt{E}_{t + \nu s, s} ({\tilde x}, {\tilde x})  |_{\Sigma \times S^1},$$ 
%and
%then of its pushforward $$   \int_{\R}  \int_{\Sigma \times S^1}   \hat{\psi}(s)   \wt{E}_{t + \nu s, s} ({\tilde x}, {\tilde x})  \dd \theta \dd S_{\Sigma} \dd s = %\rho_{s*} \hat{\psi}  \alpha^* \wt{E}_{t,s}  |_{\Sigma \times S^1} .$$
%By \eqref{FORMULAE} this determines the principal symbol of the ladder
%trace.

As examples relevant to this article, we record some basic calculations of canonical relations and orders in \cite{DG75} in the setting of half-wave
groups $e^{\rmi t Q}$ generated by positive elliptic first order pseudo-differential operators $Q$ on  a Riemannian manifold $X$.
Let $\Delta: X \to X \times X$ be the diagonal embedding. The corresponding pullback operator, the operator of restriction to the diagonal, $\Delta^*$ is a Fourier integral operator of order $\frac{n}{4}$ where $n = \dim X$, and has canonical relation, 
$$WF'(\Delta^*) = \{( (x, \xi - \eta),  (x, \xi), (y, \eta)) \in (T^* X \times T^*X \times T^*X )\setminus 0\} .$$
Moreover,  the wave group $U(t) = e^{\rmi t Q}$ of any positive elliptic operator is a Fourier integral operator of order $-\frac{1}{4}$ associated to the space-time
graph of the bi-characteristic flow  $\Phi^t$, i.e.  Hamilton flow of the principal symbol $\sigma_Q$ of $Q$. It follows that $\Delta^* U$ is a Fourier integral
distribution on $\R \times X$ with canonical relation,
$$WF'(\Delta^*(U)) = \{(((t, \tau) , (x, \xi - \eta)), \tau + |\xi|_g = 0, (x, \xi) = G^t(x, \eta) \}.$$
 Moreover, if  $\Pi: \R \times X \to \R$ denotes  the natural projection, then the associated push-forward  $\Pi_*$ has order $\half - \frac{n}{4} $ when $n = \dim X$, and 
$$WF'(\Pi_*) = \{(((t, \tau), ((t, \tau), (x, 0)))\}. $$
The  wave trace  $\hat{\sigma}(t) = \Tr e^{\rmi t Q}$ is  a Fourier integral distribution on $\R$ of order $-\frac{1}{4} + \frac{d}{2} = - \frac{1}{4} + \frac{2n -1}{2}$ where $d = \dim S^*X$, and 
$$WF'(\hat{\sigma}(t)) = \{(t, \tau)\mid \tau < 0, \exists (x, \xi), (x, \xi) = \Phi^t(x, \xi), \tau + |\xi|_g =0 \}.$$
The principal symbol of $\Tr e^{\rmi t Q}$ at $t =0$ on $T^*_0 \R$ is a half-density on the negative sub-cone and   is calculated in \cite[Proposition 2.1]{DG75}
and equals $C_n \mathrm{Vol}(S^*M) |d \tau|^{\half}$ where $S^*M = \{\sigma_Q = 1\}$. Homogeneous Fourier integral distributions on $\R$ of  order $\frac{2n-1}{2} - \frac{1}{4}$ that are microlocally supported on $\{0\} \times \R_-$ have the form $\int_0^{\infty} s^{\frac{d-1}{2}} e^{\rmi ts} \dd s $, where $d = 2n-1$. 
%These results hold for the case $Q = D_Z^{\wt{\Sigma}}$
%on $\wt{\Sigma}$ as do the calculations of Taylor-Uribe \cite{TU92}, completing the proof
%of (A). 

\section{Proof of Theorem \ref{LADDERSING1}}

We now prove Theorem \ref{LADDERSING1}.  In keeping track of the orders (and universal dimensional constants)  of the many compositions to follow, it is helpful to keep
in mind that they all must agree with the product case of Section \ref{PRODUCT}. Although we provide details on orders and symbols of the
partial compositions en route to the final result, it is not really necessary to do so because we only need to know the order of the final composition,
and we have computed it in the product case in Section \ref{PRODUCT}. For this reason,  we only 
give full details on the canonical relations, and do not always give full details on orders.
{ To simplify statements when dealing with principal symbols given with respect to parametrizations we introduce the following notation. If $\iota: M \to N$ is a local diffeomorphism and $\rho$ a compactly supported half-density on $M$ we define the {\sl push-forward} $\iota_* \rho$ by  $\iota_* \rho(y) = \sum_\kappa ((\iota_\kappa)^{-1})^*\rho(x_\kappa)$, where $\{x_\kappa\}$ is the set of elements in $\iota^{-1}(y) \cap \supp \rho$ and $(\iota_\kappa)^{-1}$ are local inverses of $\iota$ defined near $x_\kappa$. This sum is always finite and defines a half-density on $N$.
}

  We will need the following calculation from \cite[Section 8]{SZ18} in all the computations below.  The geodesic vector field, i.e. the Hamiltonian vector field
of $\frac{1}{2} g^{-1}$, defines a local flow on the null covectors   $T^*_{0,\pm} \wt M$ and the space of orbits is the symplectic manifold $\widetilde{\ncal }$. The flow parameters $s_1$, $s_2$
and the symplectic volume form $\mathrm{dV}_{\wt \ncal}=\frac{1}{d!} \omega^d$  on $\wt {\mathcal{N}}$ then define a half-density,
\begin{equation} 
 \label{dC} |\dd_C|^{\frac{1}{2}}=|\dd s_1|^\frac{1}{2} \otimes |\dd s_2|^\frac{1}{2}  \otimes |\mathrm{dV}_{\wt{\ncal}}|^\frac{1}{2} =  |\dd s_1|^\frac{1}{2} \otimes  | \dd\mu_{T^*_0 \wt M})|^\half, 
 \end{equation}
where $\dd\mu_{T^*_0 \wt M}$ is  the Liouville measure on $T^*_0 \wt M$.
We record $$\dim M = n, \dim T^* M = 2 n, \dim T^*M_{0,\pm}= 2n-1,
\dim C = 2 n, \dim \wt{\ncal} = 2 n- 2,  \textrm{ and } d=n-1.$$
The  operator with kernel $\dd_x \wt E_{t,s}(x,y)$ has order $-1$ and principal symbol equal to $$\pm\frac{1}{2} (2 \pi)(\xi \wedge) |\dd s|^{\half} \otimes |\dd t|^{\half} \otimes |\dd_C|^{\half} 
%= \pi 
%\left(\nu_x E_t(x,y) -  \nu_y E_t(x,y)   \right)\;\mathrm{dV}_\Sigma(x)
,$$ on each component $C_+$ and $C_-$, where $(\xi \wedge)$ denotes the operator of exterior multiplication by $\xi$ ($d_x \wt E_t(x,y)$ is a one form).  The proof is essentially the same as for \cite[Lemma 8.2]{SZ18}.
The principal symbol of $\dd_{\tilde x} (\wt E_{t,s}(\tilde x,\tilde y) + \wt E_{-t,-s}(\tilde y,\tilde x) ) |_{\tilde x= \tilde y} $ is easily computed using 
\eqref{SYMBOLFULL}.

\begin{rem} \label{nuxREM} When computing the  $*\dd_{\tilde x} \wt E(\tilde x,\tilde y)$ and restricting to $\wt \Sigma$, we often use below the general  
formula that $\iota_{\Sigma}^* (* d_x E) = \nu_x E dV_{\wt \Sigma}$, where $\iota: \Sigma \to M$ is the embedding. \end{rem}

  The first step in proving Theorem \ref{LADDERSING1} in the first approach of Section \ref{MICROSECT}  is to prove,
  \begin{lem}\label{lemma81} Let $\nu$ be admissible.  Then,  $\psi(D_Z - \nu D_\theta) \wt{E}$ is a Fourier integral operator, $$ \psi(D_Z - \nu D_\theta) \wt{E} \in  I^{-2} (\wt M \times \wt M, \wt{\lcal}_{\nu}), $$  with canonical relation,
   
   $$\wt{\lcal}_{\nu} : = \{(\zeta', \zeta) \in \dot{T}^* \wt M \times \dot{T}^*\wt M:\; \exists t \in \supp \hat{\psi}, s \in \R:  \wt{G}^{s}  \exp t H_{p_{\nu}} (\zeta) =  \zeta',  
   p_{\nu}(\zeta) = \tilde g(\zeta, \zeta) = 0\}. $$  This canonical relation is parametrized by  { the local diffeomorphism}
    $$(t, s, \zeta) \in \R \times \R \times \mathrm{Char}(\wt{\Box}, P_{\nu}) \to (\wt{G}^{s}  \exp t H_{p_{\nu}} (\zeta), \zeta)$$
   and its principal symbol is the push-forward of the half-density $(-\frac{\rmi}{2} ) \hat{\psi}(t) |\dd t|^{\half} \otimes |\dd s|^{\half} \otimes |\mu_{D\mathrm{Char}_\nu}|^{\half}, $ where $\mathrm{Char}(\half \wt{\Box}, P_{\nu}) =D\mathrm{Char}_\nu$
   is the double-characteristic variety of Definition \ref{DCHAR}, and $\mu_{D\mathrm{Char}_{\nu}}$ is the natural
   Liouville surface measure on this variety.  \end{lem}
   
   \begin{proof}    

 By \cite[Proposition 2.1]{TU92},   $\psi(D_Z - \nu D_\theta) \in I^{-\half}(\wt M \times \wt M, \ical), $ where
   \begin{equation} \label{TU} \ical = \{(\zeta_1, \zeta_2) \in  \dot{T}^* \wt{M} \times \dot{T}^*\wt{M} \mid p_{\nu}(\zeta_1) = 0, { \exists  t \in \supp \hat\psi}, 
   \exp t H_{p_{\nu}} (\zeta_2) = \zeta_1\} \end{equation} 
   By \cite[Lemma 2.6]{TU92}, the principal symbol of $\psi(D_Z - \nu D_\theta)$ is $(2 \pi)^{-\frac{1}{2}} \hat{\psi}(t) |\dd t|^{\half} \otimes |\dd \mu_L|^{\half} $.
     We then compose this canonical relation with that of $\wt{E}$, given in Theorem \ref{DHTH}, to obtain $\wt{\lcal}_{\nu}$.
   Further, we compose the principal symbol of $\nu_x \wt{E}$ (also computed from Theorem \ref{DHTH}) with the principal symbol of  $\psi(D_Z - \nu D_\theta)$ to obtain
   the stated order and principal symbol.
 \end{proof}
%  \begin{rem}
%   The parametrization of the canonical relation defines $\wt{\lcal}_{\nu}$ as an immersed Lagrangian submanifold of $\dot T^* \wt{M} \times \dot T^* \wt{M}$. Different parameter values may however map to the same point. The formula for the principal symbol has to be understood in this sense, as a function of the parameters. The principal symbol as a function on the image in  $\dot T^* \wt{M} \times \dot T^* \wt{M}$ can be obtained by summation. Similar formulae for the principal symbol below all need to be understood in that sense.
%      \end{rem}
 
      Next we compose with $e^{\rmi tD_\theta}$ to obtain,
      \begin{lem} \label{lemma82}
      $  \nu_{\tilde{x}} e^{\rmi tD_\theta} \psi(D_Z - \nu D_\theta) \wt{E}$ is a Fourier integral operator, $$  \nu_{\tilde{x}}e^{\rmi tD_\theta} \psi(D_Z - \nu D_\theta) \wt{E} \in I^{- \frac{5}{4}} (\R \times \wt M \times \wt M, \wt{\Gamma}_{\nu}), $$  with canonical relation,
   $$ \wt{\Gamma}_{\nu} : = \begin{array}{l}    \{(t, p_\theta(\zeta), \zeta', \zeta) \in T^*\R \times \dot{T}^* \wt M \times \dot{T}^*\wt M \mid \;  \exists t' \in \supp \hat{\psi},\exists s, t \in \R, \\ e^{t\frac{\partial}{\partial \theta}} \wt{G}^{s}  \exp t' H_{p_{\nu}} (\zeta) =   \zeta',  
   p_{\nu}(\zeta) = \tilde g(\zeta, \zeta) = 0\} \end{array}.$$ Its principal symbol is the push-forward of 
   $$ \sigma_{ \nu_{\tilde{x}} e^{\rmi tD_\theta}  \psi(D_Z - \nu D_\theta)  \wt{E}} = | \langle \xi, \nu_x \rangle|  (-\frac{\rmi}{2} (2 \pi)^\frac{1}{4})\hat{\psi}(t')   |\dd t|^{\half}  \otimes |\dd t'|^{\half} 
   \otimes  |\dd s|^{\half}  \otimes |\mu_{D\mathrm{Char}_\nu}|^{\half}. $$

\end{lem}

Here, we use Remark \ref{nuxREM},  the composition formula, and that $e^{\rmi t D_\theta}$  introduces the new variable $t$ and is of order $-\frac{1}{4}$.

   We then compose with $\rcal_{\wt{\Sigma} \times \wt{\Sigma}}$, the pullback under the embedding {
   $\iota_{\wt{\Sigma}} \times \iota_{\wt{\Sigma}}=\iota_{\wt{\Sigma} \times \wt{\Sigma}}: \wt{\Sigma} \times \wt{\Sigma} \to \wt{M} \times \wt{M}$,} and use \eqref{PB} to calculate that
   the canonical relation of $\rcal_{\wt{\Sigma} \times \wt{\Sigma}}  \nu_{\tilde{x}} e^{\rmi tD_\theta} \psi(D_Z - \nu D_\theta)  \wt{E}$
   is $$
    \iota_{\wt{\Sigma} \times \wt{\Sigma}}^* \Gamma_{\nu} \circ \wt{\ccal}    = \begin{array}{ll}  \{(t, p_\theta(q, \xi), q, \xi, q', \xi') \mid (q, \xi, q', \xi') \in \dot{T}^* \wt{\Sigma} \times \dot{T}^* \wt{\Sigma}, { \exists t' \in \supp \hat \psi}, s \in \R,  \exists (\zeta, \zeta') \\\\  \in T^*_q\wt{M} \times T^*_{q'} \wt{M}, \zeta |_{T_q \wt{\Sigma}} = \xi,  \zeta' |_{T_{q'} \wt{\Sigma}} = \xi', 
   p_{\nu}(\zeta) = g(\zeta, \zeta) = 0,  e^{t \frac{\partial}{\partial \theta}} \wt{G}^{s}  \exp t' H_{p_{\nu}} (\zeta)= \zeta')\}.  \end{array}
   $$
   We note that $p_\theta(q, \xi) = p_\theta(\zeta)$, where $\zeta |_{T \wt \Sigma} = \xi$, since $\wt \Sigma =  \Sigma \times S^1 $ and the restriction does not touch the $S^1$ factor.  
      
 The next step is to compose with the diagonal embedding $\Delta_{\wt{\Sigma}}$. Using its canonical relation in  \ref{FIOSECT}, the canonical relation of
 $$\Delta_{\wt{\Sigma}}^* \rcal_{\wt{\Sigma} \times \wt{\Sigma}} \nu_{\tilde{x}} e^{\rmi t D_\theta} \psi(D_Z - \nu D_\theta)    \wt{E}(x,y)$$
   is 
   { 
   $$\Delta_{\wt{\Sigma}}^*  \iota_{\wt{\Sigma} \times \wt{\Sigma}}^* \wt{ \Gamma}_{\nu}   = \begin{array}{ll}  \{(t, p_\theta(q, \xi), q, \xi - \xi') \mid (q, \xi),  (q, \xi') \in \dot{T}^* \wt{\Sigma} \times \dot{T}^* \wt{\Sigma}: \exists t' \in \supp \hat{\psi}, \\ \\  \exists (\zeta, \zeta')  \in \dot{T}^*_q\wt{M} \times \dot{T}^*_{q} \wt{M},  \zeta |_{\dot{T}_q \wt{\Sigma}} = \xi,  \zeta' |_{T_{q} \wt{\Sigma}} = \xi', 
   p_{\nu}(\zeta) = \tilde g(\zeta, \zeta) = 0, \\ \\ e^{t \frac{\partial}{\partial\theta}}  \wt{G}^{s}  \exp t' H_{p_{\nu}} (\zeta)= \zeta')\}.  \end{array}$$}
The final step is to
   compose with $(\Pi_{\wt{\Sigma} })_*$ (integration over $\wt{\Sigma}$) and use the wave front relation
   for this operator in Section \ref{FIOSECT} to get that
   {
   $$(\Pi_{\wt{\Sigma}})_* \Delta_{\wt{\Sigma}}^*  \iota_{\wt{\Sigma} \times \wt{\Sigma}}^* \wt{\Gamma}_{\nu}     = \begin{array}{ll}  \{(t, \tau) \mid \exists t' \in \supp
   \hat{\psi}, \exists (q, \xi)\in \dot{T}^* \wt{\Sigma}:  \exists \zeta  \in \dot{T}^*_q\wt{M}, \zeta |_{\dot{T}_q \wt{\Sigma}} = \xi,  \tau = p_\theta(\zeta),\\ \\   
   p_{\nu}(\zeta) = \tilde g(\zeta, \zeta) = 0, e^{t \frac{\partial}{\partial \theta}}  \wt{G}^{s}  \exp t' H_{p_{\nu}} (\zeta)= \zeta)\}.  \end{array}$$}
   
   We re-write the fixed point equation as $e^{(t - \nu t') \frac{\partial}{\partial \theta}} e^{t' Z} G^{s}(\zeta) = \zeta. $ Since $e^{(t - \nu t') \frac{\partial}{\partial \theta}}, e^{t' Z} $ act on different
   components of $\zeta$ and $\wt G^{s} = \exp s H_{\sigma^2} \circ G^{s}$ is a product flow, the fixed point equation splits into two
   equations:
   \begin{itemize} \item  $\exists  \zeta \in \ncal_1\cap \{p_Z  = \nu \sigma\}$ so that  $e^{t' Z} (\zeta) = \zeta$, i.e. 
    $t'$  lies in the set $\pcal_{\nu}$ of  periods of $e^{t' Z}$ on $\ncal_1 \cap \{p_Z  = \nu \sigma\}$. 
   
   \item $ t - \nu t'  \in 2 \pi \Z$ for some $t' \in \pcal_{\nu} { \cap \supp \hat \psi}$. In particular, this includes $t = t' = 0$.
   \end{itemize}
   This proves the statement $$ \mathrm{sing-supp} \Upsilon^{(1)}(t) \subset \{t \in [0, 2 \pi]:  \exists t' \in \pcal_{\nu} { \cap \supp \hat \psi}, \;t - \nu t'  \in 2 \pi \Z\} $$  in Theorem \ref{LADDERSING1}.
   
   \begin{rem} If we change the circle $S^1$ from $\R/2 \pi \Z$ to $\R/ T \Z$, then the second condition
   becomes  $ t - \nu t'  \in \frac{2 \pi}{T} \Z$ for some $t' \in \pcal_{\nu}$. \end{rem}

   { Next we compute the singularity at $t=0$ under the  assumption that
  $\supp \hat \psi \in (-\pi/\nu,\pi/\nu)$. Then the equation $e^{t \frac{\partial}{\partial \theta}} \wt{G}^{s}  \exp t' H_{p_{\nu}} (\zeta)= \zeta'$ f implies that $t - \nu t' =0$. Since $\tilde \Sigma$ is a Cauchy surface the fixed point equation determines $s$ uniquely as a function of $t'$ as long as $t$ is sufficiently close to zero.}
       
    As explained in \cite[Section 1.1]{SZ18},  the restriction of a null  covector to $\wt \Sigma$ defines a covector in $T^*\wt \Sigma \setminus 0$, 
 defining  a { smooth maps $\wt \ncal_\pm$} to $T^*\wt \Sigma \setminus 0$, since  for each element $\eta \in \dot{T}^* \wt \Sigma$ there is precisely { one} lightlike future/past directed covector $\xi \in T^*\wt M \setminus 0$ 
 whose pull-back is $\eta$. The Hamiltonian $p_{\nu} $ \eqref{pnudef} also restricts to $T^* \wt{\Sigma}$
 and thus defines a codimension one conic submanifold we denote by 
 $$p_{\nu, \wt \Sigma}^{-1}(0) = \{(q, \xi) \in T^* \wt{\Sigma} \mid p_{\nu}(q, \xi) =0\}. $$

 \begin{rem}  Note that $Z$ is transverse to $\wt{\Sigma}$; in the special case of static spacetimes, one may define
 $\Sigma$ to a hypersurface normal to $Z$. In that case, the restriction $p_{\nu} |_{T^* \wt{\Sigma}}$ would be equal $\nu \sigma  |_{T^* \wt{\Sigma}}$ and its zero set would be $T^* \Sigma$.  \end{rem}

      The order of the singularity is given by $\frac{e}{2}$, where is  the excess $e$ of the composition
      defining the trace. As in \cite[Theorem 4.5]{DG75}, if the dimension of the fixed point in the level set $\ncal_1(\nu)$ is $d$, then 
      the principal  singularity is the homogeneous distribution $\tilde \mu_{\frac{d-1}{2}} = \int_{\R}  s ^{\frac{d-1}{2}} e^{-i s t} \dd s$.  For $t = 0$,   the fixed point set all of $\ncal_1(\nu) : = \{(x, \xi) \in \ncal_1: \langle \xi, Z\rangle = \nu\} $, and has  dimension $d = 2n-3$. Hence the leading singularity is of type $\tilde \mu_{n - 2}$. Note that, as mentioned above,
      $\ncal_1 \simeq T^*\Sigma$ (of dimension $2n-2$).

      Finally, we must compute the symbol. As computed in \cite[Lemma 8.4 - Lemma 8.5]{SZ18}, the factor of $ | \langle \xi, \nu_x \rangle|$ in Lemma \ref{lemma82} is cancelled 
      when in the composition with $R_{\Sigma} \Delta^*$, and the  principal symbol is,
    $$\sigma_{\rcal_{\wt{\Sigma} } \Delta^*\nu_{\tilde{x}} e^{\rmi tD_\theta}  \psi(D_Z - \nu D_\theta)   \wt{E}} = (-\frac{\rmi}{2} (2 \pi)^{-\frac{1}{2}})\hat{\psi}(t')
 |\dd t|^{\half} \otimes |\dd t'|^{\half} \otimes  \iota_{\wt{\Sigma}}^* |\mu_{\{p_{\nu} = 0\} \cap T^*_{0, \wt \Sigma}} |^{\half}. $$
Here,  $\mu_{\{p_{\nu} = 0\} \cap T^*_{0, \wt \Sigma}}$ is the Liouville measure on the codimension
    two submanifold of $T^* \wt{\Sigma}$ of $T^* \wt M$  defined by the pullback of the defining
    functions $ \iota_{\wt{\Sigma}^{-1}}(\half g(\xi, \xi), p_{\nu}(\xi))$ of the double characteristic
    variety.

      It follows that 
      %In terms of the parametrization,    $$(s, t) \in \R_s \times \R_t \times \R_{t'}  \times T^* \Sigma \to (s, 0, t, \tau, (q, \xi), 
    %e^{s \frac{\partial}{\partial \theta}}  \wt{G}^{t'}  \exp t H_{p_{\nu}} (\zeta(q, \xi))) |_{T \wt{\Sigma}}, $$
   % (with identifications
  % of $(q, \xi) \to \zeta(q, \xi)): T^*\Sigma \simeq \ncal$ as above),
 the principal symbol of the pushforward to $\R$  on $T^*_0 \R$ is given by,
  
   $$\sigma_{\Upsilon^{(1)}_{\nu, \psi}}  |_{T^*_0 \R} =   (2 \pi)^{3/4-(n-1)} \hat{\psi}(0) \; \mu_{\ncal_1(\nu)} (\ncal_1(\nu))\; |\tau|^{n-2} |d \tau|^{\half}.$$
   This gives a leading singularity of the form
   $$
   \Upsilon^{(1)}_{\nu, \psi}(s) \sim (2 \pi)^{-n+1} \hat{\psi}(0) \; \mu_{\ncal_1(\nu)} (\ncal_1(\nu)) \tilde \mu_{n-2 }(s) + c_1  \tilde \mu_{n-3}(s) 
+ \ldots,
   $$
   where 
   $$
    \tilde \mu_k(s) = \int_{-\infty}^{\infty} \mathrm{e}^{ \rmi s \tau} |\tau|^{k} \mathrm{d} \tau.
   $$
   Since $\Upsilon^{(1)}_{+,\nu, \psi}$ can be constructed from $\Upsilon^{(1)}_{\nu, \psi}$ by projecting to the non-negative Fourier coefficients we obtain that
   $\Upsilon^{(1)}_{+,\nu, \psi}$ has leading singularity of the form 
   \begin{gather} \label{firstcoeff}
    (2 \pi)^{-n+1} \hat{\psi}(0) \; \mu_{\ncal_1(\nu)} (\ncal_1(\nu))  \mu_{n-2}(t)\;
    \end{gather} as stated in Theorem \ref{LADDERSING1}. 
   This completes the 
 proof of  Theorem \ref{LADDERSING1}. 
   
\section{Further Fourier integral operator calculations}

Although we do not need them to prove Theorem \ref{LADDERSING1}, we provide some further calculations of canonical relations and
principal symbols using the second approach to the composition in Section \ref{MICROSECT}, since they are also of interest and are closer to
the calculations in \cite{SZ18}.

   \subsection{$\Tr_{\hcal_{KG}} U(t,s)$ as a Fourier integral distribution} 
 
   The next Proposition is analogous to the calculations in \cite[Section 8]{SZ18}. In the following, we compose $\wt E$ with $U(t,s)$
   rather than with $\psi(D_Z - \nu D_\theta)$ and take the trace. { We will assume throughout that $s$ and $t$ are sufficiently close to zero.}

 \begin{prop} \label{TrUts} 

 The distributional trace  
\begin{equation} \label{TRUtDEF3}
 \Tr_{\hcal_{KG}} (U(t,s) \wt E) =  \int_{\Sigma \times S^1}  * \left( \dd_{{\tilde x}} (E_{t,s}({\tilde x}, {\tilde y}) + E_{-t, -s}({\tilde x},{\tilde y})) \right) |_{{\tilde y}={\tilde x}},
 \end{equation}
 is a Fourier integral distribution on $\R \times S^1$ of order $n-1$, 
 with wave front relation,
 $$\begin{array}{lll} \mathrm{WF}'(\Tr_{\hcal_{KG}} (U(t,s)) ) & \subseteq & \{(t, \tau, s, \sigma) \in T^*(\R \times S^1) \mid \exists \ell \in \R,({\tilde x}, {\tilde \xi}) \in T^*\wt M: \\&&\\ && 
 G^{\ell} \circ e^{t Z} e^{\rmi s D_{\theta}} ({\tilde x}, {\tilde \xi}) = ({\tilde x}, {\tilde \xi}),\; ({\tilde \xi}, {\tilde \xi}) =0,  \tau= p_Z({\tilde x}, {\tilde \xi}), \sigma = p_\theta ({\tilde x}, {\tilde \xi}) \} \\ &&\\
 & = &  \{(t, \tau, \theta, \sigma) \in T^*(\R \times S^1) \mid \exists [\gamma] \in \wt{\ncal}: \\ && \\ &&e^{t Z} e^{\rmi s D_{\theta}} [\gamma] = [\gamma], \tau = p_Z[\gamma],   \sigma = p_\theta[\gamma] \}  \end{array}$$
\end{prop}

\begin{rem}

Recalling the analogy with the trace $\hat{\sigma}$  of the wave group in \cite{DG75},  with  $\wt \ncal$ playing the role of $T^*X$, we see that the order
is consistent with the order $ - \frac{1}{4} + \frac{2n -1}{2}$ with $2n -1 = \dim S^* X$  of  $\hat{\sigma}$ mentioned in Section \ref{FIOSECT}. There
is an extra time variable by comparison with $\hat{\sigma}$ so the order should be $-\frac{1}{2} + \frac{\dim \wt \ncal -1}{2}.$

\end{rem}

\begin{proof} The proof is very similar to that of \cite[Theorem 1.4]{SZ18} for the case of $\Tr U(t)$ on $\hcal_{KG}$. The only difference is that we now consider
the two-parameter flow $U(s,t)$ \eqref{U(t,s)}.

We first describe $U(t,s)$ as a Fourier integral operator. Define the `moment Lagrangian', $$\Gamma = \{(t, \tau, \theta, \sigma, \zeta, -(f_s \circ g_t) (\zeta)), \zeta \in T^* \wt{M} \}, $$
where $\tau = p_Z(\zeta), \sigma = p_\theta(\zeta).$ 

\begin{lem}  \label{ULEM}

$U \in I^{-\half}(\R \times \R \times \wt{M} \times \wt{M}, \Gamma). $ 
\end{lem} The statement and proof are the same as in  \cite[Lemma 3.1]{GU89}, where  $U(t, s) = e^{\rmi (t P + s Q)} $. Here,
$P  = D_Z, Q = D_\theta$.

\begin{rem}
In the setting of compact Riemannian manifolds, the multi-variable trace formula was analyzed in \cite{CdV79}, and the analysis is similar in the present Lorentzian setting. \end{rem}

We next describe $\wt{E}_{t,s} = U(t,s) \wt E$ as a Fourier integral operator. The description is almost immediate from the Duistermaat - H\"ormander theorem
reviewed in  Theorem \ref{DHTH}. 
 
\begin{lem} \label{DHLEM} $\wt{E}_{t,s} \in I^{-2}(\R \times S^1 \times \wt{M} \times \wt{M}, \ccal)$, where $$\ccal = \{(t, p_Z(\zeta_1),  s, p_\theta(\zeta_1), (\zeta_1, \zeta_2)) \in T^*(\R \times S^1) \times \wt{C}, \;\;\;
e^{s \frac{\partial}{\partial \theta}}  e^{t Z}(\zeta_1) =  \zeta_2  \}. $$
Here,   $\wt{C} $ is defined in \eqref{CDEF}. The canonical relation $ \ccal$ is parametrized by 
$$\R_t \times S^1_s \times \R_w \times \mathrm{Char}(\wt{\Box})  \to \ccal, (t, s,w,  \zeta) \to (t, s, \zeta_1(Z),(e^{s \frac{\partial}{\partial \theta}} e^{t Z} \wt{G}^w(\zeta),\zeta).  $$
The   principal symbol pulls back  under the parametrization to,
\begin{equation} \label{SYMBOLFULL}  \sigma_{E_{t,s}}|_{\ccal_\pm} = \mp \frac{\rmi}{2} (2 \pi) |\dd t|^{\half} \otimes
 |\dd s|^{\half} \otimes |\dd w|^{\half} \times |\dd  \mu_{\mathrm{Char}\wt \Box}|^{\half}. \end{equation}
 \end{lem}

\begin{proof} This follows by the composition calculus using that  $\wt {E}_{t,s} = U(t,s) \wt{E}$. Composing with $\wt E$ is transversal, so $U \circ \wt{E}: = \wt{E}_{t,s}$ is also a Fourier integral operator, whose canonical relation is the composition of the Hamiltonian flows of the three principal symbols in \eqref{SYMBOLS}.   The order of the composition is given by the sum rule in Section \ref{FIOSECT} with $e = 0$, using Theorem \ref{DHTH} and 
Lemma \ref{ULEM}.

\end{proof}

%\begin{rem} Since the components $\zeta$ in $ \ccal$ are null vectors, it is useful to reduce the canonical relation to the
 %{\it null moment Lagrangian},
%\begin{equation} \label{NULLML}  \Gamma_{\wt \ncal}  = \{(t, p_Z([\zeta]) , s, p_\theta([\zeta]), [\zeta], -(f_s \circ g_t) ([\zeta]), [\zeta] \in \wt \ncal \}, %\end{equation}
%where $[\zeta]$ is the equivalent class of $\zeta$ and $\zeta \in \mathrm{Char}(\wt \Box)$.
%\end{rem}

Next, we consider  the  
canonical relation and principal symbol of the operator
$$R_{\Sigma} \circ d_{\tilde x}(\wt E_{t,s}(\tilde x,\tilde y) + \wt E_{-t, -s}(\tilde y, \tilde x)).$$

\begin{lem} \label{CREts} Let $\iota: \wt \ncal \to \wt \Sigma$ be the map \eqref{iota}. The order of  $R_{\wt{\Sigma}} \circ \Delta^*  \circ \dd_{\tilde x} \wt E_{t,s}(\tilde x,\tilde y)$ is $-1 + \frac{n+2}{4} $ and its  canonical relation is given by,
$$\Lambda =  \{(t, \tau, s, \sigma, \iota e^{s \frac{\partial}{\partial \theta}}  \ee^{t Z} \iota^{-1} \eta -\eta)) \mid \eta \in T^* \wt{\Sigma} \setminus 0, \tau = p_Z(\eta),
\sigma = p_\theta(\eta)\}. $$
\end{lem}

\begin{proof}  Except for the fact that we have two flows, $e^{t Z}, e^{s \frac{\partial}{\partial \theta}}$, the proof is essentially the same as for \cite[Lemma 8.4]{SZ18}.
Therefore, we only sketch the proof and concentrate on the additional features.
$\Lambda$ is the composition $\Lambda_{\wt{\Sigma}} \circ \ccal$, namely,
$$\begin{array}{l}  \{(\zeta |_{T \wt{\Sigma}}, \zeta) \in \dot{T}^* \Sigma \times \dot{T}^*_{\wt{\Sigma}} \wt{M}
\}
\circ  \{(t, \tau, s, \sigma,  \zeta_1, \zeta_2) \in \dot{T}^*(\R \times \wt{M} \times \wt{M}) \mid \tau = \zeta_1(Z) , (e^{s \frac{\partial}{\partial \theta}} e^{t Z}(\zeta_1),  \zeta_2) \in C\} \\ \\
= \{(t, \tau, s, \sigma, (e^{s \frac{\partial}{\partial \theta}} e^{t Z}\zeta_1) |_{T  \wt{\Sigma}}, \zeta_2) \mid \tau = \zeta_1(Z),  \sigma = p_\theta (\zeta_1), (\zeta_1,\zeta_2) \in C \cap T^*_{\wt{\Sigma }} \wt{M} \times T^*\wt{M} \}  \end{array}$$
We  pull back to the diagonal using the definitions in Section \ref{FIOSECT} to determine the canonical relation. Regarding the order, we use
the statement in Section \ref{FIOSECT} that under pullback the order increases by $\frac{k}{4}$ where $k $ is the codimension. The composite
pullback to the diagonal of $\wt{M}$ (increasing the order by $\frac{n +1}{4}$ and then by restriction to $\wt \Sigma$ (increasing the order by $\frac{1}{4}$)
implies the order statement in the Lemma.
\end{proof}

Finally, we consider the pushforward defined  by  integration over $\wt{\Sigma}$. 
\begin{lem} \label{piRDELTAE} The order of $\pi_* (R_{\Sigma} \Delta^* \dd_{\tilde x} \wt{E}_{t,s}(\tilde x,\tilde y) ))$ is $-1 + \frac{n+2}{4} + 1 - \frac{n}{4}  + \frac{2n-3}{2} = n-1$ and its canonical relation  is given by,
\begin{equation} \label{PFCR}\begin{array}{lll}\mathrm{WF} (\pi_* (R_{\Sigma} \Delta^* \dd_{\tilde x} E_{t,s}(\tilde x,\tilde y) )) & \subseteq &
 \{(t, \tau, s, \sigma) \mid \exists \gamma, e^{s \frac{\partial}{\partial \theta}}   \ee^{t Z} \gamma = \gamma \in \wt{\ncal}, \tau =  p_Z(\gamma), \sigma = p_\theta(\gamma)\}.
 \end{array} \end{equation}
 \end{lem}
 \begin{proof} The canonical relation is the pushforward of the one in the previous Lemma. As in the calculation of the order of the pushforward
 in Section \ref{FIOSECT}, the order of the pushforward is  $1 - \frac{\dim \wt \Sigma}{4}$ (since we have two time variables). The composition
 is not transversal, but has an excess $e = 2n-2 $ equal to the dimension of the fixed point set, which is all of $S \wt \ncal$ and has dimension $2n-3$.	
 
 \end{proof}

 This proves Proposition \ref{TrUts}.
\end{proof}
 
 \subsubsection{Principal symbol of $\Tr_{\hcal_{KG}} U(t,s)$ at $(t,s) = (0,0)$}
 
Proposition \ref{TrUts} implies      the first statement of Theorem \ref{LADDERSING1}, namely that the pair of singular times are a subset of the periods of
 $(s, t) \to e^{s \frac{\partial}{\partial \theta}}   \ee^{t Z}e^{t Z}$ acting on $\wt{\ncal}.$ The calculation of the principal symbol of the trace at $(s,t) = (0,0)$
 is essentially the same as in  \cite[Section 8.2.1]{SZ18}. 

 The principal symbol is   a homogeneous half-density on $T^*_{(0,0)} \R^2$ and therefore may be represented as a homogeneous 
 multiple of $|\dd \tau \wedge \dd \sigma|^{\half}$. The coefficient density is the volume of the fiber $\pcal_{\ncal}^{-1}(\tau, \sigma) \subset \wt \ncal,$ where
 $ \pcal_{\ncal} = (p_Z, p_\theta)$  with respect to the natural `Leray form' $\frac{\Omega_{\ncal}}{\dd p_Z \wedge \dd p_\theta}$ (see e.g. \cite{CdV79}).  By definition,
 this is the Liouville measure on the codimension-two fiber. We denote it by $\mu_{L} (\pcal_{\ncal}^{-1}(s, \tau).$ Then, 
 
 \begin{cor} \label{SYMBOLUCOR} The order of  $\Tr_{\hcal_{KG}} U(t,s)$ at $(t,s) = (0,0)$ equals $n-1$  and the principal symbol at $(s,t) = (0,0)$
 is the half density on $T^*_{(0,0)} \R^2$  given by
 $$C_n \mu_{\pcal_{\ncal}^{-1}(s, \tau)} |\dd \tau \wedge \dd \sigma|^{\half}. $$
 \end{cor}
 
 \begin{proof}  By Lemma \ref{DHLEM} and Lemma \ref{SYMBOLFULL},  the operator with kernel $\dd_{\tilde x} E_{t,s}(\tilde x,\tilde y)$ has order $-1$ and principal symbol equal to $\frac{1}{2} (2 \pi) (\xi \wedge) |\dd t|^{\half} \otimes  |\dd s|^{\half} \otimes |\dd_C|^{\half}$ on each component  $C_+$ and $C_-$, where $(\xi \wedge)$ denotes the operator of exterior multiplication by $\xi$. Restriction to $\R \times \wt{\Sigma} \times  \wt{\Sigma}$ gives an operator of order $-\frac{1}{2}$. %with principal symbol
% $(2 \pi)^{\frac{1}{2}} |\dd t|^{\half} \otimes |\dd s|^{\half} \otimes  |\dd V_{T^* \wt{\Sigma}}|^{\half}$, where we have used the natural parametrisation of the canonical %relation on $\R \times \wt{\Sigma}\times \wt{\Sigma}$.

  Restriction to the diagonal and integration over $\wt{\Sigma}$ gives an element in $I^{(n- 1)}(\R^2)$ with principal symbol $$ C_n \mu_{\pcal_{\ncal}^{-1}(s, \tau)} |\dd \tau \wedge \dd \sigma|^{\half}  $$ at $s= t=0$, where $C_n$ is a dimensional constant.   Note that $\dim \wt \ncal = 2(n+1) -2$ and the fiber has dimension $2n - 2$, so $ \mu_{\pcal_{\ncal}^{-1}(s, \tau)} $ is homogeneous of degree $n-1$.
  % \simeq C (\tau^2 + \sigma^2)^{\frac{n-1}{2}} $.
 \end{proof}

\subsection{Completion of the proof of Theorem \ref{LADDERSING1} by the second approach}

To complete the proof of Theorem \ref{LADDERSING1}, we further need to   change variables with $\alpha$ in  $\Tr_{\hcal_{KG}} U(t,s)$, to multiply by $\hat{\psi}$ and then integrate in $s$, i.e. we need to find the wave front set and principal symbol of 
\begin{equation} \label{LAST} \int_{\R} \hat{\psi}(s) \left(\Tr_{\hcal_{KG}} U(s, t - \nu s)\right) \dd s. \end{equation}
Note that  $U(s, t - \nu s) = e^{\rmi (t - \nu s) D_\theta} e^{s D_Z}$  \eqref{U(t,s)}. As above, let $\alpha(t, s) = (s, t - \nu s).$

 The change of variables results in composing the canonical relation of Lemma \ref{piRDELTAE}  with the graph of the
canonical transformation $(t,s, \tau, \sigma) \in T^*\R^2 \to (\alpha(t,s), (d \alpha)^{*-1} (\tau, \sigma))$, or equivalently, pulling it back under this map.
%\edit{  Write $\alpha(t, s) = (t', s'), \alpha^{-1 *}(\tau, \sigma) = (\tau', \sigma')$, i.e. $t' = t - \nu s, s' = s; \tau' = \tau, \sigma' = \nu \tau + \sigma). $}
Thus, the  canonical relation of $\alpha^* \pi_* (R_{\Sigma} \Delta^* \nu_{\tilde x} E_{t,s}(\tilde x,\tilde y) ))$ is given by,
\begin{equation} \label{PFalpha}\mathrm{WF} (\alpha^* \pi_* (R_{\Sigma} \Delta^* \nu_{\tilde x} E_{t,s}(\tilde x,\tilde y) ))  \subseteq  \begin{array}{l} 
 \{ (\alpha(t,s), - (\dd \alpha)^{*-1} (\tau, \sigma))\mid \exists \gamma, e^{s \frac{\partial}{\partial \theta}}   \ee^{t Z} \gamma = \gamma \in \wt{\ncal}, \\ \\ \tau =  p_Z(\gamma), \sigma = p_\theta(\gamma)\}. \end{array} 
  \end{equation}

We then pushforward under $\rho(s,t) = t$ to get,
\begin{equation} \label{PFalpharho}\mathrm{WF} (\rho_* \alpha^* \pi_* (R_{\Sigma} \Delta^* \nu_{\tilde x} E_{t,s}(\tilde x,\tilde y) ))  \subseteq  \begin{array}{l} 
 \{ (t, - \tau) \mid \exists (s, 0): \\ \\ (\alpha(t,s), (\dd \alpha)^{*-1} (\tau, \sigma)) \in \mathrm{WF} (\alpha^* \pi_* (R_{\Sigma} \Delta^* \nu_{\tilde x} E_{t,s}(\tilde x,\tilde y) )) \}. \end{array} 
  \end{equation}
  If $\alpha(t, s) =  (s, t - \nu s), $ then $\alpha^{-1 *}(\tau, \sigma) = (\nu \tau + \sigma, \tau)$, and  the  condition that 
  $$ (\alpha(t,s),  (\dd \alpha)^{*-1} (\tau, \sigma)) \in \mathrm{WF} (\alpha^* \pi_* (R_{\Sigma} \Delta^* \nu_{\tilde x} E_{t,s}(\tilde x,\tilde y) ))$$ is,
$$  \exists \gamma, e^{s \frac{\partial}{\partial \theta}}   \ee^{(t -  \nu s)  Z} \gamma = \gamma,   \;\; (p_Z(\gamma),  p_\theta(\gamma)) =(\tau, \sigma).$$
The condition that $(t, \tau, s, 0) \in \mathrm{WF} (\alpha^* \pi_* (R_{\Sigma} \Delta^* \nu_{\tilde x} E_{t,s}(\tilde x,\tilde y) ))$ further adds the condition that the WF point is  $(s, t- \nu s, \nu \tau, \tau) $. Note that $p_{\nu}(\nu \tau, \tau) = 0$. 
If we denote the equivalence class of $\gamma \in \wt \ncal$ with respect to the $S^1$ action by $[\gamma]$ the condition is that
$$\exp s H_{p_{\nu}} ([\gamma]) = [\gamma], \;\; p_{\nu}[\gamma] = 0, $$
as claimed in Theorem \ref{LADDERSING1}.

   \subsection{\label{COMPGUSECT} Comparison to \cite{GU89}}

The distribution $\Upsilon^{(1)}(t)$ is the analogue of  
$$\Upsilon(s) = \sum_{j=0}^{\infty} \phi(k^{\perp} \cdot \underline{\lambda}_j) e^{\rmi s \underline{\lambda}_j \cdot \ell}. $$ \cite[(4.2)]{GU89}. In \cite[Theorem 2.8]{GU89} they calculate its order and principal symbol at $s= 0$. 
Let
$$\fcal = \{(s, t, \zeta) \in \R \times \R \times T^*\wt{M} \mid \zeta \in W, f_s g_t(\zeta) = \zeta\}$$
be the fixed point variety  \cite[(2.5)]{GU89}, where $W = X_1 \cap X_2$
is a codimension two coisotropic submanifold, with $X_1 = p^{-1}(k_1), X_2 = 	q^{-1} (k_2), $
where $p = \sigma_P, q = \sigma_Q$ and $\underline{k} = (k_1, k_2)$ generates the ladder (line) $\lcal$. $S: \fcal \to \R$ is \cite[(2.8)]{GU89}, the map, 
$$S (T_1, T_2, \zeta) =  s = \underline{T}  \cdot \underline{k}. $$
The dimension of $\fcal$ for a given $t$ is the excess of the composition in \cite[(3.10)]{GU89} and the order of $\Upsilon$ at a singularity $s$ 
is $\dim S^{-1}(s)/2 - 3/2$

At  $s = 0$ its order (or ``degree'')  is $\frac{d}{2} - 1$ where $d = \dim S^{-1}(0) \subset \fcal$. The half-density part of the  symbol is defined
as follows: first, define the 
 density $\alpha$ on $\fcal$ \cite[(4.11)]{GU89}  by
$$\alpha = |\dd t|^{\half} \otimes |\alpha_1|^{\half} \otimes |\dd \theta|^{\half} \otimes |\alpha_2|^{\half}, $$
where $\alpha_1$ is Liouville density on $\{p =1\}$, $\alpha_2$ is Liouville density on $\{q = 1\}$.  The principal symbol
of $\Upsilon$ by the integral over $S^{-1}(0)$  of the density
$$2 \pi \hat{\psi}(t) \alpha$$
times a Maslov and sub-principal phase which is zero in our case. 

As mentioned in the introduction, our result on $\Upsilon^{(1)}$ is not simply the Lorentzian analogue of theirs because $\hcal_{KG}$ is
not a subspace of $L^2(\wt{M})$. But, at least intuitively, we may relate our result to theirs as follows: our $\Upsilon^{(1)}$, as in \cite[(4.2)]{GU89}, 
corresponds to the case $\ell_1 = 0, \ell_2 = 1$, $k_1 = \nu, k_2 = 1$.  
The   cotangent bundle $T^* M$ of \cite{GU89} should  be replaced by $\mathrm{Char}(\wt{\Box})$ (or, more precisely, by its quotient $\wt \ncal$ by the geodesic flow.)  
The  co-isotropic cone of the present article  is  $\sigma(\sigma_{D_Z} - \nu \sigma_{D_\theta}) =0$ or $\langle \xi, Z \rangle = \nu \langle \xi, \partial_\theta\rangle$. For
us, $\xi = (\xi', \sigma)$ and the equation is $\langle Z, \xi' \rangle = \nu \sigma . $ The base of the cone may be fixed by setting $\sigma =1$ 
and from the equation, our $W$ is  $\{\tau = \nu, \sigma =1\}$. Thus,  our $\ell$ is $\R_+ (\nu, 1)$.

 \section{\label{LADDERCORSECT} Proof of Theorem  \ref{LADDERCOR}}
In this section we study the asymptotics of the smooth Weyl functions \eqref{Nnupsi}.

\begin{equation} \label{Nnupsi2} N_{\nu, \psi} (m) : =
  \sum_{j \in \Z} \psi(\lambda_j (m) - \nu m). \end{equation}
  We determine the asymptotics as $m \to \infty$ by studying the singularity of \eqref{UPSILONHARDY} 
  at $s =0$ and applying a Hardy Tauberian theorem.

         %   \subsection{Proof of Corollary  \ref{LADDERCOR} for  $ \dd N^{(1)}_{\nu, \psi} (x)$ }
      
    Thus, we study the Hardy distribution,  \begin{equation} \dd N^{(1)}_{\nu, \psi} (x)
= \sum_{j \in \Z, m \geq 0} \psi(\lambda_j (m) - \nu m) \delta(x - m),\end{equation} Its Fourier transform is
the Hardy space distribution,
\begin{equation} \label{Upsilonb}  \Upsilon^{(1)}_{+,\nu, \psi}(s) : = \sum_{m \in \Z_+, j \in \Z} \psi(\lambda_j(m) - m \nu) 
e^{\rmi m s}, \end{equation}
which may be regarded as a Hardy Fourier series on $S^1$. We now give background on Hardy Fourier
series.

 The space of Hardy distributions $H^+(S^1)$
is the space of distributions with only positive Fourier
coefficients. Such distributions are boundary values of holomorphic functions in the upper half-plane. Hardy Lagrangian distributions are sums of Hardy homogeneous distributions, for which the 
Fourier coefficients are homogeneous in $m$. We define the   space $H^+_r (S^1)$ of homogeneous Hardy distributions  of
degree $s$ to be the  one-dimensional space   spanned by
\begin{equation} \label{NUJ} \begin{array}{l} \nu_{r} (t) =   \sum_{m=1}^{\infty} m^{r} e^{\rmi m t}.
\end{array}
\end{equation} It  has  singularities only at $t \in 2 \pi \Z$.
%It is convenient to introduce a slightly different entire family
%of distributions,
%\begin{equation} \label{MUJ}  \begin{array}{l} \mu_s(t) = \sum_{N=1}^{\infty} \n^{s} e^{\rmi t \n}
%\end{array} \end{equation}
%The distributions $\nu_j(t)$ are periodic while $\mu_j(t + 2 \pi)
%= e^{\rmi \pi \beta} \mu_j(t) $. We will only be concerned with these
%distributions for values $s \in \Z.$

\begin{prop}  \label{EITP} Assume that the fixed point sets of the flow of $e^{t Z}$ on \eqref{QUOT} are clean. Then,
there  exist coefficients $A_k, A_{j k}$
such that: $$ \begin{array}{lll}  \Upsilon^{(1)}_{+,\nu, \psi}(s) &
\sim & \sum_{k = 0}^{\infty} A_{k} \nu_{n - 2- k} (s ) \\ & & \\
& + & \sum_{j: \tau_j\not= 0}  \sum_{k = 0 }^{\infty}
A_{jk}\nu_{n_j - 2- k } (s- \tau_j);
\end{array}$$

\end{prop}

\begin{proof} It is clear that  $ \Upsilon^{(1)}_{+,\nu, \psi}(s)$ is periodic
and by definition has only positive
frequencies. Moreover, it is a  Lagrangian distribution, i.e. a sum of homogeneous
distributions near $0$.  Since the distributions (\ref{NUJ}) span
the periodic Lagrangian Hardy distributions, it follows that
 it may be expressed as a sum of these distributions. 
 We only need to know the  order of
 the terms, and these are given in Theorem \ref{LADDERSING1}.
\end{proof}

To complete the proof of Theorem \ref{LADDERCOR}, we combine Proposition \ref{EITP}
with \eqref{NUJ}. We are only interested in the singularity at $s =0$ so we ignore the second term, which
has homogeneous singularities away from $s =0$. We thus have,
$$\sum_{m \in \Z_+, j \in \Z} \psi(\lambda_j(m) - m \nu) 
e^{\rmi m s} \simeq \sum_{m \geq 0, k = 0}^{\infty}  A_{k} m^{n - 2- k} e^{\rmi ms} $$
where $\simeq$ means that the difference of the two sides is smooth in a neighborhood of $s =0$. 
By matching Hardy Fourier expansions, it follows that (cf. \eqref{Nnupsi})
\begin{gather*}
 N_{\nu, \psi}(m) =  \sum_{ j \in \Z} \psi(\lambda_j(m) - m \nu) \simeq   \sum_{k = 0}^{\infty}  A_{k} m^{n - 2- k},\\
\end{gather*}
where $\simeq$ means here that the two sides differ by a rapidly decaying function of $m$. 
Comparing with \eqref{firstcoeff} shows that
$$
 A_0 = a_0(\nu, \psi) = (2 \pi)^{-n+1} \hat{\psi}(0) \mu_{L} (\ncal_{1}(\nu)). 
$$
This is
precisely the statement of Theorem \ref{LADDERCOR} .

%Define
%$$\sigma_m(x) = \int_{-\infty}^x  \dd N^{(2)} _{\nu, \psi} (u)
%= \sum_{j, m: |\lambda_j(m) |\leq x} \psi(\lambda_j (m) - \nu m). $$
%%\edit{Put in absolute value to make it one -sided}
%We assume that $\psi \geq 0$ so that the sum is monotone increasing. 

  %    \subsection{Theorem \ref{LADDERCOR} for $d N^{(2)} _{\nu, \psi} (x)
%$  }
 %     Secondly, we prove Theorem \ref{WEYLCOR} for 
%%        \begin{equation} \dd N^{(2)} _{\nu, \psi} (x)
%= \sum_{j, m \in \Z} \psi(\lambda_j (m) - \nu m) \delta(x - \lambda_j(m)).\end{equation} 
%This is more complicated because the distribution is not periodic, much less Hardy. 

\section{\label{WEYLSECT} Proof of Theorem \ref{WEYLCOR}}

In this section we prove Theorem \ref{WEYLCOR} for the sharp Weyl functions \eqref{NnuC}, 
\begin{equation} \label{NnuC2} N_{\nu, C}(m) : = \# \{j \mid \frac{\lambda_j(m)}{m}  \in [\nu - \frac{C}{m}, \nu + \frac{C}{m} ]\},  \end{equation}
where $C > 0$ is a given constant. The main point is to use $\psi$ in Theorem \ref{LADDERCOR} { to approximate an} indicator function. 

First we will need a proof of Lemma \ref{LADDERCORMORE}. We now need to take into account the other periods $\tau_j$ in Proposition \ref{EITP}. { Under a cleanliness assumption we can drop the assumption on the support of $\hat \psi$ and conclude that $\Upsilon_1(s)$ is a Fourier-Hardy distribution with isolated conormal singularities. A non-degeneracy condition then implies that the dimensions $n_j$ of the fixed point set at periods other than $0$ are strictly smaller than $n$.}
Hence, the additional terms for non-zero periods  contribute
$O(m^{n-3})$, under the assumption of cleanliness and non-degeneracy. { This assumption is however not necessary and can be replaced by the stated zero-measure condition. Indeed, in this case the remainder is of order $o(m^{n-2})$. }This
follows from the now-standard Duistermaat-Guillemin-Ivrii argument as in \cite[Theorem 29.1.5]{HoIV}. The argument is to introduce microlocal
cutoffs $b$ to an open set of volume $O(\epsilon)$ around the set of points on periodic orbits of period $\leq T$, and break up the trace using $I = B + b$ where $B = I -b$. { In our setting this cut-off is introduced in the integral formula for the trace in Theorem \ref{UphiFORMa}. Both $B$ and $b$ are pseudo-differential operators on $\wt \Sigma$.}
The Tauberian theorem gives `credit'
$O(\epsilon)$ for the volume of the micro-support of $B$, and $O(\frac{1}{T})$ for the length of the interval around $0$ where the trace composed with $b$
is free of singularities for $t \not=0$. Since $\epsilon, T$ are arbitrary, one obtains the remainder $o(m^{n-2})$.
We then complete the proof using the argument of \cite{DG75, BrU91}. 

\begin{lem} \label{Lemma3.3} For $\nu$ admissible, $C > 0$,  there exists a constant $k= k(\nu, C)$ so that 
$$\# \{j \mid \; |\lambda_j(m) - m \nu | \leq C \} \leq k  \;m^{n-2}. $$
\end{lem}
This is an immediate corollary of Lemma \ref{LADDERCORMORE}, using a test function satisfying  $\psi \geq 1$ on the indicated set. 

 { Now we choose a non-negative even smooth function $\psi$ with non-negative compactly supported Fourier transform $\hat \psi$ and $\hat \psi(0)=1$.
For $\delta > 0$ define the scaled}  $\psi_{\delta}(x) = \delta^{-1} \psi(\delta^{-1} x)$. Then define,
$$\chi_{C, \delta} : = \psi_{\delta} * {\bf 1}_{[-C,C]}. $$

\begin{lem} \label{Lemma3.4} For any $0 < \gamma < C$ and $N$ sufficiently large, there exist constants $k_N, K_N$ such that
$$N_{\nu, C - \gamma}(m) (1 - k_N \; (\frac{\delta}{\gamma})^N) \leq \sum_{j \geq 1} \chi_{C,\delta} (\lambda_j(m) - \nu m) \leq 
N_{\nu, C + \gamma}(m)  + K_N \; (\frac{\delta}{\gamma})^N m^{n-1}. $$
\end{lem} 
The proof is essentially identical to that of \cite[Lemma 3.3]{DG75} and  \cite[Lemma 3.4]{BrU91}, and is omitted.

To complete the proof of Theorem \ref{WEYLCOR}, we use Lemma \ref{Lemma3.3}-Lemma \ref{Lemma3.4} to obtain,
$$\left\{ \begin{array}{l} N_{\nu,  C - \gamma } (m) (1 - k_N (\delta/\gamma)^N) \leq (2 \pi)^{-n+1} 2 C V m^{n-2} + o_{\delta, \gamma} (m^{n-2}), \\ \\ 
 N_{\nu, C + \gamma } (m)  + K_N (\delta/\gamma)^N m^{n-2}  \geq (2 \pi)^{-n+1} 2 C V m^{n-2} + o_{\delta, \gamma} (m^{n-2}), \end{array} \right. $$
 where $V $ is the Liouville volume in Theorem \ref{WEYLCOR}. Dividing by $m^{n-2}$ and letting $m \to \infty $ first and then $\delta \to 0$ one gets,
 $$\limsup_{m \to \infty} m^{-(n-2)} N_{\nu, C - \gamma}(m)  \leq  (2 \pi)^{-n+1} 2 C V \leq \liminf_{m \to \infty} m^{-(n-2)} N_{\nu, C + \gamma}(m). $$
 Letting $C \to C + \gamma$ and then $C \to C - \gamma$ and letting $\gamma \to 0$ concludes the proof.

            \section{\label{PRODUCT} Product ultra-static spacetimes} 
            
            As mentioned in the Introduction, the `ladder' asymptotics of this article are non-standard
            even in the case of product spacetimes.        To clarify the symplectic geometry in the simplest setting, we briefly run through the ladder theory of this
       article in the case of product spacetimes $ \R \times \Sigma$ where $\Sigma$ is compact. In particular, we aim
       to clarify the relation between the massive geodesic flow on the unit mass hyperboloid and the space $\wt{\ncal}_{\nu}$ \eqref{QUOT}. 

We first identity the 
       function $N_{\nu, \psi}(m)$ of Theorem \ref{LADDERCOR}  and of $N_{\nu, C}(m)$ of Theorem \ref{WEYLCOR}  in the case of 
      product (ultra-static) spacetimes $(M, g) = (\R \times \Sigma, -\dd t^2 + h)$ , where
            $(\Sigma, h)$ is a compact Riemannian manifold of dimension $n-1$. Of course, the results in this case  are the well-known Weyl asymptotics for Laplacians
            of compact Riemannian manifolds, but the calculations illuminate the nature of the asymptotics in the general stationary setting
           and given simple tests of the numerology. 
           
           Let $\Delta = \Delta_h$ denote the Laplacian of $(\Sigma, h)$ and let $(\phi_j, - \omega_j^2)$ denote its spectral data, i.e. 
           $(\Delta + \omega_j^2) \phi_j = 0$. 
           The solutions of the system \eqref{BOXZ} obviously have the form, 
            $u_j(t, x) = e^{\rmi t \lambda_j} \phi_j(x)$  where $\lambda_j^2 = m^2  + \omega_j^2$. Thus, $\lambda_j(m) = \pm \sqrt{m^2 + \omega_j^2}.$  We also have rather explicitly $Q(u_j, u_j) = \lambda_j^2 \| u_j \|_{L^2(\Sigma)}^2$ which shows directly that the topology on the Cauchy data space induced by $Q$ is $H^1(\Sigma) \oplus L^2(\Sigma)$.          
            The quantum  ladder 
 $\lcal_{\nu}$ is defined, roughly speaking,  by $\frac{\lambda_j(m)}{m} \simeq \nu$ or equivalently, $\omega_j \simeq m \sqrt{\nu^2-1}$.

            We only count the positive eigenvalues of $D_Z$, where we have,
            $$N_{\nu, C}(m) = \#\{j \mid   \sqrt{m^2 + \omega_j^2} \in [\nu m - C, \nu m + C]\}. $$
             The condition is equivalent to $(\nu m - C)^2 - m^2 \leq \omega_j^2 \leq (\nu m + C)^2 - m^2$ and has no solutions unless 
             $ (\nu - 1) m + C \geq 0$. For large $m$, there are no solutions unless $\nu \geq 1$. In that case, we have,
             $$N_{\nu, C}(m) = \# \{j \mid \omega_j \in m \; [\sqrt{(\nu - \frac{C}{m})^2 - 1}, \sqrt{(\nu + \frac{C}{m})^2 -1} ]\}. $$
             The interval is centered at the point $m \sqrt{\nu^2 -1}$ and has diameter $\simeq 2 C \frac{\nu}{\sqrt{\nu^2 -1}}. $
             Hence, 
                  $$N_{\nu, C}(m) \simeq  \#\{j \mid  \omega_j \in [m \sqrt{\nu^2 -1}  -  C \frac{\nu}{\sqrt{\nu^2 -1}} , m \sqrt{\nu^2 -1}  +  C \frac{\nu}{\sqrt{\nu^2 -1}}]\}. $$
                  The asymptotics depend on the periodicity properties of the geodesic flow $G^t_h$ of $(\Sigma, h)$. For instance, if $\Sigma_h = S^3$ is
                  a standard $S^3$, then the eigenvalues of $\sqrt{-\Delta}$ occur only at the points $\sqrt{(N + 1)^2 - 1} \simeq N +1$ and have multiplicity $N^2$. 
                  On the other hand, if the set of closed geodesics of $(\Sigma, h)$ has Liouville measure zero, then it follows from the two-term Weyl  law with remainder estimates
                  of \cite{DG75} that 
                  \begin{equation} \label{PRODCASE} N_{\nu, C}(m) \simeq \alpha_{n-1} \mathrm{Vol}(\Sigma, h) (m \sqrt{\nu^2-1})^{n-2} [2 C \frac{\nu}{\sqrt{\nu^2 -1}}], \;\; n -1 = \dim \Sigma.  \end{equation}
                  Here,  $\alpha_{n-1}$ is $(2 \pi)^{-n+1}$ times the Euclidean surface measure  of the unit sphere in $\R^{n-1}$.
                  Although the calculation pertains to a special and well-known case, the order of growth in $m$ and the geometric
                  coefficient are universal and corroborate the
       details in  Theorem \ref{WEYLCOR}.

         The smoothed counting function $N_{\nu, \psi}(m)$ has a complete asymptotic expansion in $m$ irrespective of the periodicity properties
         of the geodesic flow. We can see this by studying the simpler Weyl function,
         $$N^+_{\nu, \psi}(m) = \sum_{j=1}^{\infty} \psi(\sqrt{m^2 + \omega_j^2} - \nu m) = \Tr   \psi(\sqrt{m^2 - \Delta_h} - \nu m),$$   
         using that 
         $$\psi(\sqrt{m^2 - \Delta_h} - \nu m) = \int_{\R} \hat{\psi}(t) e^{\rmi tm \left( \sqrt{1 - m^{-2}\Delta_h} - \nu \right)} \dd t. $$
         The operator $\sqrt{1 - m^{-2}\Delta_h} $ is a positive elliptic semi-classical pseudo-differential operator of order zero, and one may construct
         a semi-classical parametrix for it and apply stationary phase. 
         
         The leading  singularity at $s = 0$ of $ \Upsilon^{(1)}_{\nu, \psi}(s) : = \sum_{m, j \in \Z} \psi(\lambda_j(m) - m \nu) 
e^{\rmi m s}$ is essentially the same as for the flat case of $\wt M =\R_t \times \R^{n-1}_x \times \R_\theta$, as long as we do not integrate in the space variable $(x,\theta) \in \R^{n} \times S^1$;
for expository simplicity we only give the details in this model case. It is of interest because the leading coefficient in Theorem \ref{LADDERSING1}
is universal, and the details of the coefficients and order of singularity may be checked in this explicitly computable case. In fact, instead of homogenizing
and summing in $m$, we work on $\Sigma = \R^{n-1}$ and  compute the semi-classical asymptotics as 
$m \to \infty$ of \begin{equation} \label{m} \sum_{ j \in \Z} \psi(\lambda_j(m) - m \nu) = \int_{\R} \hat{\psi}(t) 
\Tr e^{\rmi t\sqrt{m^2 - \Delta_{\R^{n-1}}}} 
e^{- \rmi t m \nu} \dd t \end{equation} directly from the semi-classical kernel,
$$ e^{\rmi t\sqrt{m^2 - \Delta}}(x, x)  = \int_{\R^{n-1}} e^{\rmi t \sqrt{m^2 + |\xi|^2}} \dd \xi.$$
(We repeat that we do not integrate in $x$, so the integrals are well-defined distributions with the same singularities as in the compact spacelike case).
Then 
$$\begin{array}{lll}  \eqref{m} \simeq \int_{\R}   \int_{\R^{n-1}}  \hat{\psi} (t) e^{\rmi t \sqrt{m^2 + |\xi|^2}} e^{- i t m \nu} \dd t \dd \xi & = &  m^{n-1} \int_{\R}   \int_{\R^{n-1}}  \hat{\psi} (t) e^{\rmi t m \sqrt{1 + |\xi|^2}} e^{- i t m \nu} \dd t \dd \xi \\ &&\\ & = &\alpha_{n -1} 
m^{n-1} \int_{\R}   \int_0^{\infty} \hat{\psi} (t) e^{\rmi m \left(t  \sqrt{1 + \rho^2} - t  \nu \right)} \rho^{n-2} d\rho \dd t  \end{array}. $$
Here $\xi = \rho \omega$ in polar coordinates and $\alpha_{n-1}$ is the surface area of the sphere in $\R^{n-1}$.
The phase is stationary when (and only when)  $\sqrt{1 + \rho^2} = \nu, t = 0$ if $\supp \; \hat{\psi}$ is sufficiently small.  Stationary
phase then gives the expansion,
$$ \eqref{m} \simeq \alpha_{n-1} \;  \hat{\psi}(0) (\nu^2 -1)^{\frac{n-2}{2}} \frac{\nu}{\sqrt{\nu^2 -1}}  m^{n-2}, $$
consistently with the asymptotics in Theorem \ref{LADDERCOR}. It follows that the principal singularity at $t =0$ of \eqref{Upsilon} is given by,
$$\Upsilon^{(1)}(s) \sim  \alpha_{n-1} \;   (\nu^2 -1)^{\frac{n-2}{2}} \frac{\nu}{\sqrt{\nu^2 -1}} \hat{\psi}(0)   \sum_{m=0}^{\infty} m^{n-2} e^{\rmi m s} =   \alpha_{n-1} \;   (\nu^2 -1)^{\frac{n-2}{2}} \frac{\nu}{\sqrt{\nu^2 -1}}  \hat{\psi}(0)  \; \mu_{n-2}(s).$$  For the trace in the compact case, we would integrate over $\Sigma$ and pick up the
extra factor of $\mathrm{Vol}_h(\Sigma)$. The result then agrees with 
Theorem \ref{LADDERSING1}, once we show that $\mu_{L} (\ncal_{1}(\nu))  = \alpha_{n-1}  (\nu^2 -1)^{\frac{n-2}{2}} \mathrm{Vol}_h(\Sigma)$ in this case  (see \eqref{LIOUVPROD} below.)

        If instead we follow the homogenization procedure of this article by letting $m \in \Z$ be eigenvalues of $S^1$ and 
         study ladders in $\wt{\Sigma} = \Sigma \times S^1$.
We then define $\wt{M} = \R_t \times \Sigma \times S^1$ and $$\wt{\Box} = \frac{\partial^2}{\partial t^2} - \Delta -\frac{\partial^2}{\partial \theta^2} , $$
where $\Delta = \Delta_h $. The null solutions of $\wt{\Box}$  then have the form,
$$u_{m, j}(t, x, \theta) = e^{\rmi m \theta} e^{\rmi t \sqrt{\omega_j^2 + m^2} } \phi_j(x). $$
We then consider
     $$\psi(\sqrt{ \Delta_{\wt{\Sigma}}} - \nu D_\theta) = \int_{\R} \hat{\psi}(t) e^{- i t \nu D_\theta} e^{\rmi t\sqrt{ \Delta_{\wt{\Sigma}}}} \dd t. $$
     $e^{\rmi t\sqrt{ \Delta_{\wt{\Sigma}}}} $ is the half-wave operator $\wt{U} (t, (x, \theta), (x', \theta'))$  of $\wt{\Sigma}$ and has the parametrix used in \cite{DG75}. 
     The operator $e^{- i t \nu D_\theta} $ is translation in $s$ by $t$ units. Then,
     $$ \Upsilon^{(1)}_{\nu, \psi}(s) = \Tr e^{\rmi s D_\theta} \psi(\sqrt{ \Delta_{\wt{\Sigma}}} - \nu D_\theta) 
     =  \int_{\R} \int_{S^1} \int_{\Sigma} \hat{\psi}(t)  \wt{U}(t \nu + s, (x,\theta), (x,\theta)) dV(x) \dd \theta \dd t. $$
     The singularity at $s = 0$ is determined in by the same Fourier integral techniques as in \cite{DG75,GU89} or as in previous sections. Note that
    in the product case,  $$\wt E (t,t') = (\wt{ \Delta})^{-\frac{1}{2}} \sin( (t-t') (\wt \Delta)^{\frac{1}{2}}).$$

We now identify the set of admissible $\nu$ in the sense of Definition \ref{ADMISSIBLEDEF}, the double characteristic variety $\mathrm{DChar}_{\nu}(s, \sigma)$ of
Definition \ref{DCHAR}, the allowed region $\acal_{\nu}$ where $\frac{\partial}{\partial t} - \nu \frac{\partial}{\partial \theta}$ is spacelike, 
 the   governing flow $e^{t Z}$  on $\ncal_1(\nu)$, and the volume of $\ncal_1(\nu)$. 

\begin{lem} We have:
\begin{enumerate}
\item The set of admissible $\nu$ is the set $\{\nu > 1\}$. 
\item For $\nu > 1$,  $\acal_{\nu} = \wt M$, i.e.  $\frac{\partial}{\partial t} - \nu \frac{\partial}{\partial \theta}$ is everywhere spacelike.
\item For $\nu > 1$, $\mathrm{Dchar}_{\nu}(s, \sigma) = \{\tau = \sigma \nu, \; |\xi| = \sigma  \sqrt{\nu^2 -1}\} .$
\item $\ncal_1(\nu) = \{\tau =  1,   \; |\xi| =  \sqrt{\nu^2 -1}\}/\sim$ (divide by the Hamiltonian flow.)
\item $a_0(\nu, \psi) = (2 \pi)^{-n+1}  \hat{\psi}(0) \mu_{L} (\ncal_{1}(\nu)) =  \alpha_{n-1} \left(\sqrt{\nu^2 -1} \right)^{n-2} \frac{\nu}{\sqrt{\nu^2 -1}}  \mathrm{vol}_h(\Sigma),$
\item The governing Hamiltonian flow is the geodesic flow of $\Sigma$ (for any $\nu$).

\end{enumerate}

\end{lem}
The characteristic variety of $\wt{\Box}$ is defined by
$$ \wt{\mathrm{Char}} 
%\{(t, x,  s;  \tau, \xi,  \sigma) \in T^*(\wt{M}): \tau^2 -  (|\xi|^2 + \sigma^2) = 0\} 
= \bigcup_{\sigma \in \R} \hcal_{\sigma}, $$
where $\hcal_{\sigma}$ is the ``mass $\sigma$ hyperboloid bundle' 
$$\hcal_{\sigma} = \{((t, x,  \theta;  \tau, \xi,  \sigma) \in T^*(\wt{M}) \mid \tau^2 -  (|\xi|^2 + \sigma^2) = 0\}. $$
Evidently, $\sigma$ is a homogeneous coordinate playing the role of $m$.  We then define $\wt{\ncal}$ by quotienting $\wt{\mathrm{Char}}$ by the
Hamiltonian flow of  $\sigma_{\wt{\Box}} = - \tau^2 + (|\xi|^2 + \sigma^2)$. Since the flow preserves each hyperboloid, $$ \wt{\ncal}  =  \bigcup_{\sigma \in \R} \hcal_{\sigma}/\sim. $$

In Section \ref{SYMPSECT} we defined a homogeneous  symplectic isomorphism $\iota: \wt \ncal \to T^* \wt \Sigma$  \eqref{iota}, which amounts
here to restricting the null covector $(\tau, \xi, \sigma)$ to $T_{t, x, \theta}  \wt \Sigma$. Since $\wt \Sigma$ is a level set of the time variable $t$, $d t = 0$ (and so $\tau = 0$)
on $T \wt \Sigma$,  and $\iota(t, x, \theta; \tau, \xi, \sigma) = (t, x, \theta; 0, \xi, \sigma).$

The Killing vector field in this context is $Z = \frac{\partial}{\partial t}$, and the principal symbol of $D_Z$ is
$\langle (\tau, \xi, \sigma), \frac{\partial}{\partial t}\rangle = \tau$. The ladder is non-empty if and only if $\nu \geq 1$; since the equations
to not involve the base variables $(t, x, \theta)$, $\acal_{\nu} = \wt M$ when $\nu \geq 1 $ and is empty otherwise. We henceforth
assume $\nu > 1$.

The classical ladder in $T^* \wt{M}$ (or in $\wt \ncal$)   is defined by $p_{\mu}(t, x, \theta; \tau, \xi, \sigma) = \tau - \nu \sigma = 0$, and is  given by
\begin{equation}\label{lcalnu} \lcal_{\nu} = \bigcup_{(\theta, \sigma) \in T^* S^1} \{\tau = \sigma \nu, \; |\xi| = \sigma \sqrt{\nu^2 -1}\} \times \{(\theta, \sigma)\}. \end{equation}
Under $\iota$ \eqref{iota}, the ladder is taken to
$$\iota(\lcal_{\nu}) = \bigcup_{(\theta, \sigma) \in T^* S^1}  \{ |\xi| = \sigma \sqrt{\nu^2 -1}\} \times \{(\theta, \sigma)\}  = \{(|\xi| =  \sigma \sqrt{\nu^2 -1}\}  \subset T^* \Sigma \times T^* S^1. $$ For fixed $(\theta, \sigma)$, the classical ladder is the  $\tau$-slice   of a mass hyperboloid (bundle), 
$$\{\tau = \sigma \nu, \; |\xi| = \sigma  \sqrt{\nu^2 -1}\} .$$

As above, we fix $\sigma =1$ to de-homogenize and then obtain the spherical slice, 
%The unit bundle $\sigma_Z(t, x,s,  \tau, \xi, \sigma) =1$
%in $\lcal_{\nu} $ is the simply the unit mass hyperboloid (bundle) 
$$\{\tau = \nu,  \; |\xi| =  \sqrt{\nu^2 -1}\} \subset \mathrm{Char}(\Box) \subset   T^* M.$$
%fixing the mass $\sigma = \frac{1}{\nu}$. \edit{Above, we set $\sigma =1$}
We then quotient by time translation. The  quotient  is $\ncal_1(\nu) $ \eqref{ncal1nu} 
$$\ncal_1(\nu) \simeq \{(x, \xi) \in T^* \Sigma \mid |\xi|= \sqrt{\nu^2 -1}\}. $$
The Liouville volume $\mu_L(\ncal_1(\nu))$ is by definition the usual Liouville volume of the right side, namely, \begin{equation} \label{LIOUVPROD} \mu_L(S^*_{\sqrt{\nu^2 -1}}\Sigma)
= \mathrm{Vol}(S_{n-2}) \left(\sqrt{\nu^2 -1} \right)^{n-2} \frac{\nu}{\sqrt{\nu^2 -1}}  \mathrm{Vol}_h(\Sigma). \end{equation} Indeed, the Liouville measure satisfies
$\mu_L(S^*_{\sqrt{\nu^2 -1}}\Sigma) = \frac{d}{d\nu} \mathrm{Vol}_{\Omega}\{|\xi| \leq \sqrt{\nu^2-1} \}$. 
Compared to \eqref{PRODCASE}, the additional factor of
$2 C $ comes from the width of the interval (i.e. the test function).
 
The governing  dynamical system underlying the ladder trace $\Tr e^{t Z} |_{{\mathfrak H}_{\nu}}$ is identified by $\iota$ with the null foliation
of the sphere bundle $\{|\xi| = \sqrt{\nu^2 -1}\} \subset T^* \Sigma$, i.e. with its geodesic flow. 
%  the flow of the Hamiltonian
%$H_{\nu}(x, s;  \xi, \sigma) = |\xi| - \sqrt{\nu^2 -1} \sigma$ on $\iota(\lcal_{\nu})$. This flow is the product flow $G^t_{\Sigma} \times e^{t \sqrt{\nu^2-1} %%\frac{\partial}{\partial \theta}}$ on $T^*\Sigma \times T^* S^1$ restricted to $\iota(\lcal_{\nu}) $. We further identify $\iota(\lcal_{\nu}) \simeq T^* \Sigma \times %S^1$
%under the map $(x, s; \xi, \sigma) \in \lcal_{\nu} \to (x, s; \xi)$. It is invertible under $\sigma = (\nu^2 -1)^{-\half}  |\xi| $. Under this identification, the Hamiltonian
It follows that the zero-measure condition on periodic orbits in Theorem \ref{WEYLCOR} is independent of $\nu$, and is the well-known condition
that the set of periodic orbits of $G^t$ has measure zero for $t \not=0$. 

%\begin{rem}
%In the semi-classical setting, the semi-classical characteristic variety of $m^{-2} \Box -1$ is the
%% (bundle of cotangent)  mass hyperboloid(s) $\tau^2 = |\xi|^2 + m^2$ in $T^*M$. In homogeneous ladder theory,
%this mass hyperboloid is replaced by a conic (homogeneous) hypersurface given by the cone \eqref{lcalnu}  over this mass hyperboloid. \end{rem}

%\section{\label{TAUBERAPP} Appendix on Tauberian theory}

%We review the proof of the Tauberian theorem in   \cite[Lemma 3.3]{DG75}, \cite[Theorem 3.2]{BrU91}.

\end{document}